\documentclass[letter,10pt]{article}
\usepackage[margin=0.8in]{geometry} 
\usepackage{amsfonts,amssymb,amsmath,amsthm}
\usepackage{graphicx,float,afterpage} 
\usepackage{enumerate,color,caption,tikz,boxedminipage}
\usepackage{changepage}
\usepackage[round]{natbib}
\usepackage{hyperref}
\hypersetup{colorlinks=true,linkcolor=red,citecolor=blue}

\usepackage{chngcntr,apptools,titlesec}
\AtAppendix{\counterwithin{lemma}{section}}
\titleformat{\section}{\normalfont\centering\fontsize{12}{13}\bfseries}{\thesection}{1em}{}
\titleformat{\subsection}{\normalfont\normalsize\bfseries}{\thesubsection}{1em}{}
\titleformat{\subsubsection}{\normalfont\normalsize\bfseries}{\thesubsubsection}{1em}{}
\renewcommand{\thesection}{\arabic{section}}
\renewcommand{\thesubsection}{\arabic{section}.\arabic{subsection}}

\usepackage[ruled,linesnumbered,algoruled,boxed]{algorithm2e}

\usepackage{array,multirow,ctable}

\newtheorem{lemma}{Lemma}
\newtheorem{theorem}{Theorem}
\newtheorem{proposition}{Proposition}
\newtheorem{corollary}{Corollary}
\theoremstyle{definition}
\newtheorem{definition}{Definition}
\theoremstyle{remark}
\newtheorem{remark}{Remark}

\DeclareMathOperator*{\argmin}{argmin}

\DeclareMathOperator{\HC}{\text{HC}}

\newcommand{\vertiii}[1]{{\left\vert\kern-0.25ex\left\vert\kern-0.25ex\left\vert #1 
		\right\vert\kern-0.25ex\right\vert\kern-0.25ex\right\vert}}
\newcommand{\vertii}[1]{{\left\vert\kern-0.25ex\left\vert #1 
		\right\vert\kern-0.25ex\right\vert}}
\newcommand{\smalliii}[1]{{\vert\kern-0.25ex\vert\kern-0.25ex\vert #1 \vert\kern-0.25ex\vert\kern-0.25ex\vert}}
\newcommand{\smallii}[1]{{\vert\kern-0.25ex\vert #1\vert\kern-0.25ex\vert}}
\newcommand{\vect}[1]{\mathop{\mathrm{vec}}(#1)}
\newcommand{\VAR}{\text{VAR}}
\newcommand{\R}{\mathbb{R}}
\newcommand{\E}{\mathbb{E}}
\newcommand{\Cov}{\text{Cov}}
\newcommand{\tr}{\text{trace}}

\newcommand{\Normal}{\mathcal{N}}
\newcommand{\rank}{\text{rank}}

\makeatletter
\newsavebox{\@brx}
\newcommand{\llangle}[1][]{\savebox{\@brx}{\(\m@th{#1\langle}\)}%
	\mathopen{\copy\@brx\mkern2mu\kern-0.9\wd\@brx\usebox{\@brx}}}
\newcommand{\rrangle}[1][]{\savebox{\@brx}{\(\m@th{#1\rangle}\)}%
	\mathclose{\copy\@brx\mkern2mu\kern-0.9\wd\@brx\usebox{\@brx}}}
\makeatother

\makeatletter
\renewcommand{\paragraph}{%
	\@startsection{paragraph}{4}%
	{\z@}{2.0ex \@plus 1ex \@minus .2ex}{-1em}%
	{\normalfont\normalsize\em}%
}
\makeatother

\title{\large Regularized Estimation and Testing \\\vspace*{1mm} for High-Dimensional Multi-Block Vector-Autoregressive Models}
\author{\normalsize Jiahe Lin\thanks{Ph.D Candidate, Department of Statistics, University of Michigan.}, \quad George Michailidis\thanks{Corresponding Author. Professor, Department of Statistics and the Informatics Institute, University of Florida, 205 Griffin-Floyd Hall, PO Box 118545, Gainesville, FL,  32611-8545, USA. Email: \texttt{gmichail@ufl.edu.}}}
\date{}

\begin{document}
	\maketitle
	
\begin{abstract}
Dynamical systems comprising of multiple components that can be partitioned into distinct {\em blocks} originate in many scientific areas. A pertinent example is the interactions between financial assets and selected macroeconomic indicators, which has been studied at {\em aggregate level} --e.g. a stock index and an employment index-- extensively in the macroeconomics literature. A key shortcoming of this approach is that it ignores potential influences from other related components (e.g. Gross Domestic Product) that may exert influence on the system's dynamics and structure and thus produces incorrect results. To mitigate this issue, we consider a multi-block linear dynamical system with Granger-causal ordering between blocks, wherein the blocks' temporal dynamics are described by vector autoregressive processes and are influenced by blocks higher in the system hierarchy. We derive the maximum likelihood estimator for the posited model for Gaussian data in the high-dimensional setting based on appropriate regularization schemes for the parameters of the block components. To optimize the underlying non-convex likelihood function, we develop an iterative algorithm with convergence guarantees. We establish theoretical properties of the maximum likelihood estimates, leveraging the decomposability of the regularizers and a careful analysis of the iterates. Finally, we develop testing procedures for the null hypothesis of whether a block ``Granger-causes" another block of variables. The performance of the model and the testing procedures are evaluated on synthetic data, and illustrated on a data set involving log-returns of the US S\&P100 component stocks and key macroeconomic variables for the 2001--16 period. 
\end{abstract}

\section{Introduction.}\label{sec:intro}

The study of linear dynamical systems has a long history in control theory \citep{Kumar1986Stochastic} and economics \citep{hansen2013recursive} due to their analytical tractability and ease to estimate their parameters. Such systems in their so-called {\em reduced form} give rise to Vector Autoregressive (VAR) models \citep{lutkepohl2005new} that have been widely used in macroeconomic modeling for policy analysis \citep{sims1980macroeconomics, sims1982policy, sims1992interpreting}, in financial econometrics \citep{gourieroux2001financial}, and more recently in functional genomics \citep{shojaie2012adaptive}, financial systemic risk analysis \citep{billio2012econometric} and neuroscience \citep{seth2013interoceptive}. 

In many applications, the components of the system under consideration can be naturally partitioned into {\em interacting blocks}. For example, \cite{cushman1997identifying} studied the impact of monetary policy in a small open economy, where the economy under consideration is modeled as one block, while variables in other (foreign) economies as the other. Both blocks have their own autoregressive structure, and the inter-dependence between blocks is unidirectional: the foreign block {\em  influences} the small open economy, but {\em not} the other way around, thus effectively introducing a {\em linear ordering} amongst blocks. Another example comes from the connection between the stock market and employment macroeconomic variables \citep{fitoussi2000roots, phelps1999behind,farmer2015stock} that focuses on the impact through a wealth effect mechanism of the former on the latter. Once again, the underlying hypothesis of interest is that the stock market influences employment, but not the other way around.
In another application domain, molecular biologists conduct time course experiments on cell lines or animal models and collect data across multiple molecular compartments (transcripotmics, proteomics, metabolomics, lipidomics) in order to delineate mechanisms for disease onset and progression or to study basic biological processes. In this case, the interactions amongst the blocks (molecular compartments) are clearly delineated (transciptomic compartment influencing the proteomic and metabolomic ones), thus leading again to a linear ordering of the blocks \citep[see][]{richardson2016statistical}. 

The proposed model also encompasses the popular in marketing, regional science and growth theory VAR-X model, provided that the temporal 
evolution of the {\em exogenous} block of variables ``X" exhibits autoregressive dynamics. For example, \citet{nijs2001category} examine the sensitivity of over 500 product prices to various marketing promotion strategies (the exogenous block), while \citet{pauwels2008moving} examine changes in subscription rates, search engine referrals and marketing efforts of customers when switched from a free account to a fee-based structure, the latter together with customer characteristics representing the exogenous block. \citet{pesaran2004modeling} examine regional
inter-dependencies, building a model where country specific macroeconomic indicators evolve according to a VAR model and they are influenced 
exogenously by key macroeconomic variables from neighboring countries/regions. Finally, \citet{abeysinghe2001estimation} studies the impact of the
price of oil on Gross Domestic Product growth rates for a number of countries, while controlling for other exogenous variables such as 
the country's consumption and investment expenditures along with its trade balance.

The proposed model gives rise to a network structure that in its most general form corresponds to a multi-partite graph, depicted in Figure \ref{fig:diagram} for 3 blocks, that exhibits a {\em directed acyclic structure} between the constituent blocks, and can also exhibit additional dependence between the nodes in each block. Selected properties of such  multi-block structures, known as {\em chain graphs} \citep{drton2008sinful}, have been studied in the literature. Further, their maximum likelihood estimation for {\em independent and identically distributed} Gaussian data under a high-dimensional {\em sparse} regime is thoroughly investigated in \citet{lin2016penalized}, where a provably convergent estimation procedure is introduced and its theoretical properties are established.

%

\begin{figure}[htbp]
	\centering
	\begin{boxedminipage}{10.5cm}
		\def\layersep{4cm}
		\begin{tikzpicture}[shorten >=1pt,->,draw=black!50, node distance=\layersep]
		\tikzstyle{every pin edge}=[<-,shorten <=1pt]
		\tikzstyle{neuron}=[circle,draw=black!80,minimum size=25pt, inner sep=0pt]
		\tikzstyle{bottom neuron}=[neuron,fill=blue!30];
		\tikzstyle{middle neuron}=[neuron,fill=green!30];
		\tikzstyle{top neuron}=[neuron,fill=red!20];
		\tikzstyle{annot} = [text width=5em, text centered]
		
		\node[bottom neuron] (I-1) at (0,-1) {$x_{t-1}$};
		\node[middle neuron] (I-2) at (0,-2) {$y_{t-1}$};
		\node[top neuron] (I-3) at (0,-3) {$z_{t-1}$};
		
		\path[yshift=0cm] node[bottom neuron] (H-1) at (\layersep,-1) {$x_t$};
		\path[yshift=0cm] node[middle neuron] (H-2) at (\layersep,-2) {$y_t$};
		\path[yshift=0cm] node[top neuron] (H-3) at (\layersep,-3) {$z_t$};
		
		\path[yshift=0cm] node[bottom neuron] (O-1) at (2*\layersep,-1) {$x_{t+1}$};
		\path[yshift=0cm] node[middle neuron] (O-2) at (2*\layersep,-2) {$y_{t+1}$};
		\path[yshift=0cm] node[top neuron] (O-3) at (2*\layersep,-3) {$z_{t+1}$};
		
		\path[->, blue] (I-1) edge (H-1);
		\path[->, blue] (H-1) edge (O-1);
		\path[->, green] (I-2) edge (H-2);
		\path[->, green] (H-2) edge (O-2);
		\path[->, red] (I-3) edge (H-3);
		\path[->, red] (H-3) edge (O-3);
		
		\path[->,dashed, thick,blue] (I-1) edge (H-2);
		\path[->,dashed, thick,blue] (H-1) edge (O-2);
		\path[->,dashed, thick,green] (I-2) edge (H-3);
		\path[->,dashed, thick,green] (H-2) edge (O-3);
		
		\node[annot,above of=H-1, node distance=1cm] (hl) {time $t$};
		\node[annot,left of=hl] {time $t-1$};
		\node[annot,right of=hl] {time $t+1$};
		\end{tikzpicture}	
	\end{boxedminipage}
	\caption{Diagram for a dynamic system with three groups of variables}\vspace*{2mm}\label{fig:diagram}
\end{figure}
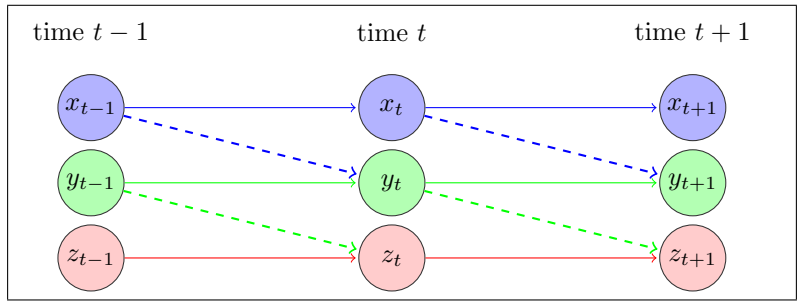

%

Given the wide range of applications of multi-block VAR models, which in addition encompass the widely used VAR-X model, the key contributions
of the current paper are fourfold: (i) formulating the model as a recursive dynamical system and examining its stability properties; (ii) developing a provably convergent algorithm for obtaining the regularized maximum likelihood estimates (MLE) of the model parameters under high-dimensional scaling; (iii) establishing theoretical properties of the ML estimates; and (iv) devising a testing procedure for the parameters that connect the constituent blocks of the model: if the null hypothesis is not rejected, then one is dealing with a set of independently evolving VAR models, otherwise with the posited multi-block VAR model. Finally, the model, estimation and testing procedures are illustrated on an important problem in macroeconomics, as gleaned by the background of the problem and discussion of the results provided in Section~\ref{sec:realdata}. 

For the multi-block VAR model, we assume that the time series within each block are generated by a Gaussian VAR process. Further, the transition matrices within and across blocks can be either {\em sparse} or {\em low rank}. The posited regularized Gaussian likelihood function is not {\em jointly convex} in all the model parameters, which poses a number of technical challenges that are compounded by the presence of temporal dependence. These are successfully addressed and resolved in
Section~\ref{sec:theory}, where we provide a numerically convergent algorithm and establish the theoretical properties of the resulting ML estimates, that constitutes a key contribution in the study of multi-block VAR models. 

The remainder of this paper is organized as follows. In Section~\ref{sec:ProblemFormulation}, we introduce the model setup and the corresponding estimation procedure. In Section~\ref{sec:theory}, we provide consistency properties of the obtained ML estimates 
under a high-dimensional scaling. In Section~\ref{sec:testing}, we introduce the proposed testing framework, both for low-rank
and sparse interaction matrices between the blocks.  Section~\ref{sec:simulation} contains selected numerical results that assess the performance of the estimation and testing procedures. Finally, an application to financial and macroeconomic data that was previously discussed as 
motivation for the model under consideration is presented in Section~\ref{sec:realdata}.

{\em Notation.} Throughout this paper, we use $\smalliii{A}_1$ and $\smalliii{A}_\infty$ respectively to denote the matrix induced $1$-norm and infinity norm of $A\in\mathbb{R}^{m\times n}$, that is, $\smalliii{A}_1 = \max_{1\leq j\leq n}\sum_{i=1}^m |a_{ij}|$, $\smalliii{A}_\infty = \max_{1\leq i\leq m}\sum_{j=1}^n|a_{ij}|$, and use $\|A\|_1$ and $\|A\|_\infty$ respectively to denote the element-wise $1$-norm and infinity norm: $\|A\|_1=\sum_{i,j}|a_{ij}|$, $\|A\|_\infty = \max_{i,j}|a_{ij}|$. Moreover, we use $\smalliii{A}_{*}$, $\smalliii{A}_F$ and $\smalliii{A}_{\text{op}}$ to denote the nuclear, Frobenius and operator norms of $A$, respectively. For two matrices $A$ and $B$ of commensurate dimensions, denote their inner product by $\llangle A, B\rrangle = \tr(A'B)$. Finally, we write $A\succsim B$ if there exists some absolute constant $c$ that is independent of the model parameters such that $A\geq cB$.

\section{Problem Formulation.}\label{sec:ProblemFormulation}

To convey the main ideas and the key technical contributions, we 
consider a recursive linear dynamical system comprising of two blocks of variables, whose structure is given~by:
\begin{equation}\label{eqn:model0}
\begin{split}
X_t & = AX_{t-1} + U_t, \\
Z_t & = BX_{t-1} + CZ_{t-1} + V_t,
\end{split}
\end{equation}
where $X_t\in\mathbb{R}^{p_1}, Z_t\in\mathbb{R}^{p_2}$ are the variables in groups 1 and 2, respectively. The temporal intra-block dependence is captured by transition matrices $A$ and $C$, while the inter-block dependence by $B$. Noise processes $\{U_t\}$ and $\{V_t\}$, respectively, capture
additional contemporaneous intra-block dependence of $X_t$ and $Z_t$, after conditioning on their respective past values. Further, we assume 
that $U_t$ and $V_t$ follow mean zero Gaussian distributions with covariance matrices given by $\Sigma_u$ and $\Sigma_v$, i.e.,
\begin{equation*}
U_t \sim \mathcal{N}(0, \Sigma_u), \quad \text{and} \quad V_t\sim\mathcal{N}(0,\Sigma_v).
\end{equation*}
With the above model setup, the parameters of interest are transition matrices $A\in\R^{p_1\times p_1}$, $B\in\R^{p_2\times p_1}$ and $C\in\R^{p_2\times p_2}$, as well as the covariances $\Sigma_u, \Sigma_v$.  In high-dimensional settings, different combinations of structural assumptions can be imposed on these transition matrices to enable their estimation from limited time series data.
In particular, the intra-block transition matrices $A$ and $C$ are sparse, while the inter-block matrix $B$ can be either sparse or low rank.
Note that the block of $X_t$ variables acts as an {\em exogenous} effect to the evolution of the $Z_t$ block \citep[e.g.,][]{cushman1997identifying,nicolson2016VARX-L}. Further, we assume $\Omega_u:=\Sigma_u^{-1}$ and $\Omega_v:=\Sigma_v^{-1}$ are sparse.

\begin{remark}
	For ease of exposition, we posit a $\VAR(1)$ modeling structure. Extensions to general multi-block structures akin to the one depicted in Figure \ref{fig:diagram} and $\VAR(d)$ specifications are rather straightforward and briefly discussed in Section~\ref{sec:Discussion}. 
\end{remark}

The triangular (recursive) structure of the system enables a certain degree of separability between $X_t$ and $Z_t$.  In the posited model, $X_t$ is a stand-alone $\VAR(1)$ process, and the time series in block $Z_t$ is ``Granger-caused"' by that in block $X_t$, but not vice versa. The second equation in~\eqref{eqn:model0}, as mentioned in the introductory section, also corresponds to the so-called ``VAR-X" model in the econometrics literature \citep[e.g.,][]{sims1980macroeconomics,bianchi2010copula,pesaran2015time}, that extends the standard VAR model to include influences from lagged values of {\em exogenous} variables. Consider the joint process $W_t=(X_t',Z_t')'$, it corresponds to a $\VAR(1)$ model whose transition matrix $G$ has a block triangular form:
\begin{equation}\label{eqn:model1}
W_t = G W_{t-1} + \varepsilon_t, \qquad \text{where}\quad G:=\begin{bmatrix}
A & O \\ B & C
\end{bmatrix}, \quad \varepsilon_t = \begin{bmatrix}
U_t  \\ V_t
\end{bmatrix}.
\end{equation}
The model in~\eqref{eqn:model1} can also be viewed from a Structural Equations Modeling viewpoint involving time series data, and also has a Moving Average representation corresponding to a structural VAR representation with Granger causal ordering \citep{lutkepohl2005new}. As mentioned in the introductory section, the focus of this paper is model parameter estimation under high-dimensional scaling, rather than their cause and effect relationship. For a comprehensive discourse of causality issues for VAR models, we refer the reader to \citet{granger1969investigating,lutkepohl2005new}, and references therein. 

\medskip
Next, we introduce the notion of stability and spectrum with respect to the system. 
\paragraph{\indent System Stability.} To ensure that the joint process $\{W_t\}$ is stable  \citep{lutkepohl2005new}, we require the spectral radius, denoted by $\rho(\cdot)$, of the transition matrix $G$ to be smaller than~1, which is guaranteed by requiring that 
$\rho(A)<1$ and $\rho(C)<1$, since 
\begin{equation*}
\left|  \lambda \mathrm{I}_{p_1\times p_2} - G \right| = \left| \begin{matrix} \lambda \mathrm{I}_{p_1} - A & O \\ - B & \lambda \mathrm{I}_{p_2} - C \end{matrix}\right| = |\lambda \mathrm{I}_{p_1} - A| |\lambda \mathrm{I}_{p_2} - C|,
\end{equation*}
implying that the set of eigenvalues of $G$ is the union of the sets of eigenvalues of $A$ and $C$, hence
\begin{equation*}
\rho(A)<1~,\rho(C)<1, \qquad \Rightarrow \quad \rho(G) = \max\{ \rho(A) ,\rho(C)\} < 1.
\end{equation*}
The latter relation implies that the stability of such a recursive system imposes spectrum constraints only on the diagonal blocks that govern the intra-block evolution, whereas the off-diagonal block that governs the inter-block interaction is left unrestricted.  
\paragraph{\indent Spectrum of the joint process.}  Throughout, we assume that the spectral density of $\{W_t\}$ exists, which then possesses a special structure as a result of the block triangular transition matrix $G$. Formally, we define the spectral density of $\{W_t\}$~as
\begin{equation*}
f_W(\theta) = \frac{1}{2\pi} \sum_{h=-\infty}^\infty \Gamma_W(h)e^{-ih\theta}, \qquad \theta\in[-\pi,\pi],
\end{equation*}
where $\Gamma_W(h):=\E W_tW_{t+h}'$. For two (generic) processes $\{X_t\}$ and $\{Z_t\}$, define their cross-covariance as $\Gamma_{X,Z}(h) = \E X_t Z_{t+h}'$ and $\Gamma_{Z,X}(h)=\E Z_tX_{t+h}'$. In general, $\Gamma_{X,Z}(h)\neq \Gamma_{Z,X}(h)$. The cross-spectra are defined~as:
\begin{equation*}
f_{X,Z}(\theta) := \frac{1}{2\pi}\sum_{h=-\infty}^\infty \Gamma_{X,Z}(h)e^{-ih\theta}, \quad \text{and}\quad f_{Z,X}(\theta) :=  \frac{1}{2\pi}\sum_{h=-\infty}^\infty \Gamma_{Z,X}(h)e^{-ih\theta}, ~~\theta\in[-\pi,\pi].
\end{equation*} 
For the model given in~\eqref{eqn:model1}, by writing out the dynamics of $Z_t$, the cross-spectra between $X_t$ and $Z_t$ are given~by
\begin{equation}\label{eqn:fXZ}
f_{X,Z}(\theta) (\mathrm{I}_{p_2}-C'e^{-i\theta}) = f_X(\theta)B'e^{-i\theta}, \qquad \text{and} \qquad (\mathrm{I}_{p_2}-Ce^{i\theta})f_{Z,X}(\theta) = Be^{i\theta} f_X(\theta).
\end{equation}
Similarly, we have 
\begin{equation}\label{eqn:fZ}
(\mathrm{I}_{p_2}-Ce^{i\theta})f_Z(\theta) = Be^{i\theta} f_{X,Z}(\theta) + f_{V,Z}(\theta).
\end{equation}
Combining~\eqref{eqn:fXZ} and~\eqref{eqn:fZ}, and after some algebra, the spectrum of the joint process $W_t$ is given~by
\begin{equation}\label{eqn:fW}
\small
f_W(\theta) = \big[H_1(e^{i\theta})\big]^{-1}\Big(\big[H_2(e^{i\theta})\big] \big[ \mathbf{1}_{2\times 2}\otimes f_X(\theta)\big] \big[H_2(e^{-i\theta})\big]^\top + \begin{bmatrix}
O & O \\ O & \Sigma_v
\end{bmatrix} \Big) \big[ H_1(e^{-i\theta})\big]^{-\top},
\end{equation}
where $\mathbf{1}_{2\times 2}$ is a $2\times 2$ matrix with all entries being 1, and 
\begin{equation*}
H_1(x) := \begin{bmatrix}
\mathrm{I}_{p_1} & O \\ O & \mathrm{I}_{p_2}-Cx
\end{bmatrix}\in\mathbb{R}^{(p_1+p_2)\times (p_1+p_2)}, ~~~H_2(x) := \begin{bmatrix}
\mathrm{I}_{p_1} & O \\ O & Bx
\end{bmatrix}\in\mathbb{R}^{(p_1+p_2)\times (2p_1)}.
\end{equation*}
Equation~\eqref{eqn:fW} implies that the spectrum of the joint process $\{W_t\}$ can be decomposed into the sum of two parts: the first, is a function of $f_X(\theta)$, while the second part involves the embedded idiosyncratic error process $\{V_t\}$ of $\{Z_t\}$, which only affects the right-bottom block of the spectrum. Note that since $\{W_t\}$ is a $\VAR(1)$ process, its matrix-valued characteristic polynomial is given by 
\begin{equation*}
\mathcal{G}(\theta) := \mathrm{I}_{(p_1+p_2)} - G\theta, 
\end{equation*}
and its spectral density also takes the following form \citep[c.f.][]{hannan1970multiple,anderson1978maximum}:
\begin{equation*}
f_W(\theta) = \frac{1}{2\pi} \big[ \mathcal{G}^{-1}(e^{i\theta}) \big] \Sigma_\varepsilon \big[ \mathcal{G}^{-1}(e^{i\theta}) \big]^*,
\end{equation*}
with
\begin{equation*}
\mathcal{G}(x) = \begin{bmatrix}
\mathrm{I}_{p_1} - Ax &  O \\ -Bx & \mathrm{I}_{p_2} - Cx
\end{bmatrix}, \qquad \Sigma_\varepsilon = \begin{bmatrix}
\Sigma_u & O \\ O & \Sigma_v
\end{bmatrix},
\end{equation*}
and $\mathcal{G}^*$ being the conjugate transpose. One can easily reach the same conclusion as in~\eqref{eqn:fW} by multiplying each term, followed by some algebraic manipulations. 

\subsection{Estimation.} \label{sec:estimation}
Next, we outline the algorithm for obtaining the ML estimates of the transition matrices $A,B$ and $C$ and inverse covariance matrices 
$\Sigma_u^{-1}$ and $\Sigma_v^{-1}$ from time series data. We allow for a potential high-dimensional setting, where the ambient dimensions $p_1$ and $p_2$ of the model exceed the total number of observations~$T$. 

Given centered times series data  $\{x_0,\cdots,x_T\}$ and $\{z_0,\cdots,z_T\}$, we use $\mathcal{X}^T$ and $\mathcal{Z}^T$ respectively, to denote the ``response" matrix from time $1$ to $T$, that is:
\begin{equation*}
\mathcal{X}^T = \begin{bmatrix} x_1 & x_2 & \dots & x_T \end{bmatrix}' 
\quad \text{and} \quad 
\mathcal{Z}^T = \begin{bmatrix} z_1 & z_2 & \dots & z_T \end{bmatrix}',
\end{equation*}
and use $\mathcal{X}$ and $\mathcal{Z}$ without the superscript to denote the ``design" matrix from time $0$ to $T-1$:
\begin{equation*}
\mathcal{X} = \begin{bmatrix} x_0 & x_1 & \dots & x_{T-1} \end{bmatrix}'
\quad \text{and} \quad 
\mathcal{Z} = \begin{bmatrix} z_0 & z_1 & \dots&  z_{T-1} \end{bmatrix}'.
\end{equation*}
We use $\mathcal{U}$ and $\mathcal{V}$ to denote the error matrices. To obtain estimates for the parameters of interest, we formulate optimization problems using a penalized log-likelihood function, with regularization terms corresponding to the imposed structural assumptions on the model parameters--sparsity and/or low-rankness. To solve the optimization problems, we employ block-coordinate descent algorithms, akin to those described in \citet{lin2016penalized}, to obtain the solution.

As previously mentioned, $\{X_t\}$ is not ``Granger-caused" by $Z_t$ and hence it is a stand-alone $\VAR(1)$ process; this enables us to separately estimate the parameters governing the $X_t$ process ($A$ and $\Sigma^{-1}_u$) from those of the $Z_t$ process ($B$, $C$, and $\Sigma^{-1}_v$).

\paragraph{\indent Estimation of $A$ and $\Sigma_u^{-1}$.}  Conditional on the initial observation $x_0$, the likelihood of $\{x_t\}_{t=1}^T$ is given by:
\begin{equation*}
\begin{split}
L(x_T,x_{T-1},\cdots,x_1|x_0) & = L(x_T|x_{T-1},\cdots,x_0)L(x_{T-1}|x_{T-2},\cdots,x_0)\cdots L(x_1|x_0) \\
& = L(x_T|x_{T-1})L(x_{T-1}|x_{T-2})\cdots L(x_1|x_0),
\end{split}
\end{equation*}
where the second equality follows from the Markov property of the process. The log-likelihood function is given~by:
\begin{equation*}
\ell(A,\Sigma_u^{-1}) = \frac{T}{2}\log\det (\Sigma_u^{-1}) -\frac{1}{2}\sum_{t=1}^T (x_t - Ax_{t-1})'\Sigma_u^{-1} (x_t-Ax_{t-1}) + \text{constant}.
\end{equation*}
Letting $\Omega_u:=\Sigma_u^{-1}$, then the penalized maximum likelihood estimator can be written as
\small
\begin{equation}\label{eqn:AOmega_u}
(\widehat{A},\widehat{\Omega}_u) = \argmin\limits_{\substack{A\in\mathbb{R}^{p_1\times p_2}\\\Omega_u\in\mathbb{S}^{++}_{p_1\times p_1} }} \Big\{\text{tr}\big[\Omega_u  (\mathcal{X}^T - \mathcal{X}A')'(\mathcal{X}^T - \mathcal{X}A')/T \big] -  \log\det \Omega_u 
+ \lambda_A\|A\|_1 + \rho_u \|\Omega_u\|_{1,\text{off}} \Big\}. 
\end{equation}
\normalsize
Algorithm~\ref{algo:1} describes the key steps for obtaining $\widehat{A}$ and $\widehat{\Omega}_u$. 

\small
%

\begin{algorithm}[ht]
	\caption{Computational procedure for estimating $A$ and $\Sigma_u^{-1}$.}\label{algo:1}
	
	\KwIn{Time series data $\{x_t\}_{t=1}^T$, tuning parameter $\lambda_A$ and $\rho_u$. }
	\BlankLine
	\textbf{Initialization:} Initialize with $\widehat{\Omega}_u^{(0)}=\mathrm{I}_{p_1}$, then 
	\hspace*{3cm}$\widehat{A}^{(0)} = \argmin\nolimits_A \big\{ \tfrac{1}{T}\vertiii{\mathcal{X}^T-\mathcal{X}A'}_F^2 + \lambda_A\|A\|_1 \big\}$\;
	\textbf{Iterate until convergence:}
	\hspace*{5mm}\begin{minipage}[t]{14cm}
		-- Update $\widehat{\Omega}_u^{(k)}$ by graphical Lasso \citep{friedman2008sparse} on the residuals with the plug-in estimate~$\widehat{A}^{(k)}$\;
		-- Update $\widehat{A}^{(k)}$ with the plug-in $\widehat{\Omega}_u^{(k-1)}$ and cyclically update each row with a Lasso penalty, 
		which solves
		\begin{equation}\label{eqn:optA}
		\min_A \big\{ \tfrac{1}{T}\text{tr}\big[\widehat{\Omega}^{(k-1)}_u (\mathcal{X}^T - \mathcal{X}A)'(\mathcal{X}^T - \mathcal{X}A)/T\big]  + \lambda_A \|A\|_1  \big\}.
		\end{equation}
	\end{minipage}
	\BlankLine
	\KwOut{Estimated sparse transition matrix $\widehat{A}$ and sparse inverse covariance matrix~$\widehat{\Omega}_u$.}
\end{algorithm}

%

\normalsize
Specifically, for fixed $\widehat{\Omega}_u$, each row $j=1,\dots,p_1$ of $\widehat{A}$ is cyclically updated~by:
\begin{equation*}
\widehat{A}^{[s+1]}_{j\cdot} = \argmin\limits_{\beta\in\mathbb{R}^{p_1}}\Big\{ \tfrac{\widehat{\omega}_u^{jj}}{T} \vertii{ \mathcal{X}^T_{\cdot j} + r_j^{[s+1]} -\mathcal{X}\beta }_2^2 + \lambda_A \|\beta\|_1 \Big\},
\end{equation*}
where 
\begin{equation*}
r^{[s+1]}_j = \tfrac{1}{\widehat{\omega}_u^{jj}} \Big[ \sum_{i=1}^{j-1}\widehat{\omega}_u^{ij} \big(\mathcal{X}^T_{\cdot j}-\mathcal{X}(\widehat{A}_{i\cdot}^{[s+1]})'\big)  + \sum_{i=j+1}^{p_1}\widehat{\omega}_u^{ij} \big(\mathcal{X}^T_{\cdot j}-\mathcal{X}(\widehat{A}_{i\cdot}^{[s]})'\big)       \Big].
\end{equation*}
Here $\widehat{\Omega}^{(k)}_u=\big[( \widehat{\omega}_u^{ij})^{(k)}\big]$ is the estimate from the previous iteration, and for notation convenience we drop the index $(k)$ for the outer iteration and use $[s]$ to denote the index for the inner iteration, for each round of cyclic update of the rows. 

\paragraph{\indent Estimation of $B$, $C$ and $\Sigma_v^{-1}$.}  Similarly, to obtain estimates of $B$, $C$ and $\Omega_v:=\Sigma_v^{-1}$, we formulate the optimization problem as follows:
\small
\begin{align}
(\widehat{B},\widehat{C},\widehat{\Omega}_v) = \argmin\limits_{\substack{B\in\mathbb{R}^{p_2\times p_1}, C\in\mathbb{R}^{p_2\times p_2}\\ \Omega_v\in\mathbb{S}^{++}_{p_2\times p_2}}} \Big\{ \text{tr} \big[ \Omega_v&  (\mathcal{Z}^T - \mathcal{X}B' - \mathcal{Z}C')'(\mathcal{Z}^T - \mathcal{X}B' - \mathcal{Z}C')/T \big] - \log\det \Omega_v  \nonumber \\
&  + \lambda_B \mathcal{R}(B) + \lambda_C \|C\|_1 + \rho_v\vertii{\Omega_v}_{1,\text{off}} \Big\}, \label{eqn:BCOmega_v}
\end{align}
\normalsize
where the regularizer $\mathcal{R}(B)=\|B\|_1$ if $B$ is assumed to be sparse, and $\mathcal{R}(B)=\smalliii{B}_*$ if $B$ is assumed to be low rank. Algorithm~\ref{algo:2} outlines the procedure for obtaining estimates $\widehat{B}$, $\widehat{C}$ and $\widehat{\Omega}_v$. Note that $\widehat{B}^{(k)}$ and $\widehat{C}^{(k)}$ need to be treated as a ``joint block" in the outer update and convergence of the ``joint block" is required before moving on to updating $\Omega_v$.  

\small
%
\begin{algorithm}[ht]
	\KwIn{Time series data $\{x_t\}_{t=1}^T$ and $\{z_t\}_{t=1}^T$, tuning parameters $\lambda_B$, $\lambda_C$,~$\rho_v$.}
	\textbf{Initialization:} Initialize with $\widehat{\Omega}_v^{(0)}=\mathrm{I}_{p_2}$, then 
	\begin{equation*}\vspace*{-5mm}
	(\widehat{B}^{(0)},\widehat{C}^{(0)}) = \argmin\limits_{(B,C)}\big\{\tfrac{1}{T} \vertiii{\mathcal{Z}^T - \mathcal{X}B' - \mathcal{Z}C'}_F^2 + \lambda_B\mathcal{R}(B) + \lambda_C \|C\|_1 \big\}.
	\end{equation*}
	\BlankLine
	\textbf{Iterate until convergence:}
	\begin{minipage}[t]{14cm}
		-- Update $\widehat{\Omega}_v^{(k)}$ by graphical Lasso on the residuals with the plug-in estimates $\widehat{B}^{(k)}$ and $\widehat{C}^{(k)}$\;
		-- For fixed $\widehat{\Omega}^{(k)}_v$, $(\widehat{B}^{(k+1)},\widehat{C}^{(k+1)})$ solves
		\begin{equation*}
		\min_{B,C}\big\{\tfrac{1}{T} \text{tr}\big[ \widehat{\Omega}^{(k)}_v(\mathcal{Z}^T - \mathcal{X}B' - \mathcal{Z}C')'(\mathcal{Z}^T - \mathcal{X}B' - \mathcal{Z}C')\big] + \lambda_B\mathcal{R}(B) + \lambda_C \|C\|_1 \big\}.
		\end{equation*}
		\hspace*{2mm}\begin{minipage}[t]{14cm}
			$\cdot$ Fix $\widehat{C}^{[s]}$, update $\widehat{B}^{[s+1]}$ by Lasso or singular value thresholding, which solves
			\begin{equation*}
			\min_B\big\{ \tfrac{1}{T} \text{tr}\big[ \widehat{\Omega}^{(k)}_v(\mathcal{Z}^T - \mathcal{X}B' - \mathcal{Z}\widehat{C}^{[s]'})'(\mathcal{Z}^T - \mathcal{X}B' - \mathcal{Z}\widehat{C}^{[s]'})\big] + \lambda_B\mathcal{R}(B)   \big\};
			\end{equation*}
			$\cdot$ Fix $\widehat{B}^{[s]}$, update $\widehat{C}^{[s]}$ by Lasso, which solves
			\begin{equation*}
			\min_C\big\{ \tfrac{1}{T} \text{tr}\big[ \widehat{\Omega}^{(k)}_v(\mathcal{Z}^T - \mathcal{X}\widehat{B}^{[s]'} - \mathcal{Z}C')'(\mathcal{Z}^T - \mathcal{X}\widehat{B}^{[s]'} - \mathcal{Z}C')\big] + \lambda_C\|C\|_1  \big\}.
			\end{equation*}
		\end{minipage}
	\end{minipage}
	\BlankLine
	\KwOut{Estimated transition matrices $\widehat{B}$, $\widehat{C}$ and sparse $\widehat{\Omega}_v$.}
	\caption{Computational procedure for estimating $B$, $C$ and $\Sigma_v^{-1}$.}\label{algo:2}
\end{algorithm}

%

\normalsize
In many real applications, $B$ is low rank and $C$ is sparse, while in other settings both are sparse. In the first case,  $X_t$ ``Granger-causes" $Z_t$ and the information can be compressed to a lower dimensional space spanned by a relative small number of bases compared to the dimension of the blocks, and $Z_t$ is autoregressive through a subset of its components. Next, we give details for updating $B$ and $C$ under this model specification. 

For fixed $\widehat{\Omega}^{(k)}_v$, with $B$ being low rank and $C$ sparse, the updated $\widehat{B}^{(k)}$ and $\widehat{C}^{(k)}$ satisfies
\small
\begin{equation*}
(\widehat{B}^{(k)},\widehat{C}^{(k)}) = \argmin\limits_{B,C}\left\{\tfrac{1}{T} \text{tr}\big[ \widehat{\Omega}^{(k-1)}_v(\mathcal{Z}^T - \mathcal{X}B' - \mathcal{Z}C')'(\mathcal{Z}^T - \mathcal{X}B' - \mathcal{Z}C')\big] + \lambda_B\vertiii{B}_* + \lambda_C \vertii{C}_1 \right\}.
\end{equation*}
\normalsize
To obtain the solution to the above optimization problem, $B$ and $C$ need to be updated alternately, and within the update of each block, an iterative algorithm is required. Here we drop the superscript $(k)$ that denotes the outer iterations, and use $[s]$ as the inner iteration index for the alternate update between $B$ and $C$, and use $\{t\}$ as the index for the within block iterative update. 
\begin{itemize}
	\item[--] Update $\widehat{B}^{[s+1]}$ for fixed $\widehat{C}^{[s]}$: instead of directly updating $B$, update $\widetilde{B}:=\widehat{\Omega}_v^{1/2}B$ with
	\begin{equation}\label{eqn:solve_for_B}
	\widetilde{B}^{\{t+1\}} = \mathcal{T}_{\alpha_{t+1}\lambda_B} \big( \widetilde{B}^{\{t\}} - \alpha_{t+1} \nabla g(\widetilde{B}^{\{t\}}) \big),
	\end{equation}
	where $\mathcal{T}_{\tau}(\cdot)$ is the singular value thresholding operator with thresholding level~$\tau$, 
	\begin{equation*}
	g(\widetilde{B}) = \tfrac{1}{T}\text{tr}\big[ \widetilde{B}'\widetilde{B} \mathcal{X}'\mathcal{X} - 2 \mathcal{X}' (\mathcal{Z}^T-\mathcal{Z}\widehat{C}^{[s]}) \widehat{\Omega}_v^{1/2}\widetilde{B}\big], 
	\end{equation*}
	and
	\begin{equation*}
	\nabla g(\widetilde{B}) = \tfrac{2}{T} \big[ \widetilde{B}\mathcal{X}'\mathcal{X} - \widehat{\Omega}_v^{1/2}(\mathcal{Z}^T-\mathcal{Z}\widehat{C}^{[s]'})' \mathcal{X}  \big].
	\end{equation*}
	Denote the convergent solution by $\widetilde{B}^{\{\infty\}}$, and $\widehat{B}^{[s+1]}=\widehat{\Omega}_v^{-1/2}\widetilde{B}^{\{\infty\}}$.
	\item[--] Update $\widehat{C}^{[s]}$ for fixed $\widehat{B}^{[s]}$: each row $j=1,\dots,p_2$ of $\widehat{C}^{[s]}$ is cyclically updated~by
	\begin{equation*}
	\widehat{C}^{\{t+1\}}_{j\cdot} = \argmin\limits_{\beta\in\mathbb{R}^{p_2}}\Big\{ \tfrac{\widehat{\omega}_v^{jj}}{T} \vertii{ (\mathcal{Z}^T - \mathcal{X}\widehat{B}^{[s]'})_{\cdot j} + r_j^{\{t+1\}} -\mathcal{Z}\beta }_2^2 + \lambda_C \|\beta\|_1 \Big\},
	\end{equation*}
	where 
	\begin{equation*}
	r^{\{t+1\}}_j = \tfrac{1}{\widehat{\omega}_v^{jj}} \Big[ \sum_{i=1}^{j-1}\widehat{\omega}_v^{ij} \big[(\mathcal{Z}^T - \mathcal{X}\widehat{B}^{[s]'})_{\cdot j}-\mathcal{Z}(\widehat{C}_{i\cdot}^{\{t+1\}})'\big]  + \sum_{i=j+1}^{p_2}\widehat{\omega}_v^{ij} \big[ (\mathcal{Z}^T - \mathcal{X}\widehat{B}^{[s]'})_{\cdot j}-\mathcal{Z}(\widehat{C}_{i\cdot}^{\{t\}})'\big]       \Big],
	\end{equation*}
	and $\widehat{\omega}_v^{ij}$ are entries of $\widehat{\Omega}_v$ coming from the previous outer iteration.
\end{itemize}
\medskip
Although based on the outlined  procedure, a number of iterative steps are required to obtain the final estimate, we have empirically observed that the
number of iterations between the $B, C$ and $\Omega_v$ blocks is usually rather small. Specifically, based on large number of simulation settings (selected
ones presented in Section~\ref{sec:simulation}), for fixed $\widehat{\Omega}^{(k)}_v$, the alternate update for $B$ and $C$ usually converges within 20 iterations, while the update involving $(B,C)$ and $\Omega_v$ takes less than 10 iterations. 

In case, $B$ is sparse, it can be updated by Lasso regression as outlined in Algorithm 2 (details omitted).

\begin{remark}
	In the low dimensional setting where the Gram matrix $\mathcal{X}'\mathcal{X}$ is invertible, the update of $B$ when $\mathcal{R}(B)=\vertiii{B}_*$ can be obtained by a one-shot SVD and singular value thresholding; that is, we first obtain the generalized least squares estimator, then threshold the singular values at a certain level. On the other hand, in the high dimensional setting, an iterative algorithm is required. Note that the singular value thresholding algorithm corresponds to a proximal gradient algorithm, and thus a number of acceleration schemes are available \citep[see][]{nesterov1983method,nesterov1988approach,nesterov2005smooth,nesterov2007gradient}, whose  theoretical properties have been thoroughly investigated in \citet{tseng2008proximal}. We recommend using the acceleration scheme proposed by \citet{beck2009fast}, in which the ``momentum" is carried over to the next iteration, as an extension of \citet{nesterov1983method} to composite functions. Instead of updating $\widetilde{B}$ with \eqref{eqn:solve_for_B}, an ``accelerated" update within the SVT step is given~by:
	\begin{equation*}
	\begin{split}
	y &= \widetilde{B}^{\{t\}} - \frac{t-1}{t+2} \big( \widetilde{B}^{\{t\}} - \widetilde{B}^{\{t-1\}}\big), \\
	\widetilde{B}^{\{t+1\}} &= \mathcal{T}_{\alpha_{t+1}\lambda_B}\big(y-\alpha_{t+1}\nabla g(\widetilde{B}^{\{t\}})\big).
	\end{split}
	\end{equation*}
\end{remark}

Note that the objective function in~\eqref{eqn:AOmega_u} is not {\em jointly convex} in both parameters, but {\em biconvex}. Similarly in~\eqref{eqn:BCOmega_v}, the objective function is biconvex in $\big[(B,C),\Omega_v\big]$. Consequently, convergence to a stationary point is guaranteed, as long as estimates from all iterations lie within a ball around the true value of the parameters, with the radius of the ball upper bounded by a universal constant that only depends on model dimensions and sample size \citep[Theorem 4.1]{lin2016penalized}. This condition is satisfied upon the establishment of consistency properties of the estimates. 

To establish consistency properties of the estimates requires the existence of good initial values for the model parameters $(A, \Omega_u)$, and $(B, C, \Omega_v)$, respectively, in the sense that they are sufficiently close to the true parameters. For the $(A,\Omega_u)$ parameters, the results in \citet{basu2015estimation} guarantee that for random realizations of $\{X_t, E_t\}$, with sufficiently large sample size, the errors of $\widehat{A}^{(0)}$ and $\widehat{\Omega}_u^{(0)}$ are bounded with high probability, which provides us with good initialization values. Yet, additional handling of the bounds is required to ensure that estimates from subsequent iterations are also uniformly close to the true value (see Section~\ref{sec:ConsistencyProp} Theorems~\ref{thm:algo1}). A similar
property for $(\widehat{B}^{(0)}, \widehat{C}^{(0)}, \widehat{\Omega}_v^{(0)})$ and subsequent iterations is established in Section~\ref{sec:ConsistencyProp} Theorems~\ref{thm:algo2} (see also Theorem~\ref{thm:consistencyBC} in Appendix~\ref{appendix:theorems}).

\section{Theoretical Properties.} \label{sec:theory}

In this section, we investigate the theoretical properties of the penalized maximum likelihood estimation procedure proposed in Section~\ref{sec:ProblemFormulation}, with an emphasis on the error bounds for the obtained estimates. We focus on the model specification in which the inter-block transition matrix $B$ is {\em low rank}, which is of interest in many applied settings. Specifically, we consider the consistency properties of $\widehat{A}$ and $(\widehat{B},\widehat{C})$ that are solutions to the following two optimization problems:
\small
\begin{equation}\label{eqn:A-Omega}
(\widehat{A},\widehat{\Omega}_u) = \argmin\limits_{A,\Omega_u}\Big\{ \text{tr}\left[ \Omega_u(\mathcal{X}^T - \mathcal{X}A')'(\mathcal{X}^T - \mathcal{X}A')/T  \right] -\log\det\Omega_u  + \lambda_A \|A\|_1 + \rho_u\|\Omega_u\|_{1,\text{off}}\Big\},
\end{equation}
\normalsize and
\small
\begin{equation}\label{eqn:BC-Omega}
\begin{split}
(\widehat{B},\widehat{C},\widehat{\Omega}_v) = \argmin\limits_{B,C,\Omega_v} \Big\{ \text{tr} \big[ \Omega_v   (\mathcal{Z}^T - \mathcal{X}B' - \mathcal{Z}C')'(\mathcal{Z}^T - \mathcal{X}B' - \mathcal{Z}C')/T \big] - \log\det \Omega_v  \\
+ \lambda_B \vertiii{B}_* + \lambda_C \|C\|_1 + \rho_v\vertii{\Omega_v}_{1,\text{off}} \Big\}.
\end{split}
\end{equation}
\normalsize
The case of a sparse $B$ can be handled similarly to that of $A$ and/or $C$ with minor modifications (details shown in Appendix~\ref{appendix:sparseB}). 

\subsection{A road map for establishing the consistency results.} \label{sec:roadmap}

Next, we outline the main steps followed in establishing the theoretical properties for the model parameters.
Throughout, we denote with a superscript~``$\star$" the true value of the corresponding parameters. 

The following key concepts, widely used in high-dimensional regularized estimation problems, are needed in subsequent developments.

\begin{definition}[Restricted Strong Convexity (RSC)]\label{defn:RSC} For some generic operator $\mathfrak{X}:\mathbb{R}^{m_1\times m_2}\mapsto \mathbb{R}^{T\times m_1}$, it satisfies the RSC condition with respect to norm $\Phi$ with curvature $\alpha_{\text{RSC}}>0$ and tolerance $\tau > 0$ if 
	\begin{equation*}
	\frac{1}{2T}\vertiii{\mathfrak{X}(\Delta)}_F^2 \geq \alpha_{\text{RSC}} \vertiii{\Delta}_F^2 - \tau \Phi^2(\Delta), \qquad \text{for some }\Delta\in \mathbb{R}^{m_1\times m_2}. 
	\end{equation*} 
\end{definition}

Note that the choice of the norm $\Phi$ is context specific. For example, in sparse regression problems, $\Phi(\Delta)=\|\Delta\|_1$ corresponds to the element-wise $\ell_1$ norm of the matrix (or the usual vector $\ell_1$ norm for the vectorized version). The RSC condition becomes equivalent to the {\em restricted eigenvalue (RE) condition} \citep[see][and references therein]{loh2011high,basu2015estimation}. This is the case for the problem of
estimating transition matrix $A$. For estimating $B$ and $C$, define $\mathcal{Q}$ to be the weighted regularizer $\mathcal{Q}(B,C):=\vertiii{B}_*+\tfrac{\Lambda_C}{\lambda_B}\|C\|_1$, and the associated norm $\Phi$ in this setting is defined as
\begin{equation*}
\Phi(\Delta) := \inf\limits_{B_{\text{aug}}+C_{\text{aug}}=\Delta} \mathcal{Q}(B,C).
\end{equation*}

\begin{definition}[Diagonal dominance] A matrix $\Omega\in\mathbb{R}^{p\times p}$ is strictly diagonally dominant if
	\begin{equation*}
	|\Omega_{ii}| > \sum_{j\neq i} |\Omega_{ij}|, \quad \forall~~i=1,\cdots,p. 
	\end{equation*}
\end{definition}

\begin{definition}[Incoherence condition \citep{ravikumar2011high}]\label{def:incoherence} A matrix $\Omega\in\mathbb{R}^{p\times p}$ satisfies the incoherence condition if:
	\begin{equation*}
	\max\limits_{e\in (S_\Omega)^c}\| H_{eS_\Omega}(H_{S_\Omega S_\Omega})^{-1}\|_1 \leq 1-\xi, \quad \text{for some }\xi\in(0,1),
	\end{equation*}
	where $H_{S_\Omega S_\Omega}$ denotes the Hessian of the log-determinant barrier $\log\det\Omega$ restricted to the true edge set of $\Omega$ denoted by $S_\Omega$, and $H_{eS}$ is similarly defined. 
\end{definition}

The above two conditions are associated with the inverse covariance matrices $\Omega_u$ and $\Omega_v$. Specifically, the diagonal dominance condition is required for $\Omega_u^\star$ and $\Omega_v^\star$ as we build the consistency properties for $\widehat{A}$ and $(\widehat{B},\widehat{C})$ with the penalized maximum likelihood formulation. The incoherence condition is primarily required for establishing the consistency of $\widehat{\Omega}_u$ and~$\widehat{\Omega}_v$. 

We additionally introduce the upper and lower extremes of the spectrum, defined~as
\begin{equation*}
\mathcal{M}(f_X) := \mathop{\text{esssup}}\limits_{\theta\in[-\pi,\pi]} \Lambda_{\max}(f_X(\theta)) \qquad\text{and}\qquad  \mathfrak{m}(f_X) := \mathop{\text{essinf}}\limits_{\theta\in[-\pi,\pi]} \Lambda_{\min}(f_X(\theta)).
\end{equation*}
Analogously, the upper extreme for the cross-spectrum is given~by:
\begin{equation*}
\mathcal{M}(f_{X,Z}) := \mathop{\text{esssup}}\limits_{\theta\in[-\pi,\pi]}\sqrt{\Lambda_{\max}(f^*_{X,Z}(\theta) f_{X,Z}(\theta)) },
\end{equation*}
with $f^*_{X,Z}(\theta)$ being the conjugate transpose of $f_{X,Z}(\theta)$. With this definition, 
\begin{equation*}
\mathcal{M}(f_{X,Z}) = \mathcal{M}(f_{Z,X}).
\end{equation*}

Next, consider the solution to~\eqref{eqn:A-Omega} that is obtained by the alternate update of $A$ and $\Omega_u$. If $\Omega_u$ is held fixed, then $A$ solves~\eqref{eqn:Abar}, and we denote the solution by $\bar{A}$ and its corresponding vectorized version as $\bar{\beta}_A:=\text{vec}(\bar{A})$:
\begin{equation}\label{eqn:Abar}
\bar{\beta}_A := \argmin\limits_{\beta\in\mathbb{R}^{p_1^2}} \big\{ -2\beta' \gamma_X + \beta'\Gamma_X \beta + \lambda_A\|\beta\|_1 \big\},
\end{equation}
where
\begin{equation} \label{eqn:GammaA-gammaA}
\Gamma_X = \Omega_u \otimes \tfrac{\mathcal{X}'\mathcal{X}}{T}, \qquad \gamma_X = \tfrac{1}{T}\left(\Omega_u\otimes \mathcal{X}'\right) \text{vec}(\mathcal{X}^T).
\end{equation}
Using a similar notation, if $A$ is held fixed, then $\Omega_u$ solves~\eqref{eqn:Omegabar}:
\begin{equation}\label{eqn:Omegabar}
\bar{\Omega}_u := \argmin\limits_{\Theta\in\mathbb{S}^{++}_{p_1\times p_1}}\big\{ \log\det\Omega_u - \text{trace}\left( S_u\Omega_u \right) + \rho_u\|\Omega_u\|_{1,\text{off}} \big\},
\end{equation}
where 
\begin{equation}\label{eqn:SOmegau}
S_u = \tfrac{1}{T}(\mathcal{X}^T - \mathcal{X}A')'(\mathcal{X}^T - \mathcal{X}A'). 
\end{equation}
For {\em fixed} realizations of $\mathcal{X}$ and $\mathcal{U}$, by \cite{basu2015estimation}, the error bound of $\bar{\beta}_A$ relies on (1) $\Gamma_X$ satisfying the RSC condition defined above; and (2) the tuning parameter $\lambda_A$ is chosen in accordance with a deviation bound condition associated with $\|\mathcal{X}'\mathcal{U}\Omega_u/T\|_\infty$. By \citet{ravikumar2011high}, the error bound of $\bar{\Omega}_u$ relies on how well $S_u$ concentrates around $\Sigma_u^\star$, that is, $\|S_u-\Sigma_u^\star\|_\infty$. Specifically, for~\eqref{eqn:GammaA-gammaA} and~\eqref{eqn:SOmegau}, with $\Omega_u^\star$ and $A^\star$ plugged in respectively, for {\em random} realizations of $\mathcal{X}$ and $\mathcal{U}$, these conditions hold with high probability. In the actual implementation of the algorithm, however, quantities in~\eqref{eqn:GammaA-gammaA} and~\eqref{eqn:SOmegau} are substituted by estimates so that at iteration $k$, $\widehat{\beta}_A^{(k)}$ and $\widehat{\Omega}_u^{(k)}$ solve
\begin{align*}
\widehat{\beta}^{(k)}_A & := \argmin\limits_{\beta\in\mathbb{R}^{p_1^2}} \big\{ -2\beta' \widehat{\gamma}^{(k)}_X + \beta'\widehat{\Gamma}^{(k)}_X \beta + \lambda_A\|\beta\|_1 \big\}, \\
\widehat{\Omega}^{(k)}_u & := \argmin\limits_{\Omega_u\in\mathbb{S}^{++}_{p_1\times p_1}}\big\{ \log\det\Omega_u - \text{trace}\big( \widehat{S}^{(k)}_u\Omega_u \big) + \rho_u\|\Omega_u\|_{1,\text{off}} \big\},
\end{align*}
where
\begin{equation*}
\widehat{\Gamma}^{(k)}_X = \widehat{\Omega}^{(k-1)}_u \otimes \tfrac{\mathcal{X}'\mathcal{X}}{T},\quad  \widehat{\gamma}^{(k)}_X = \tfrac{1}{T}\big( \widehat{\Omega}_u^{(k-1)}\otimes \mathcal{X} \big),  \quad \widehat{S}^{(k)}_u = \tfrac{1}{T}\big[\mathcal{X}^T - \mathcal{X}(\widehat{A}^{(k)})'\big]'\big[\mathcal{X}^T - \mathcal{X}(\widehat{A}^{(k)})'\big].
\end{equation*}
As a consequence, to establish the finite-sample bounds of $\widehat{A}$ and $\widehat{\Omega}_u$ given in~\eqref{eqn:A-Omega}, we need $\widehat{\Gamma}^{(k)}_X$ to satisfy the RSC condition, a bound on $\|\mathcal{X}'\mathcal{U}\widehat{\Omega}_u^{(k-1)}\|_\infty$ and a bound on $\|\widehat{S}_u^{(k)} - \Sigma_u^\star\|_\infty$ for all $k$. Toward this end, we prove that for random realizations of $\mathcal{X}$ and $\mathcal{U}$, with high probability, the RSC condition for $\widehat{\Gamma}_X^{(k)}$ and the universal bounds for $\|\mathcal{X}'\mathcal{U}\widehat{\Omega}_u^{(k-1)}\|_\infty$ and $\|\widehat{S}_u^{(k)} - \Sigma_u^\star\|_\infty$ hold for {\em all iterations} $k$, albeit the quantities of interest rely on estimates from the previous or current iterations. Consistency results of $\widehat{A}$ and $\widehat{\Omega}_u$ then readily follow.

Next, consider the solution to~\eqref{eqn:BC-Omega} that alternately updates $(B,C)$ and $\Omega_v$. As the regularization term involves both the nuclear norm penalty and the $\ell_1$ norm penalty, additional handling of the norms is required which leverages the idea of decomposable regularizers \citep{agarwal2012noisy}. Specifically, if $\Omega_v$ and $(B,C)$ are respectively held fixed, then
\begin{align*}
(\bar{B},\bar{C}) & := \argmin\limits_{B,C}\Big\{\tfrac{1}{T} \text{tr}\big[ \Omega_v(\mathcal{Z}^T - \mathcal{X}B' - \mathcal{Z}C')'(\mathcal{Z}^T - \mathcal{X}B' - \mathcal{Z}C')\big] + \lambda_B\vertiii{B}_* + \lambda_C \vertii{C}_1 \Big\},\\
\bar{\Omega}_v & := \argmin\limits_{\Omega_v}\big\{ \log\det\Omega_v - \text{trace}\big( S_v\Omega_v \big) + \rho_v\|\Omega_v\|_{1,\text{off}} \big\},
\end{align*}
where $S_v=\tfrac{1}{T}(Z^T-\mathcal{X}B'-\mathcal{Z}C')'(Z^T-\mathcal{X}B'-\mathcal{Z}C')$. If we let $\mathcal{W}:=[\mathcal{X},\mathcal{Z}]\in\R^{T\times (p_1+p_2)}$, and define the operator $\mathfrak{W}_{\Omega_v}:\mathbb{R}^{p_2\times (p_1+p_2)}\mapsto \mathbb{R}^{T\times p_2}$ induced jointly by $\mathcal{W}$ and $\Omega_v$ as 
\begin{equation}\label{eqn:W-operator}
\mathfrak{W}_{\Omega_v}(\Delta) := \mathcal{W}\Delta'\Omega_v^{1/2} ~~\text{for }\Delta\in\mathbb{R}^{p_2\times (p_1+p_2)},
\end{equation}
then $\bar{B}_{\text{aug}}:= [\bar{B},O_{p_2\times p_2}]$ and $\bar{C}_{\text{aug}}:=[O_{p_2\times p_1},\bar{C}]$ are equivalently given by
\small
\begin{equation}
(\bar{B}_{\text{aug}},\bar{C}_{\text{aug}}) = \argmin\limits_{B,C}\Big\{\tfrac{1}{T} \vertiii{ \mathcal{Z}^T\Omega_v^{1/2} - \mathfrak{W}_{\Omega_v}(B_{\text{aug}} + C_{\text{aug}})}_F^2 + \lambda_B\vertiii{B}_* + \lambda_C \vertii{C}_1 \Big\},
\end{equation}
\normalsize
where $B_{\text{aug}}:=[B,O_{p_2\times p_2}],C_{\text{aug}}:=[O_{p_2\times p_1},C]\in\mathbb{R}^{p_2\times (p_1+p_2)}$. Then, for fixed realizations of $\mathcal{Z}$, $\mathcal{X}$ and $\mathcal{V}$, with an extension of \citet{agarwal2012noisy} the error bound of $(\bar{B}_{\text{aug}},\bar{C}_{\text{aug}})$ relies on (1) the operator $\mathfrak{W}_{\Omega_v}$ satisfying the RSC condition; and (2) tuning parameters $\lambda_B$ and $\lambda_C$ are respectively chosen in accordance with the deviation bound conditions associated with
\begin{equation}\label{eqn:W-deviation}
\smalliii{\mathcal{W}'\mathcal{V}\Omega_v/T}_{\text{op}} \qquad \text{and} \qquad \smalliii{\mathcal{W}'\mathcal{V}\Omega_v/T}_{\infty}.
\end{equation}
The error bound of $\bar{\Omega}_v$ again relies on $\|S_v-\Sigma_v^\star\|_\infty$. In an analogous way, for the actual alternate update, 
\begin{equation*}
\begin{split}
(\widehat{B}^{(k)}_{\text{aug}},\widehat{C}^{(k)}_{\text{aug}}) & = \argmin\limits_{B,C}\Big\{\tfrac{1}{T} \vertiii{ \mathcal{Z}^T \big[\widehat{\Omega}^{(k-1)}_v\big]^{1/2} - \mathfrak{W}_{\widehat{\Omega}^{(k-1)}_v}(B_{\text{aug}} + C_{\text{aug}})}_F^2 + \lambda_B\vertiii{B}_* + \lambda_C \vertii{C}_1 \Big\}, \\
\widehat{\Omega}^{(k)}_v & := \argmin\limits_{\Omega_v}\big\{ \log\det\Omega_v - \text{trace}\big( \widehat{S}^{(k)}_v\Omega_v \big) + \rho_v\|\Omega_v\|_{1,\text{off}} \big\},
\end{split}
\end{equation*}
and the error bound of $(\widehat{B},\widehat{C},\widehat{\Omega}_v)$ defined in~\eqref{eqn:BC-Omega} depends on the properties of $\mathfrak{W}_{\widehat{\Omega}_v^{(k)}}$, $\smalliii{\mathcal{W}'\mathcal{V}\Omega^{(k)}_v/T}_{\text{op}}$ and $\smalliii{\mathcal{W}'\mathcal{V}\Omega_v^{(k)}/T}_{\infty}$ for all $k$. Specifically, when $\Omega_v$ and $(B,C)$ (in~\eqref{eqn:W-operator} and~\eqref{eqn:W-deviation}, resp.) are substituted by estimated quantities, we prove that the RSC condition and bounds hold with high probability for random realizations of $\mathcal{Z}$, $\mathcal{X}$ and $\mathcal{V}$, for {\em all iterations} $k$, which then establishes the consistency properties of $(\widehat{B},\widehat{C})$ and $\widehat{\Omega}_v$.

\subsection{Consistency results for the Maximum Likelihood estimators.}\label{sec:ConsistencyProp}

Theorems~\ref{thm:algo1} and \ref{thm:algo2} below give the error bounds for the estimators in~\eqref{eqn:A-Omega} and~\eqref{eqn:BC-Omega} obtained through Algorithms~\ref{algo:1} and~\ref{algo:2}, using random realizations coming from the stable VAR system defined in~\eqref{eqn:model0}. As previously mentioned, to establish error bounds for both the transition matrices and the inverse covariance matrix obtained from alternating updates, we need to take into account that the quantities associated with the RSC condition and the deviation bound condition are based on {\em estimated quantities} obtained from the previous iteration. On the other hand, the sources of randomness contained in the observed data are fixed, hence errors from observed data stop accumulating once all sources of randomness are considered after a few iterations, which govern both the leading term of the error bounds and the probability for the bounds to hold. 

Specifically, using the same notation as defined in Section~\ref{sec:roadmap}, we obtain the error bounds of the estimated transition matrices and inverse covariance matrices iteratively, building upon that for all iterations $k$: 
\begin{enumerate}[(1)]
	\itemsep-0.5em
	\item $\widehat{\Gamma}^{(k)}_X$ or the operator $\mathfrak{W}_{\widehat{\Omega}_v^{(k)}}$ satisfies the RSC condition;
	\item deviation bounds hold for 
	$\|\mathcal{X}'\mathcal{U}\widehat{\Omega}_u^{(k)}/T\|_\infty$, $\|\mathcal{W}'\mathcal{V}\widehat{\Omega}_v^{(k)}/T\|_\infty$, and $\smalliii{\mathcal{W}'\mathcal{V}\widehat{\Omega}_v^{(k)}/T}_\text{op}$;
	\item a good concentration given by $\|\widehat{S}_u^{(k)}-\Sigma_u^\star\|_\infty$ and $\|\widehat{S}_v^{(k)}-\Sigma_v^\star\|_\infty$.
\end{enumerate}
We keep track of how the bounds change in each iteration until convergence, by properly controlling the norms and track the rate of the error bound that depends on $p_1,p_2$ and $T$, and reach the conclusion that the error bounds {\em hold uniformly for all iterations}, for the estimates of
both the transition matrices $A, B$ and $C$ and the inverse covariance matrices $\Omega_u$ and $\Omega_v$. 

\begin{theorem}\label{thm:algo1}
	Consider the stable Gaussian VAR process defined in~\eqref{eqn:model0} in which $A^\star$ is assumed to be $s_A^\star$-sparse. Further, assume the following:
	\begin{itemize}
		\item[C.1] The incoherence condition holds for $\Omega_u^\star$.
		\item[C.2] $\Omega_u^\star$ is diagonally dominant. 
		\item[C.3] The maximum node degree of $\Omega_u^\star$ satisfies $d_{\Omega_u^\star}^{\max}=o(p_1)$.  
	\end{itemize} 
	Then, for random realizations of $\{X_t\}$ and $\{U_t\}$, and the sequence $\{\widehat{A}^{(k)},\widehat{\Omega}_u^{(k)}\}_{k}$ returned by Algorithm~\ref{algo:1} outlined in Section~\ref{sec:estimation},
	there exist constants $c_1,c_2,\tilde{c}_1,\tilde{c}_2>0,\tau>0$ such that for sample size $T\succsim \max\{(d_{\Omega_u^\star}^{\max})^2,s_A^\star\}\log p_1$, with probability at least $$1-c_1\exp(-c_2T)-\tilde{c}_1\exp(-\tilde{c}_2\log p_1)-\exp(-\tau \log p_1),$$
	the following hold for all $k\geq 1$ for some $C_0,C'_0>0$ that are functions of the upper and lower extremes $\mathcal{M}(f_X),\mathfrak{m}(f_X)$ 
	of the spectrum $f_X(\theta)$ and do not depend on $p_1,T$ or $k$:
	\begin{enumerate}[(i)]
		\item $\widehat{\Gamma}_X^{(k)}$ satisfies the RSC condition;
		\item $\|\mathcal{X}'\mathcal{U}\widehat{\Omega}_u^{(k)}/T\|_\infty \leq C_0\sqrt{\tfrac{\log p_1}{T}}$;
		\item $\| \widehat{S}_u^{(k)} - \Sigma_u^\star\|_\infty \leq C_0'\sqrt{\tfrac{\log p_1}{T}}$.
	\end{enumerate}
	As a consequence, the following bounds hold uniformly for all iterations $k\geq 1$:
	\begin{equation*}
	\vertiii{\widehat{A}^{(k)}-A^\star}_F = O\Big(\sqrt{\tfrac{ s_A^\star\log p_1}{T}}\Big), \qquad \vertiii{\widehat{\Omega}_u^{(k)}-\Omega_u^\star}_F = O\Big( \sqrt{\tfrac{(s^\star_{\Omega_u}+p_1)\log p_1}{T}}\Big).
	\end{equation*}
\end{theorem}

It should be noted that the above result establishes the {\em consistency for the ML estimates} of the model presented in
\citet{basu2015estimation}.

\begin{theorem}\label{thm:algo2}
	Consider the stable Gaussian VAR system defined in~\eqref{eqn:model0} in which $B^\star$ is assumed  to be low rank with rank $r_B^\star$ and $C^\star$ is assumed to be $s_C^\star$-sparse. Further, assume the following 
	\begin{itemize}
		\item[C.1] The incoherence condition holds for $\Omega_v^\star$.
		\item[C.2] $\Omega_v^\star$ is diagonally dominant. 
		\item[C.3] The maximum node degree of $\Omega_v^\star$ satisfies $d_{\Omega_v^\star}^{\max} = o(p_2)$.   
	\end{itemize}
	Then, for random realizations of $\{X_t\}$, $\{Z_t\}$ and $\{V_t\}$, and the sequence $\{(\widehat{B}^{(k)},\widehat{C}^{(k)}),\widehat{\Omega}_v^{(k)}\}_k$ returned by Algorithm~\ref{algo:2} outlined in Section~\ref{sec:estimation}, 
	there exist constants $\{c_i,\tilde{c}_i\},i=(0,1,2)$ and $\tau>0$ such that for sample size $T\succsim (d_{\Omega_v^\star}^{\max})^2(p_1+2p_2)$, with probability at least $$1- c_0\exp\{-\tilde{c}_0(p_1+p_2)\} - c_1\exp\{-\tilde{c}_1(p_1+2p_2)\} - c_2\exp\{-\tilde{c}_2 \log[p_2(p_1+p_2)]\} - \exp\{-\tau \log p_2\},$$
	the following hold for all $k\geq 1$ for $C_0,C_0',C_0''>0$ that are functions of the upper and lower extremes $\mathcal{M}(f_W),\mathfrak{m}(f_W)$
	of the spectrum $f_W(\theta)$ and of the upper extreme $\mathcal{M}(f_{W,V})$ of the cross-spectrum $f_{W,V}(\theta)$ 
	and do not depend on $p_1,p_2$ or $T$: 
	\begin{enumerate}[(i)]
		\item $\widehat{\Gamma}_W^{(k)}$ satisfies the RSC condition;
		\item $\|\mathcal{W}'\mathcal{V}\widehat{\Omega}_v^{(k)}/T\|_\infty \leq C_0\sqrt{\tfrac{(p_1+p_2)+p_2}{T}}$ and $\smalliii{\mathcal{W}'\mathcal{V}\widehat{\Omega}_v^{(k)}/T}_\text{op} \leq C'_0\sqrt{\tfrac{(p_1+p_2)+p_2}{T}}$;
		\item $\| \widehat{S}_v^{(k)} - \Sigma_v^\star\|_\infty \leq C_0''\sqrt{\tfrac{(p_1+p_2)+p_2}{T}}$.
	\end{enumerate}	
	As a consequence, the following bounds hold uniformly for all iterations $k\geq 1$:
	\begin{equation*}
	\vertiii{\widehat{B}^{(k)}-B^\star}^2_F+\vertiii{\widehat{C}^{(k)}-C^\star}^2_F = O\Big(\tfrac{\max\{r_B^\star,s_C^\star\}(p_1+2p_2)}{T}\Big),~~
	\vertiii{\widehat{\Omega}_v^{(k)} - \Omega_v^\star}_F = O\Big(\sqrt{\tfrac{(s^\star_{\Omega_v}+p_2)(p_1+2p_2)}{T}}\Big).
	\end{equation*}
\end{theorem}

\begin{remark}
	It is worth pointing out that the initializers $\widehat{A}^{(0)}$ and $(\widehat{B}^{(0)},\widehat{C}^{(0)})$ are slightly different from those obtained in successive iterations, as they come from the penalized least square formulation where the inverse covariance matrices are temporarily assumed diagonal. Consistency results for these initializers under deterministic realizations are established in Theorems~\ref{thm:consistencyA} and~\ref{thm:consistencyBC} (see Appendix~\ref{appendix:theorems}), and the corresponding conditions are later verified for random realizations in Lemmas~\ref{lemma:REX} to~\ref{lemma:operator-infinity} (see Appendix~\ref{appendix:KeyLemmas}). These theorems and lemmas serve as the stepping stone toward the proofs of Theorems~\ref{thm:algo1} and~\ref{thm:algo2}. 
	
	Further, the constants $C_0, C'_0, C''_0$ reflect both the temporal dependence among $X_t$ and $Z_t$ blocks, as well as the cross-sectional dependence within and across the two blocks.
\end{remark}
\subsection{The effect of temporal and cross-dependence on the established bounds.} \label{sec:bound-discussion}

We conclude this section with a discussion on the error bounds of the estimators that provides additional insight into the impact of
temporal and cross dependence within and between the blocks; specifically, how the exact bounds depend on the underlying processes through their spectra when explicitly taking into consideration the triangular structure of the joint transition matrix. 

First, we introduce additional notations needed in subsequent technical developments. The definition of the spectral densities and the cross-spectrum are the same as previously defined in Section~\ref{sec:ProblemFormulation} and their upper and extremes are defined in Section~\ref{sec:roadmap}. For $\{X_t\}$ defined in~\eqref{eqn:model0}, let $\mathcal{A}(\theta) = \mathrm{I}_{p_1} - A\theta$ denote the characteristic matrix-valued polynomial of $\{X_t\}$ and $\mathcal{A}^*(\theta)$ denote its conjugate. We further define its upper and lower extremes by:
\begin{equation*}
\mu_{\max}(\mathcal{A}) = \max\limits_{|\theta|=1} \Lambda_{\max}\left( \mathcal{A}^*(\theta) \mathcal{A}(\theta) \right), ~~~\mu_{\min}(\mathcal{A}) = \min\limits_{|\theta|=1} \Lambda_{\min}\left( \mathcal{A}^*(\theta) \mathcal{A}(\theta) \right).
\end{equation*}
The same set of quantities for the joint process $\{W_t=(X_t',Z_t')'\}$ are analogously defined, that is, 
\begin{equation*}
\mathcal{G}(\theta) = \mathrm{I}_{p_1+p_2} - G\theta, ~~\mu_{\max}(\mathcal{G}) = \max\limits_{|\theta|=1} \Lambda_{\max}\left( \mathcal{G}^*(\theta) \mathcal{G}(\theta) \right), ~~\mu_{\min}(\mathcal{G}) = \min\limits_{|\theta|=1} \Lambda_{\min}\left( \mathcal{G}^*(\theta) \mathcal{G}(\theta) \right).
\end{equation*}
Using the result in Theorem~\ref{thm:algo2} as an example, we show how the error bound depends on the underlying processes $\{(X'_t,Z'_t)'\}$. Specifically, we note that the bounds for $(\widehat{B}^{(k)},\widehat{C}^{(k)})$ can be equivalently written as
\begin{equation*}
\smalliii{\widehat{B}^{(k)}-B^\star}^2_F+\smalliii{\widehat{C}^{(k)}-C^\star}^2_F \leq \bar{C}\Big(\tfrac{\max\{r_B^\star,s_C^\star\}(p_1+2p_2)}{T}\Big),
\end{equation*}
which holds for all $k$ and some constant $\bar{C}$ that does not depend on $p_1,p_2$ or $T$. Specifically, by Theorem~\ref{thm:consistencyBC}, Lemmas~\ref{lemma:RSCW} and~\ref{lemma:operator-infinity}, 
\begin{equation*}
C_0 \propto \left[\mathcal{M}(f_W) + \Lambda_{\max}(\Sigma_v) + \mathcal{M}(f_{W,V})\right]/\mathfrak{m}(f_W). 
\end{equation*}
This indicates that the exact error bound depends on $\mathfrak{m}(f_W), \mathcal{M}(f_W)$ and $\mathcal{M}(f_{W,V})$. Next, we provide bounds on these quantities. The joint process $W_t$ as we have noted in~\eqref{eqn:model1}, is a $\VAR(1)$ process with characteristic polynomial
$\mathcal{G}(\theta)$ and spectral density $f_W(\theta)$. The bounds for $\mathfrak{m}(f_W)$ and $\mathcal{M}(f_W)$ are given by \citet[][Proposition 2.1]{basu2015estimation}, that is,
\begin{equation}\label{eqn:bound:lu}
\mathfrak{m}(f_W) \geq  \frac{\min\{\Lambda_{\min}(\Sigma_u),\Lambda_{\min}(\Sigma_v)\}}{(2\pi)\mu_{\max}(\mathcal{G})} \quad \text{and} \quad \mathcal{M}(f_W) \leq  \frac{\max\{\Lambda_{\max}(\Sigma_u),\Lambda_{\max}(\Sigma_v)\}}{(2\pi)\mu_{\min}(\mathcal{G})}. 
\end{equation}
Consider the bound for $\mathcal{M}(f_{W,V})$. First, we note that $\{V_t\}$ is a sub-process of the joint error process $\{\varepsilon_t\}$, where $\varepsilon_t=(U_t',V_t')'$. Then, by Lemma~\ref{lemma:boundsubprocess}, 
\begin{equation*}
\mathcal{M}(f_{W,V}) \leq \mathcal{M}(f_{W,\varepsilon}) \leq \mathcal{M}(f_W)\mu_{\max}(\mathcal{G}),
\end{equation*}
where the second inequality follows from \citet[][Proof of Proposition 2.4]{basu2015estimation}.

What are left to be bounded are $\mu_{\min}(\mathcal{G})$ and $\mu_{\max}(\mathcal{G})$. By Proposition 2.2 in \citet{basu2015estimation}, these two quantities are bounded by:
\begin{equation}\label{eqn:bound-upper}
\mu_{\max}(\mathcal{G})  \leq \Big[ 1 + \frac{\vertiii{G}_\infty + \vertiii{G}_1}{2}\Big]^2\\
\end{equation}
and 
\begin{equation*}
\mu_{\min}(\mathcal{G}) \geq (1-\rho(G))^2 \cdot \vertiii{P}^{-2}_{\text{op}}\cdot\vertiii{P^{-1}}^{-2}_{\text{op}},
\end{equation*}
where $G = P\Lambda_GP^{-1}$ with $\Lambda_G$ being a diagonal matrix consisting of the eigenvalues of $G$. 
Since $\smalliii{P^{-1}}_{\text{op}}\geq \smalliii{P}_{\text{op}}^{-1}$, it follows that 
\begin{equation*}
\vertiii{P}^{-2}_{\text{op}}\cdot\vertiii{P^{-1}}^{-2}_{\text{op}} \geq \vertiii{P^{-1}}_{\text{op}}^2 \cdot \vertiii{P^{-1}}^{-2}_{\text{op}} = 1, 
\end{equation*}
and therefore
\begin{equation}\label{eqn:bound-lower}
\mu_{\min}(\mathcal{G})\geq \left(1-\max\{\rho(A),\rho(C)\}\right)^2.
\end{equation}

\begin{remark}
	{\em The impact of the system's lower-triangular  structure on the established bounds.} Consider the bounds in~\eqref{eqn:bound-upper} and~\eqref{eqn:bound-lower}. An upper bound of $\mu_{\max}(\mathcal{G})$ depends on $\vertiii{G}_\infty$ and $\vertiii{G}_1$, whereas a lower bound of $\mu_{\min}(\mathcal{G})$ involves only the spectral radius of $G$. Combined with~\eqref{eqn:bound:lu}, this suggests that the lower extreme of the spectral density is associated with the average of the maximum weighted in-degree and out-degree of the system, whereas the upper extreme is associated with the stability condition: the {\em less the system is intra- and inter-connected}, the {\em tighter the bound} for the lower extreme will be; similarly, the {\em more stable} (exhibits smaller temporal dependence) the system is, the {\em tighter the bound} for the upper extreme will be. Finally, we note that an upper bound for $(\vertiii{G}_\infty+ \vertiii{G}_1)$ is given by
	\begin{equation*}
	\max\{\vertiii{A}_\infty+\vertiii{B}_\infty ,\vertiii{C}_\infty\} + \max\{\vertiii{A}_1, \vertiii{B}_1 + \vertiii{C}_1 \}.
	\end{equation*}
	The presence of $\vertiii{B}_\infty$ and $\vertiii{B}_1$ depicts the role of the inter-connectedness between $\{X_t\}$ and $\{Z_t\}$ on the lower extreme of the spectrum, which is associated with the overall curvature of the joint process.  
	
	\noindent
	{\em The impact of the system's lower-triangular structure on the system capacity.} With $G$ being a lower-triangular matrix, we only require $\rho(A)<1$ and $\rho(C)<1$ to ensure the stability of the system. This enables the system to have ``larger capacity" (can accommodate more cross-dependence within each block), since the two sparse components $A$ and $C$ can exhibit larger average weighted in- and out-degrees compared with a system where $G$ does not possess such triangular structure. In the case where $G$ is a complete matrix, one deals with a $(p_1+p_2)$-dimensional
	VAR system and $\rho(G)<1$ is required to ensure its stability. As a consequence, the average weighted in- and out-degree requirements for each time series become more restrictive. 
\end{remark}


\section{Testing Group Granger-Causality.}\label{sec:testing}

In this section, we develop a procedure for testing the hypothesis $H_0: B=0$. Note that without the presence of $B$, the blocks $X_t$ and $Z_t$ in the model become {\em decoupled} and can be treated as two separate VAR models, whereas with a nonzero $B$, the group of variables in $Z_t$ is collectively ``Granger-caused" by those in $X_t$. Moreover, since we are testing whether or not the entire block of $B$ is zero, we do not need to rely on the exact distribution of its individual entries, but rather on the properly measured correlation between the responses and the covariates.  To facilitate presentation of the testing procedure, we illustrate the proposed framework via a simpler model setting with $Y_t = \Pi X_t + \epsilon_t$ and testing whether $\Pi=0$; subsequently, we translate the results to the actual setting of interest, namely, whether or not $B=0$ in the model $Z_t = BX_{t-1} + CZ_{t-1}+V_t$. 

The testing procedure focuses on the following sequence of tests for the rank of $B$:
\begin{equation}\label{eqn:hypothesis_lowrank}
H_0: \rank(B)\leq r, \qquad \text{for an arbitrary}\quad r< \min(p_1,p_2).
\end{equation}
Note that the hypothesis of interest,  $B=0$ corresponds to the special case with $r=0$. 
To test for it, we develop a procedure associated with {\em canonical correlations}, which leverages ideas present in the literature \citep[see][]{anderson1999asymptotic}.

As mentioned above, we consider a simpler setting similar to that in \citet{anderson1999asymptotic,anderson2002canonical}, given by 
\begin{equation*}
Y_t = \Pi X_{t} + \epsilon_t,
\end{equation*}
where $Y_t\in\R^{p_2}$, $X\in\mathbb{R}^{p_1}$ and $\epsilon_t$ is independent of $X_t$. At the population level, let 
\begin{equation*}
\E Y_t Y_t' = \Sigma_Y, \qquad \E X_tX_t' = \Sigma_X, \qquad \E Y_{t}X'_{t} = \Sigma_{YX}=\Sigma_{XY}'.
\end{equation*}
The population canonical correlations between $Y_t$ and $X_t$ are the roots of
\begin{equation*}
\left| \begin{matrix}
-\rho\Sigma_Y & \Sigma_{YX} \\ \Sigma_{XY} & -\rho\Sigma_X
\end{matrix}\right| = 0,
\end{equation*} 
i.e., the nonnegative solutions to 
\begin{equation}\label{eqn:cceqn}
|\Sigma_{YX}\Sigma_X^{-1}\Sigma_{XY} - \rho^2\Sigma_Y| = 0,
\end{equation}
with $\rho$ being the unknown. By the results in \citet{anderson1999asymptotic,anderson2002canonical}, the number of positive solutions to~\eqref{eqn:cceqn} is equal to the rank of $\Pi$, and indicates the ``degree of dependency" between processes $Y_t$ and $X_t$. This suggests that if $\rank(\Pi)\leq r<p$, we would expect $\sum_{k=r+1}^p \lambda_k$ to be small, where the $\lambda$'s solve the eigen-equation
\begin{equation*}
|S_{YX}S_X^{-1}S_{XY} - \lambda S_Y| =0, \qquad \text{with}\quad \lambda_1\geq \lambda_2\geq \cdots \geq \lambda_p,
\end{equation*}
and $S_X,S_{XY}$ and $S_Y$ are the sample counterparts corresponding to $\Sigma_X,\Sigma_{XY}$ and $\Sigma_Y$, respectively. \medbreak

With this background, we switch to our model setting given by 
\begin{equation}\label{eqn:model3}
Z_t = BX_{t-1} + CZ_{t-1} + V_t,
\end{equation}
where $V_t$ is assumed to be independent of $X_{t-1}$ and $Z_{t-1}$, $B$ encodes the canonical correlation between $Z_t$ and $X_{t-1}$, conditional on $Z_{t-1}$. We use the same notation as in Section~\ref{sec:theory}; that is,  let $\Gamma_X(h)=\E X_tX_{t+h}'$,  $\Gamma_Z(h)=\E Z_tZ_{t+h}'$, and $\Gamma_{X,Z}(h) = \E X_t Z_{t+h}'$, with $(h)$ omitted whenever $h=0$. At the population level, under the Gaussian assumption,
\begin{equation*}
\begin{bmatrix}
Z_t \\ X_{t-1} \\ Z_{t-1}
\end{bmatrix}\sim \Normal\left( \begin{bmatrix}
0 \\ 0 \\ 0
\end{bmatrix}, \begin{bmatrix}
\Gamma_Z & \Gamma_{X,Z}'(1) & \Gamma_Z(1) \\ \Gamma_{X,Z}(1) & \Gamma_X & \Gamma_{X,Z} \\ \Gamma_Z'(1) & \Gamma'_{X,Z} & \Gamma_Z
\end{bmatrix}\right) ,
\end{equation*}
which suggests that conditionally, 
\begin{equation*}
Z_t\big|Z_{t-1} \sim\Normal\left(\Gamma_{Z}(1)\Gamma_Z^{-1}Z_{t-1} , \Sigma_{00}\right)  \qquad \text{and} \qquad X_{t-1}\big| Z_{t-1}\sim\Normal\left(\Gamma_{X,Z}\Gamma_Z^{-1}Z_{t-1},\Sigma_{11}\right),
\end{equation*}
where
\begin{equation}\label{eqn:Sigma0011}
\Sigma_{00} := \Gamma_Z - \Gamma_Z(1) \Gamma_Z^{-1} \Gamma_Z'(1) \qquad \text{and} \qquad \Sigma_{11}:=\Gamma_X - \Gamma_{X,Z}\Gamma_Z^{-1}\Gamma_{X,Z}'.
\end{equation}
Then, we have that jointly
\begin{equation*}
\begin{bmatrix}
Z_t \\ X_{t-1}
\end{bmatrix}\Big| Z_{t-1} \sim \Normal\left(
\begin{bmatrix} 
\Gamma_{Z}(1) \\ \Gamma_{X,Z}
\end{bmatrix}\Gamma_Z^{-1}Z_{t-1} \;,\; 
\begin{bmatrix} 
\Gamma_Z & \Gamma_{XZ}'(1) \\ \Gamma_{XZ}(1) & \Gamma_Z 
\end{bmatrix} - \begin{bmatrix}
\Gamma_Z(1) \\ \Gamma_{XZ}
\end{bmatrix} \Gamma_Z^{-1} \begin{bmatrix}
\Gamma_Z'(1) & \Gamma_{ZX}
\end{bmatrix}
\right),
\end{equation*}
so the partial covariance matrix between $Z_t$ and $X_{t-1}$ conditional on $Z_{t-1}$ is given by
\begin{equation}\label{eqn:Sigma01}
\Sigma_{10} = \Sigma_{01}' := \Gamma_{X,Z}(1) - \Gamma_Z(1)\Gamma_Z^{-1}\Gamma_{X,Z}.
\end{equation}
The population canonical correlations between $Z_t$ and $X_{t-1}$ conditional on $Z_{t-1}$ are the non-negative roots~of
\begin{equation*}
\left|\Sigma_{01}\Sigma_{11}^{-1}\Sigma_{10}-\rho^2 \Sigma_{00}\right| = 0, 
\end{equation*}
and the number of positive solutions corresponds to the rank of $B$; see \citet{anderson1951estimating} for a discussion in which the author is
interested in estimating and testing linear restrictions on regression coefficients. Therefore, to test $\rank(B)\leq r$, it is appropriate to examine the behavior~of $\Psi_r : = \sum_{k=r+1}^{\min(p_1,p_2)} \phi_k$, where $\phi$'s are ordered non-increasing solutions~to
\begin{equation}\label{eqn:eigeqn}
|S_{01}S_{11}^{-1}S_{10}-\phi S_{00}| = 0,
\end{equation}
and $S_{01}$, $S_{11}$ and $S_{00}$ are the empirical surrogates for the population quantities $\Sigma_{01},\Sigma_{11}$ and $\Sigma_{00}$. 
For subsequent developments, we make the very mild assumption that $p_1<T$ and $p_2<T$ so that $\mathcal{Z}'\mathcal{Z}$ is invertible.

Proposition~\ref{prop:testing-rank} gives the tail behavior of the eigenvalues and Corollary~\ref{cor:granger-causal-test} gives the testing procedure for block ``Granger-causality" as a direct consequence.

\medbreak
\begin{proposition}\label{prop:testing-rank}
	Consider the model setup given in~\eqref{eqn:model3}, where $B\in\R^{p_2\times p_1}$. Further, assume all positive eigenvalues $\mu$ of the following eigen-equation are of algebraic multiplicity one:
	\begin{equation}\label{eqn:popeigen}
	\left| \Sigma_{01}\Sigma_{11}^{-1}\Sigma_{10} - \mu \Sigma_{00}\right| = 0,
	\end{equation}
	where $\Sigma_{00}, \Sigma_{11}$ and $\Sigma_{01}$ are given in~\eqref{eqn:Sigma0011} and~\eqref{eqn:Sigma01}. The test statistic for testing 
	\begin{equation*}
	H_0: \rank(B)\leq r, \qquad \text{for an arbitrary}\quad r < \min(p_1,p_2),
	\end{equation*}  
	is given by 
	\begin{equation*}
	\Psi_r : = \sum_{k=r+1}^{\min(p_1,p_2)} \phi_k,
	\end{equation*}
	where $\phi_k$'s are ordered decreasing solutions to the eigen-equation $|S_{01}S_{11}^{-1}S_{10}-\phi S_{00}| = 0$ where
	\begin{equation*}
	S_{11} = \tfrac{1}{T}\mathcal{X}'(I-P_z)\mathcal{X},~~S_{00} = \tfrac{1}{T} \left(\mathcal{Z}^T\right)' (I-P_z) \left(\mathcal{Z}^T\right),~~S_{10} = S_{01}' = \tfrac{1}{T} \mathcal{X}'(I-P_z) \left(\mathcal{Z}^T\right),
	\end{equation*}
	and $P_z = \mathcal{Z}(\mathcal{Z}'\mathcal{Z})^{-1}\mathcal{Z}'$. Moreover, the limiting behavior of $\Psi_r$ is given by
	\begin{equation*}
	T\Psi_r \sim \chi^2_{(p_1-r)(p_2-r)}.
	\end{equation*}
\end{proposition}
\begin{remark}
	We provide a short comment on the assumption that the positive solutions to \eqref{eqn:popeigen} have algebraic multiplicity one in Proposition~\ref{prop:testing-rank}. This assumption is imposed on the eigen-equation associated with population quantities, to exclude the case where a positive root has algebraic multiplicity greater than one and its geometric multiplicity does not match the algebraic one, and hence we would fail to obtain $r$ mutually independent canonical variates and the rank-$r$ structure becomes degenerate. With the imposed assumption which is common in the canonical correlation analysis literature \citep[e.g.][]{anderson2002canonical,bach2005probabilistic}, such a scenario is automatically excluded. Specifically, this condition is not stringent, as for $\phi$'s that are solutions to the eigen-equation associated with sample quantities, the distinctiveness amongst roots is satisfied with probability 1 \citep[see][Proof of Lemma 3]{hsu1941problem}. 
\end{remark}

\begin{corollary}[Testing group Granger-causality]\label{cor:granger-causal-test}
	Under the model setup in \eqref{eqn:model3}, the test statistic for testing $B=0$ is given by 
	\begin{equation*}
	\Psi_0:=\sum_{k=1}^{\min(p_1,p_2)}\phi_k,
	\end{equation*}
	with $\phi_k$ being the ordered decreasing solutions of
	\begin{equation*}
	\Big|S_{01}\big[ \mathrm{diag}(S_{11})\big]^{-1}S_{10}-\phi S_{00}\Big| = 0.
	\end{equation*}
	Asymptotically, $T\Psi_0\sim \chi^2_{p_1p_2}$. To conduct a level $\alpha$ test, we reject the null hypothesis $H_0:B=0$ if
	\begin{equation*}
	\Psi_0 > \frac{1}{T}\chi^2_{p_1p_2}(\alpha), 
	\end{equation*}
	where $\chi^2_d(\alpha)$ is the upper $\alpha$ quantile of the $\chi^2$ distribution with $d$ degrees of freedom.
\end{corollary}
\begin{remark}
	Corollary~\ref{cor:granger-causal-test} is a special case of Proposition~\ref{prop:testing-rank} with the null hypothesis being $H_0:r=0$, which corresponds to the Granger-causality test. Under this particular setting, we are able to take the inverse with respect to $\text{diag}(S_{11})$, yet maintain the same asymptotic distribution due to the fact that $S_{01}=S_{10}=0$ under the null hypothesis $B=0$. This enables us to perform the test even with $p_1>T$. 
\end{remark}

The above testing procedure takes advantage of the fact that when $B=0$, the canonical correlations among the partial regression residuals after removing the effect of $Z_{t-1}$ are very close to zero. However, the test may not be as powerful under a sparse alternative, i.e., $H_A: B$ is sparse. In Appendix~\ref{appendix:sparse-testing}, we present a testing procedure that specifically takes into consideration the fact that the alternative hypothesis is sparse, and the corresponding performance evaluation is shown in Section~\ref{sec:testing-sim} under this setting. 

%
%

\section{Performance Evaluation.}\label{sec:simulation}

Next, we present the results of numerical studies to evaluate the performance of the developed ML estimates (Section~\ref{sec:estimation}) of the model parameters, as well as that of the testing procedure (Section~\ref{sec:testing}). 

\subsection{Simulation results for the estimation procedure.}\label{sec:performance}

A number of factors may potentially influence the performance of the estimation procedure; in particular, the model dimension $p_1$ and $p_2$, the sample size $T$, the rank of $B^\star$ and the sparsity level of $A^\star$ and $C^\star$, as well as the spectral radius of $A^\star$ and $C^\star$. Hence, we consider several settings where these parameters vary. 

\medbreak
For all settings, the data $\{x_t\}_t$ and $\{z_t\}_t$ are generated according to the model
\begin{equation*}
\begin{split}
x_t & = A^\star x_{t-1} + u_t, \\
z_t & = B^\star x_{t-1} + C^\star z_{t-1} + v_t.
\end{split}
\end{equation*}
For the sparse components, each entry in $A^\star$ and $C^\star$ is nonzero with probability $2/p_1$ and $1/p_2$ respectively, and the nonzero entries are generated from $\mathsf{Unif}\left([-2.5,-1.5]\cup [1.5,2.5]\right)$, then scaled down so that the spectral radii $\rho(A)$ and $\rho(C)$ satisfy the stability condition. For the low rank component, each entry in $B^\star$ is generated from $\mathsf{Unif}(-10,10)$, followed by singular value thresholding, so that $\rank(B^\star)$ conforms with the model setup. For the contemporaneous dependence encoded by $\Omega_u^\star$ and $\Omega_v^\star$, both matrices are generated according to an Erd\"{o}s-R\'{e}nyi random graph, with sparsity being 0.05 and condition number being 3. 

Table~\ref{table:model} depicts the values of model parameters under different model settings. Specifically, we consider three major settings in which the size of the system, the rank of the cross-dependence component, and the stability of the system vary. The sample size is fixed at $T=200$ unless otherwise specified. Additional settings examined (not reported due to space considerations) are consistent with the main conclusions presented next.

\begin{table}[!h]
	\centering
	\caption{Model parameters under different model settings}\label{table:model}
	\begin{tabular}{p{3.5cm}|c|ccccc}
		\specialrule{.12em}{.05em}{.05em} 
		& & \multicolumn{5}{c}{\color{blue} model parameters}  \\ 
		& & $p_1$ & $p_2$ & $\text{rank}(B^*)$ & $\rho_A$ & $\rho_C $   \\ \hline
		\multirow{4}{*}{\parbox{3cm}{model dimension}} 
		&  A.1 & 50 & 20 & 5 & 0.5 & 0.5  \\ 
		& A.2 & 100 & 50 & 5 & 0.5 & 0.5 \\
		& A.3 & 200 & 50 & 5 & 0.5 & 0.5  \\ 
		& A.4 & 50 & 100 & 5 & 0.5 & 0.5 \\ \hline 
		\multirow{ 2}{*}{\parbox{3cm}{rank }} 	
		& B.1 & 100 & 50  & 10 & 0.5 & 0.5   \\ 
		& B.2 & 100 & 50 & 20 & 0.5 & 0.5 \\ \hline 
		\multirow{3}{*}{\parbox{3cm}{spectral radius}}  
		& C.1 & 50 & 20 & 5 & 0.8 & 0.5   \\
		&C.2 & 50 & 20 &  5 & 0.5 & 0.8 \\ 
		& C.3 & 50 & 20 &  5 & 0.8 & 0.8 \\ 
		\specialrule{.12em}{.05em}{.05em} 
	\end{tabular}
\end{table}

We use sensitivity (SEN), specificity (SPC) and relative error in Frobenius norm (Error) as criteria to evaluate the performance of the estimates of transition matrices $A$, $B$ and $C$. Tuning parameters are chosen based on BIC. Since the exact contemporaneous dependence is not of primary concern, we omit the numerical results for $\widehat{\Omega}_u$ and $\widehat{\Omega}_v$.
\begin{equation*}
\textrm{SEN} = \frac{\textrm{TP}}{\textrm{TP}+\textrm{FN}},\quad \textrm{SPE} = \frac{\textrm{TN}}{\textrm{FP}+\textrm{TN}}, \quad \textrm{Error} = \frac{\smalliii{\text{Est.}-\text{Truth}}_F}{\smalliii{\text{Truth}}_F}.
\end{equation*}
Table~\ref{table:simresult1} illustrates the performance for each of the parameters under different simulation settings
considered. The results are based on an average of 100 replications and their standard deviations are given in parentheses.
\begin{table}[!h]
	\centering
	\caption{Performance evaluation of $\widehat{A}$, $\widehat{B}$ and $\widehat{C}$ under different model settings.}\label{table:simresult1}
		\begin{tabular}{p{0.5cm}|ccc|cc|ccc}
			\specialrule{.12em}{.05em}{.05em} 
			& \multicolumn{3}{c|}{\color{blue} performance of $\widehat{A}$}  & \multicolumn{2}{c|}{\color{blue} performance of $\widehat{B}$} & \multicolumn{3}{c}{\color{blue} performance of $\widehat{C}$} \\ 
			&  SEN & SPC & Error & $\text{rank}(\widehat{B})$ & Error & SEN & SPC & Error  \\ 
			\specialrule{.08em}{.05em}{.05em}
			A.1& 0.98{\scriptsize(0.014)} & 0.99{\scriptsize(0.004)} & 0.34{\scriptsize(0.032)} & 5.2{\scriptsize(0.42)} & 0.11{\scriptsize(0.008)} & 1.00{\scriptsize(0.000)} & 0.97{\scriptsize(0.008)} & 0.15{\scriptsize(0.074)} \\ 
			A.2& 0.97{\scriptsize(0.014)} & 0.99{\scriptsize(0.001)} & 0.38{\scriptsize(0.015)} & 5.2{\scriptsize(0.42)} &  0.31{\scriptsize(0.011)} & 0.97{\scriptsize(0.008)} & 0.97{\scriptsize(0.004)} & 0.28{\scriptsize(0.033)}\\
			A.3&  0.99{\scriptsize(0.005)} & 0.96{\scriptsize(0.002)} & 0.87{\scriptsize(0.011)} & 5.8{\scriptsize(0.92)} &  0.54{\scriptsize(0.022)} & 0.98{\scriptsize(0.000)} & 0.92{\scriptsize(0.009)} & 0.28{\scriptsize(0.028)} \\  
			A.4& 0.96{\scriptsize(0.0261)} & 0.99{\scriptsize(0.002)} & 0.36{\scriptsize(0.034)} & 5.2{\scriptsize(0.42)} & 0.32{\scriptsize(0.012)} & 0.95{\scriptsize(0.009)} & 0.98{\scriptsize(0.001)} & 0.37{\scriptsize(0.010)} \\  
			\specialrule{.08em}{.05em}{.05em}
			B.1	& 0.97{\scriptsize(0.008)} & 0.99{\scriptsize(0.001)} & 0.37{\scriptsize(0.017)}  & 11.4{\scriptsize(1.17)} & 0.15{\scriptsize(0.008)} & 1.00{\scriptsize(0.000)} & 0.99{\scriptsize(0.001)} & 0.09{\scriptsize(0.021)} \\ 
			B.2	& 0.98{\scriptsize(0.008)} & 0.99{\scriptsize(0.001)} & 0.38{\scriptsize(0.016)} & 21.2{\scriptsize(0.91)} & 0.12{\scriptsize(0.006)} & 1.00{\scriptsize(0.000)} & 0.99{\scriptsize(0.001)} & 0.08{\scriptsize(0.018)} \\ 
			\specialrule{.08em}{.05em}{.05em}
			C.1	& 1.00{\scriptsize (0.000)} & 0.97{\scriptsize(0.005)} & 0.25{\scriptsize(0.015)} &  5.6{\scriptsize(0.52)} & 0.23{\scriptsize(0.006)} & 1.00{\scriptsize(0.000)} & 0.92{\scriptsize(0.021)} & 0.11{\scriptsize(0.072)} \\
			C.2	 &  0.99{\scriptsize(0.007)} & 0.95{\scriptsize(0.004)} & 0.45{\scriptsize(0.022)} &  5.0{\scriptsize(0.00)} & 0.31{\scriptsize(0.014)} & 1.00{\scriptsize(0.000)} & 0.92{\scriptsize(0.019)} & 0.04{\scriptsize(0.011)}\\ 
			C.3	&  1.00{\scriptsize(0.000)} & 0.96{\scriptsize(0.004)} & 0.18{\scriptsize(0.013)}& 6.7{\scriptsize(1.16)} & 0.19{\scriptsize(0.011)} & 1.00{\scriptsize(0.000)} & 0.87{\scriptsize(0.029)} & 0.14{\scriptsize(0.067)} \\
			C.3'&  1.00{\scriptsize(0.000)} & 0.99{\scriptsize(0.002)} & 0.13{\scriptsize(0.016)}& 5.2{\scriptsize(0.42)} & 0.23{\scriptsize(0.005)} & 1.00{\scriptsize(0.000)} & 0.90{\scriptsize(0.021)} & 0.06{\scriptsize(0.023)} \\
			\specialrule{.12em}{.05em}{.05em} 
		\end{tabular}
\end{table}
Overall, the results are highly satisfactory and all the parameters are estimated with a high degree of accuracy. Further,
all estimates were obtained in less than 20 iterations, thus indicating that the estimation procedure is numerically stable. As expected, when the the spectral radii of $A$ and $C$ increase thus leading to less stable $\{X_t\}$ and $\{Z_t\}$ processes, a larger sample size is required for the estimation procedure to match the performance of the setting with same parameters but smaller $\rho(A)$ and $\rho(C)$. This is illustrated in row C.3' of Table~\ref{table:simresult1}, where the sample size is increased to $T=500$, which outperforms the results in row C.3 in which $T=200$ and broadly matches that of row A.1. 

Next, we investigate the robustness of the algorithm in settings where the marginal distributions of $\{X_t\}$ and $\{Z_t\}$ deviate from the
Gaussian assumption posited and may be more heavy-tailed. Specifically, we consider the following two distributions that have been studied in \citet{qiu2015robust}:
\begin{itemize}
	\item $t$-distribution: the idiosyncratic error processes $\{u_t\}$ and $\{v_t\}$ are generated from multivariate $t$-distributions with degree of freedom 3, and covariance matrices $(\Omega^\star_u)^{-1}$ and $(\Omega^\star_v)^{-1}$, respectively. 
	\item elliptical distribution: $(u_1',\dots,u_T')$ and $(v_1',\dots,v_T')'$ are generated from an elliptical distribution \citep[e.g.][]{remillard2012copula} with a log-normal generating variate $\log \mathcal{N}(0,2)$ and covariance matrices $\widetilde{\Sigma}_u$ and $\widetilde{\Sigma}_v$ -- both are block-diagonal with $\Sigma_u^\star=(\Omega_u^\star)^{-1}$ and $\Sigma_v^\star=(\Omega_v^\star)^{-1}$ respectively on the diagonals. 
\end{itemize}
For both scenarios, transition matrices, $\Omega_u^\star$ and $\Omega_v^\star$ are generated analogously to those in the Gaussian setting. We present the results for $\widehat{A},\widehat{B}$ and $\widehat{C}$ under model settings A.2, B.1, C.1 and C.2 (see Table~\ref{table:model}). 
\begin{table}[!h]
	\centering
	\caption{Performance evaluation of $\widehat{A}$, $\widehat{B}$ and $\widehat{C}$ under non-Gaussian settings.}\label{table:non-Gaussian}
		\begin{tabular}{ll|ccc|cc|ccc}
			\specialrule{.12em}{.05em}{.05em} 
			&	& \multicolumn{3}{c|}{\color{blue} performance of $\widehat{A}$}  & \multicolumn{2}{c|}{\color{blue} performance of $\widehat{B}$} & \multicolumn{3}{c}{\color{blue} performance of $\widehat{C}$} \\ 
			& &  SEN & SPC & Error & $\text{rank}(\widehat{B})$ & Error & SEN & SPC & Error  \\ 
			\specialrule{.08em}{.05em}{.05em}
			\multirow{2}{*}{A.2} & t(df=3) & 0.99\scriptsize{(0.005)} & 0.95\scriptsize{(0.013)} & 0.60\scriptsize{(0.062)} & 6.00\scriptsize{(1.45)} & 0.24\scriptsize{(0.019)} & 0.96\scriptsize{(0.013)} & 0.96\scriptsize{(0.005)} & 0.27\scriptsize{(0.038)} \\
			& elliptical & 0.97\scriptsize{(0.014)} & 0.99\scriptsize{(0.001)} & 0.36\scriptsize{(0.016)} & 5.1\scriptsize{(0.30)} & 0.34\scriptsize{(0.009)} & 1.00\scriptsize{(0.000)} & 0.85\scriptsize{(0.026)} & 0.15\scriptsize{(0.033)}\\ 
			\specialrule{.12em}{.05em}{.05em} 
			\multirow{2}{*}{B.1}& t(df=3) & 0.98\scriptsize{(0.008)} & 0.95\scriptsize{(0.014)} & 0.61\scriptsize{(0.083)} & 10.4\scriptsize{(0.49)} & 0.34\scriptsize{(0.026)} &0.99\scriptsize{(0.015)} & 0.95\scriptsize{(0.004)} & 0.25\scriptsize{(0.091)} \\
			& elliptical &  0.95\scriptsize{(0.015)} & 0.99\scriptsize{(0.001)} & 0.37\scriptsize{(0.024)} & 10.1\scriptsize{(0.22)} & 0.40\scriptsize{(0.013)} & 1.00\scriptsize{(0.000)} & 0.90\scriptsize{(0.013)} & 0.09\scriptsize{(0.001)}\\
			\specialrule{.12em}{.05em}{.05em} 
			\multirow{2}{*}{C.1}& t(df=3) & 0.99\scriptsize{(0.001)} & 0.92\scriptsize{(0.011)} & 0.22\scriptsize{(0.03)} & 6.0\scriptsize{(1.13)} & 0.09\scriptsize{(0.014} &1.00\scriptsize{(0.000)} & 0.93\scriptsize{(0.016)} & 0.10\scriptsize{(0.068)} \\
			& elliptical & 1.00\scriptsize{(0.000)} & 0.90\scriptsize{(0.006} & 0.32\scriptsize{(0.013)} & 5.2\scriptsize{(0.44)} & 0.13\scriptsize{(0.007)} & 1.00\scriptsize{(0.000)} & 0.92\scriptsize{(0.020)} & 0.07\scriptsize{(0.041)} \\ \hline
			\multirow{2}{*}{C.2} & t(df=3) & 0.99 \scriptsize{(0.002)} & 0.95\scriptsize{(0.023)} & 0.37\scriptsize{(0.056)} & 5.1\scriptsize{(0.22)} & 0.22\scriptsize{(0.017)} & 1.00\scriptsize{(0.000)} & 0.89\scriptsize{(0.017)} & 0.10\scriptsize{(0.029)} \\
			& elliptical & 0.88\scriptsize{(0.029)} & 0.97\scriptsize{(0.001)} & 0.43\scriptsize{(0.032)} & 5.1\scriptsize{(0.14)} & 0.40\scriptsize{(0.020)} & 1.00\scriptsize{(0.000)} & 0.86\scriptsize{(0.026)} & 0.10\scriptsize{(0.046)} \\
			\specialrule{.12em}{.05em}{.05em} 
		\end{tabular}
\end{table}

Based on Table~\ref{table:non-Gaussian}, the performance of the estimates under these heavy-tailed settings is comparable in terms of sensitivity
and specificity for $A$ and $C$, as well as for rank selection for $B$ to those under Gaussian settings. However, the estimation error
exhibits some deterioration which is more pronounced for the $t$-distribution case. In summary, the estimation procedure proves to be very
robust for support recovery and rank estimation even in the presence of more heavy-tailed noise terms.

Lastly, we examine performance with respect to one-step-ahead forecasting. Recall that VAR models are widely used for forecasting purposes
in many application areas \citep{lutkepohl2005new}. The performance metric is given by the relative error as measured by the $\ell_2$ norm of the out-of-sample points $x_{T+1}$ and $z_{T+1}$, where the predicted values are given by \mbox{$\widehat{x}_{T+1}=\widehat{A}x_{T}$} and \mbox{$\widehat{z}_{T+1}=\widehat{B}x_{T} + \widehat{C}z_{T}$}, respectively. It is worth noting that both $\{X_t\}$ and $\{Z_t\}$ are mean-zero processes. However, since the transition matrix of $\{X_t\}$ is subject to the spectral radius constraints to ensure the stability of the corresponding process, the magnitude of the realized value  $x_t$'s is small; whereas for $\{Z_t\}$, since no constraints are imposed on the $B$ coefficient matrix that encodes the inter-dependence, $z_t$'s has the capacity of having relative large values in magnitude. Consequently, the relative error of $\widehat{x}_{T+1}$ is significantly larger than that of $\widehat{z}_{T+1}$, partially due to the small total magnitude of the denominator. 

The results show that an increase in the spectral radius (keeping the other structural parameters fixed) leads to a decrease of the relative
error, since future observations become more strongly correlated over time. On the other hand, an increase in dimension leads to a deterioration in forecasting, since the available sample size impacts the quality of the parameter estimates. Finally, an increase in the rank of the $B$ matrix is beneficial for forecasting, since it plays a stronger role in the system's temporal evolution.

\begin{table}[!h]
	\centering
	\caption{One-step-ahead relative forecasting error.}\label{table:rel-forecasting}
	\begin{tabular}{l|c|cc}
		\specialrule{.12em}{.05em}{.05em}
		&     & $\tfrac{\|\widehat{x}_{T+1}-x_{T+1}\|_2}{\|x_{T+1}\|_2}$ & $\tfrac{\|\widehat{z}_{T+1}-z_{T+1}\|_2}{\|z_{T+1}\|_2}$ \\
		\hline 
		baseline &  A.1 & 0.89{\scriptsize(0.066)} & 0.23{\scriptsize(0.075)} \\ \hline
		\multirow{3}{*}{spectral radius} & C.1 & 0.62{\scriptsize(0.100)} & 0.10{\scriptsize(0.035)}\\
		& C.2 & 0.93{\scriptsize(0.062)} & 0.17{\scriptsize(0.059)} \\
		& C.3 & 0.68{\scriptsize(0.096)} &  0.10{\scriptsize(0.045)}\\ \hline
		\multirow{2}{*}{rank} & B.1 & 0.92{\scriptsize(0.044)}& 0.14{\scriptsize(0.038)} \\
		& B.2 & 0.94{\scriptsize(0.042)} & 0.14{\scriptsize(0.025)} \\ \hline
		\multirow{3}{*}{dimension} & A.2 & 0.87{\scriptsize(0.051)} &0.24{\scriptsize(0.073)}\\
		& A.3 & 0.96{\scriptsize(0.040)} & 0.44{\scriptsize(0.139)}\\
		& A.4 & 0.89{\scriptsize(0.059)} & 0.274{\scriptsize(0.068)}\\
		\specialrule{.12em}{.05em}{.05em}
	\end{tabular}
\end{table}

\subsection{A comparison between the two-step and the ML estimates.}

We briefly compare the ML estimates to the ones obtained through the following two-step procedure:
\begin{itemize}
	\item[--] Step 1: estimate transition matrices through penalized least squares: 
	\begin{align*}
	\widehat{A}^{\text{t-s}} &= \argmin\limits_{A} \Big\{ \tfrac{1}{T}\vertiii{\mathcal{X}^T - \mathcal{X}A'}_F^2 + \lambda_A \|A\|_1 \Big\},  \\
	(\widehat{B}^{\text{t-s}},\widehat{C}^{\text{t-s}}) &= \argmin\limits_{(B,C)} \Big\{ \tfrac{1}{T}\vertiii{\mathcal{Z}^T - \mathcal{X}B'-\mathcal{Z}C'}_F^2 + \lambda_B \vertiii{B}_* + \lambda_C \vertii{C}_1  \Big\}.
	\end{align*} 
	\item[--] Step 2: estimate the inverse covariance matrices applying the graphical Lasso algorithm \citep{friedman2008sparse} to the residuals
	calculated based on the Step 1 estimates.
\end{itemize}
Note that the two-step estimates coincide with our ML estimates at iteration 0, and they yield the same error rate in terms of the relative scale of $p_1,p_2$ and $T$. We compare the two sets of estimates under setting A.1 with $B^\star$ being low rank and setting A.2 with $B^\star$ being sparse, whose entries are nonzero with probability $1/p_1$. 

In Table~\ref{table:performance}, we present the performance evaluation of the two-step estimates and the ML estimates under setting A.1. Additionally in Tables~\ref{table:errorA} and~\ref{table:errorBC}, we track the value of the objective function, the relative error (in $\|\cdot\|_F$) and the cardinality (or rank) of the estimates along iterations, with iteration 0 corresponding to the two-step estimates.  A similar set of results is shown in Tables~\ref{table:performance-A2-sparse} to~\ref{table:sparseB} for setting A.2, but with a sparse $B^\star$. All other model parameters are identically generated according to the procedure described in Section~\ref{sec:performance}.

As the results show, the ML estimates clearly outperform their two-step counterparts, in terms of the relative error in Frobenius norm.
On the other hand, both sets of estimates exhibit similar performance in terms of sensitivity and specificity and rank specification. More specifically, when estimating $A$, the ML estimate decreases the false positive rate (higher SPC). Under setting A.1, while estimating $B$ and $C$, both estimates correctly identify the rank of $B$, and the ML estimate provides a more accurate estimate in terms of the magnitude of $C$, at the expense of incorrectly including a few more entries in its support set; under setting A.2 with a sparse $B^\star$, improvements in both the relative error of $B$ and $C$ are observed. In particular, due to the descent nature of the algorithm, we observe a sharp drop in the value of the objective function at iteration 1, as well as the most pronounced change in the estimates. 

\begin{table}[h]
	\small
	\centering
	\caption{Performance comparison under A.1 with a low-rank $B$}\label{table:performance}
	\begin{tabular}{p{4cm}|ccc|cc|ccc}
		\specialrule{.12em}{.05em}{.05em} 
		& \multicolumn{3}{c|}{\color{blue} performance of $\widehat{A}$}  & \multicolumn{2}{c|}{\color{blue} performance of $\widehat{B}$} & \multicolumn{3}{c}{\color{blue} performance of $\widehat{C}$} \\ 
		&  SEN & SPC & Error & Error & $\text{rank}(\widehat{B})$ & SEN & SPC & Error  \\ 
		\specialrule{.07em}{.05em}{.05em}
		two-step estimates & 0.97 & 0.95 & 0.52 & 0.27 & 5 &  1.00 & 0.98 & 0.12 \\ 
		ML estimates & 0.97 & 0.97 & 0.36 & 0.24 & 5 &  1.00 & 0.95 & 0.05\\
		\specialrule{.12em}{.05em}{.05em} 
	\end{tabular}\bigskip
	
	\caption{Relative error of $\widehat{A}$ and the values of the objective function under A.1}\label{table:errorA}
	\begin{tabular}{l|cccccc}
		\specialrule{.12em}{.05em}{.05em} 
		iteration & 0 & 1 & 2 & 3 & 4 & 5\\ \hline
		Rel.Error & 0.521 & 0.408 &  0.376 &  0.360 & 0.359 & 0.359 \\ 
		Cardinality & 227 & 169 & 160 & 155 & 155 & 155 \\ 
		Value of Obj & 128.14 & 41.74 & 37.94 & 37.85 & 37.70 & 37.70 \\ 
		\specialrule{.12em}{.05em}{.05em} 
	\end{tabular}\bigskip
	
	\caption{Relative error of $\widehat{B}$,$\widehat{C}$ and the values of the objective function under A.1}\label{table:errorBC}
		\begin{tabular}{l|ccccccccccc}
			\specialrule{.12em}{.05em}{.05em} 
			iteration & 0 & 1 & 2 & 3 & 4 & 5 & 6 & 7 & 8 & 9 & 10 \\ \hline
			Rel.Error of $\widehat{B}$ & 0.274 & 0.235 &  0.236 &  0.237 & 0.237 & 0.237 & 0.237& 0.237& 0.237& 0.237& 0.237 \\ 
			Rank of $\widehat{B}$ & 5 & 5 &5 &5 &5 &5 &5 &5 &5 &5 &5  \\ \hline
			Rel.Error of $\widehat{C}$ & 0.119 & 0.050 & 0.049 & 0.049 & 0.048 & 0.048 & 0.048 & 0.047 & 0.047 & 0.047 & 0.047   \\ 
			Cardinality of $\widehat{C}$ &  30 & 34 & 38 & 41 & 39 & 39 & 39& 39& 39& 39& 39 \\ \hline
			Value of Obj & 160.41 & 134.26 & 131.90 & 131.48 & 131.38 & 131.17 & 131.01 & 130.96 & 130.96 & 130.85 & 130.8 \\ 
			\specialrule{.12em}{.05em}{.05em} 
		\end{tabular}
\end{table}

\begin{table}[h]
	\centering
	\caption{Performance comparison under A.2 with a sparse $B$}\label{table:performance-A2-sparse}
	\begin{tabular}{p{4cm}|ccc|ccc|ccc}
		\specialrule{.12em}{.05em}{.05em} 
		&  \multicolumn{3}{c|}{\color{blue} performance of $\widehat{A}$} & \multicolumn{3}{c|}{\color{blue} performance of $\widehat{B}$} & \multicolumn{3}{c}{\color{blue} performance of $\widehat{C}$} \\ 
		& SEN & SPC &   Error & SEN & SPC &   Error  & SEN & SPC & Error  \\ 
		\specialrule{.07em}{.05em}{.05em}
		two-step estimates&  0.97 & 0.95 & 0.44 & 0.96 & 0.98 & 0.45 &  1 & 0.99 & 0.35 \\ 
		ML estimates &  0.97 & 0.98 & 0.35 & 0.99 & 0.95 & 0.34 & 1 & 0.98 & 0.30\\
		\specialrule{.12em}{.05em}{.05em} 
	\end{tabular}\bigskip
	
	\caption{Relative error of $\widehat{A}$ and the values of the objective function under A.2}\label{table:errorA-A2-sparse}
	\begin{tabular}{l|ccccc}
		\specialrule{.12em}{.05em}{.05em} 
		iteration & 0 & 1 & 2 & 3 & 4 \\ \hline
		Rel.Error & 0.438 & 0.346& 0.351& 0.351 &0.351  \\ 
		Cardinality & 479 & 350 & 325 & 324 & 324  \\ 
		Value of Obj & 156.58 & 48.25 &  38.16 & 36.97 & 36.93 \\ 
		\specialrule{.12em}{.05em}{.05em} 
	\end{tabular}\bigskip

	\caption{Changes over iteration under A.2}\label{table:sparseB}
	\begin{tabular}{l|ccccc}
		\specialrule{.12em}{.05em}{.05em} 
		iteration & 0  &1 & 2 & 3 & 4  \\ \hline
		Rel.Error of $\widehat{B}$ & 0.454 & 0.340 & 0.337 & 0.337 & 0.337  \\ 
		Cardinality of $\widehat{B}$ &  301 & 325 & 323 &323 & 323 \\ \hline
		Rel.Error of $\widehat{C}$ & 0.35&  0.304& 0.302 &0.302& 0.301\\ 
		Cardinality of $\widehat{C}$ &  63 & 70 & 74 & 74 & 74 \\ \hline
		Value of Obj & 143.942 & 59.63 & 41.46 & 41.87 & 41.87 \\ 
		\specialrule{.12em}{.05em}{.05em} 
	\end{tabular}
\end{table}
%
%
\subsection{Simulation results for the block Granger-causality test.}\label{sec:testing-sim}

Next, we illustrate the empirical performance of the testing procedure introduced in Section~\ref{sec:testing}, together with the alternative one (in Appendix~\ref{appendix:sparse-testing})
when $B$ is sparse, with the null hypothesis being $B^\star=0$ and the alternative being $B^\star\neq 0$, either low rank or sparse. Specifically, when the alternative hypothesis is true and has a low-rank structure, we use the general testing procedure proposed in~Section~\ref{sec:testing}, whereas when the alternative is true and sparse, we use the testing procedure presented in Appendix~\ref{appendix:sparse-testing}.  We focus on evaluating the type I error (empirical false rejection rate) when $B^\star=0$, as well as the power of the test when $B^\star$ has nonzero entries. 

\smallbreak
For both testing procedures, the transition matrix $A^\star$ is generated with each entry being nonzero with probability $2/p_1$, and the nonzeros are generated from $\mathsf{Unif}\left([-2.5,-1.5]\cup [1.5,2.5]\right)$, then further scaled down so that $\rho(A^\star)=0.5$. For transition matrix $C^\star$, each entry is nonzero with probability $1/p_2$, and the nonzeros are generated from $\mathsf{Unif}\left([-2.5,-1.5]\cup [1.5,2.5]\right)$, then further scaled down so that $\rho(C^\star)=0.5$ or $0.8$, depending on the simulation setting. Finally, we only consider the case where $v_t$ and $u_t$ have diagonal covariance matrices.

We use sub-sampling as in \citet{politis1994large} and \citet{politis1999subsampling} with the number of subsamples set to 3,000;
an alternative would have been a block bootstrap procedure \citep[e.g.][]{hall1985resampling,carlstein1986use,kunsch1989jackknife}. Note that the length of the subsamples varies across simulation settings in order to gain insight on how sample size impacts the type I error or the power of the test. 

\paragraph{Low-rank testing.} To evaluate the type I error control and the power of the test, we primarily consider the case where $\rank(B^\star) =0$, with the data alternatively generated based on $\rank(B^\star)=1$. We test the hypothesis $H_0:\rank(B)=0$ and tabulate the empirical proportion of falsely rejecting $H_0$ when $\rank(B^*)=0$ (type I error) and the probability that we reject $H_0$ when $\rank(B^*)=1$ (power). In addition, we also show how the testing procedure performs when the underlying $B^\star$ has rank $r\geq 0$. In particular, when $\rank(B^\star)=r^\star$, the type I error of the test corresponds to the empirical proportion of rejections of the null hypothesis $H_0:r\leq r^\star$, while the power of the test to the empirical proportion of rejections of the null hypothesis set to $H_0:r\leq (r^\star-1)$. The latter resembles the sequential test in~\citet{johansen1988statistical}. 

Empirically, we expect that when $B^\star=0$, the value of the proposed test statistic mostly falls below the cut-off value (upper $\alpha$ quantile), while when $\rank(B^\star)=1$, the value of the proposed test statistic mostly falls beyond the critical value $\chi^2(\alpha)_{p_1p_2}/T$ with $T$ being the sample size, hence leading to a detection. Table~\ref{table:lowrank} gives the type I error of the test when setting $\alpha=0.1,0.05,0.1$, and the power of the test using the upper 0.01 quantile of the reference distribution as the cut-off, for different combinations of model dimensions $(p_1,p_2)$ and sample size. 

\begin{table}[!h]
	\centering
	\caption{Empirical type I error and power for low-rank testing}\label{table:lowrank}
		\begin{tabular}{c|c|l|ccc|c}
			\specialrule{.1em}{.05em}{.05em} 
			&	&	& \multicolumn{3}{c|}{\color{blue} type I error ($B^\star=0$)}  &  {\color{blue}power ($\rank(B^\star)=1$)} \\ 
			& $(p_1,p_2)$	&  sample size  & $\alpha=0.01$ & $\alpha=0.05$ & $\alpha=0.1$ & cut-off $\chi^2(0.01)_{p_1p_2}/T$  \\  
			\specialrule{.05em}{.05em}{.05em}
			\multirow{12}{*}{$\rho(C^\star)=0.5$} &\multirow{3}{*}{$(20,20)$} & $T=500$ & 0.028 & 0.123 & 0.227 & 1 \\
			& 	  & $T=1000$ & 0.015 & 0.073 & 0.137 & 1 \\
			& 	  & $T=2000$ & 0.011 & 0.059 & 0.118 & 1 \\ 
			\cline{2-7}
			&\multirow{3}{*}{$(50,20)$} & $T=500$ &	0.070 & 0.228 & 0.355 & 1 \\
			&     & $T=1000$ & 0.026 & 0.125 & 0.226 & 1 \\
			&     & $T=2000$ & 0.013 & 0.094 & 0.163 & 1 \\
			\cline{2-7}
			&\multirow{3}{*}{$(20,50)$} & $T=500$ & 0.484 & 0.751 & 0.857 & 1 \\
			&     & $T=1000$ & 0.089 & 0.246 & 0.375 & 1\\
			&     & $T=2000$ & 0.020 & 0.088 & 0.164 & 1 \\
			\cline{2-7}
			&\multirow{3}{*}{$(100,50)$} & $T=500$ & \underline{0.997} & \underline{0.999} & \underline{1} & 1 \\ 
			&  & $T=1000$ & 0.608 & 0.828 & 0.908 & 1 \\
			&  & $T=2000$ & 0.166 & 0.374 & 0.511 & 1 \\
			\hline
			\multirow{6}{*}{$\rho(C^\star)=0.8$} &\multirow{3}{*}{$(20,50)$} & $T=500$ & 0.533 & 0.789 & 0.880 & 1 \\
			&  & $T=1000$ & 0.130 & 0.306 & 0.452 & 1 \\
			&  & $T=2000$ & 0.045 & 0.145 & 0.252 & 1 \\
			\cline{2-7}	
			& \multirow{3}{*}{$(50,20)$} & $T=500$ & 0.083 & 0.250 & 0.382  & 1\\
			&  & $T=1000$ & 0.039 & 0.133 & 0.234 & 1 \\
			&  & $T=2000$ & 0.019 & 0.096 & 0.174 & 1 \\  			    
			\specialrule{.1em}{.05em}{.05em} 
			&       &   &  \multicolumn{3}{c|}{\textcolor{blue}{type I error $(H_0:r\leq 5)$} }& \textcolor{blue}{power $(H_0:r\leq 4)$} \\  
			&       &   & $\alpha=0.01$ & $\alpha=0.05$ & $\alpha=0.1$ & cut-off $\chi^2(0.01)_{(p_1-4)(p_2-4)}/T$ \\
			\cline{2-7}
			\multirow{7}{*}{\shortstack{$\rho(C^\star)=0.5$\\$\rank(B^\star)=5$}} &  \multirow{3}{*}{$(20,50)$} & $T=500$ & 0.092 &  0.274 & 0.400 & 1 \\
			& & $T=1000$ & 0.034 & 0.140 & 0.236 & 1 \\
			& & $T=2000$ & 0.022 & 0.096 & 0.178 & 1 \\ 
			\cline{2-7}
			& \multirow{3}{*}{$(50,20)$} & $T=500$ & 0.454 & 0.722   & 0.829  & 1 \\
			& & $T=1000$ & 0.126 & 0.313 & 0.452 & 1 \\
			& & $T=2000$ & 0.062 & 0.184 & 0.284 & 1  \\ 
			\specialrule{.1em}{.05em}{.05em}
		\end{tabular}
\end{table}
Based on the results shown in Table~\ref{table:lowrank}, it can be concluded that the proposed low-rank testing procedure 
accurately detects the presence of ``Granger causality" across the two blocks, when the data have been generated based on a truly multi-layer VAR system. Further, when $B^\star=0$, the type I error is close to  the nominal $\alpha$ level for sufficiently large sample sizes, but deteriorates for increased model dimensions. In particular, relatively large values of $p_2$ and larger spectral radius $\rho(C^\star)$ negatively impact the empirical false rejection proportion, which deviates from the desired control level of the type I error. In the case where $\rank(B^\star)=r>0$, the testing procedure provides satisfactory type I error control for
larger sample sizes and excellent power.  

\paragraph{Sparse testing.} Since the rejection rule of the HC-statistic is based on empirical process theory \citep{shorack2009empirical} and its dependence on $\alpha$ is not explicit, we focus on illustrating how the empirical proportion of false rejections (type I error) varies with the sample size $T$, the model dimensions $(p_1,p_2)$ and the spectral radius of $C^\star$. To show the power of the test, each entry in $B^\star$ is nonzero with probability $q\in(0,1)$ such that $q(p_1p_2) = (p_1p_2)^\theta$ with $\theta=0.6$, to ensure the overall sparsity of $B^\star$ satisfies the sparsity requirement posited in Proposition~\ref{prop:asymptoticallypowerful}. The magnitude is set such that the signal-to-noise ratio is 1.2. Note that the actual number of parameters is $p_1p_2$, while the total number of subsamples used is 3000 with the length of subsamples varying according to different simulation settings to demonstrate the dependence of type I error and power on sample sizes. 
\begin{table}[!h]
	\centering
	\caption{Empirical type I error and power for sparse testing}\label{table:sparse}
		\begin{tabular}{c|c|p{1cm}p{1cm}p{1cm}p{1cm}|p{1cm}p{1cm}p{1cm}p{1cm}}
			\specialrule{.1em}{.05em}{.05em} 
			&		& \multicolumn{4}{c|}{\color{blue} type I error ($B^\star=0$)}  &  \multicolumn{4}{c}{\color{blue} power ($\text{SNR}(B^\star)=0.8$)} \\ 
			&	 $(p_1,p_2)$ & 200 & 500 & 1000 & 2000 & 200 & 500 & 1000 & 2000  \\  
			\specialrule{.05em}{.05em}{.05em}
			\multirow{4}{*}{\parbox{2cm}{$\rho(C^\star)=0.5$}} & $(20,20)$ & 0.244 & 0.097  & 0.074  & 0.055  &  1 & 1  &   1 &  1 \\ 
			&	$(50,20)$ & 0.393 & 0.131 & 0.108 & 0.074 & 1 & 1 &1 & 1\\
			&	$(20,50)$ & \underline{0.996} & 0.351 & 0.153 & 0.093 & 1 & 1 & 1 & 1 \\
			&	$(100,50)$ &\underline{1.000} & \underline{0.963} & 0.270 & 0.115 & 1 & 1 &1  &1  \\
			\hline
			\multirow{2}{*}{\parbox{2cm}{$\rho(C^\star)=0.8$}} & $(50,20)$ & 0.402 & 0.158  & 0.112  & 0.075  &  0.829 & 0.996  &   1 &  1 \\ 		
			& $(20,50)$ & \underline{0.999} & 0.430  & 0.166  & 0.111  &  1 & 1  &   1 &  1 \\ 		
			\specialrule{.1em}{.05em}{.05em} 
		\end{tabular}
\end{table}

Based on the results shown in Table~\ref{table:sparse}, when $B^\star=0$, the proposed testing procedure can effectively detect the absence of block ``Granger causality", provided that the sample size is moderately large compared to the total number of parameters being tested. However, if the model dimension is large, whereas the sample size is small, the test procedure becomes problematic and fails to provide legitimate type I error control, as desired. When $B^\star$ is nonzero, empirically the test is always able detect its presence, as long as the effective signal-to-noise ratio is beyond the detection threshold.

\section{Real Data Analysis Illustration.}\label{sec:realdata}

We employ the developed framework and associated testing procedures to address one of the motivating applications. Specifically, we analyze the temporal dynamics of the log-returns of stocks with large market capitalization and key macroeconomic variables, as well as their cross-dependence. Specifically, using the notation in~\eqref{eqn:model0}, we assume that the $X_t$ block consists of the stock log-returns, while the macroeconomic variables form the $Z_t$ block. With this specification, we assume that the macroeconomic variables are ``Granger-caused" by the stock market, but not vice versa. Note that our framework allows us to pose and test a more general question than previous work in the economics literature considered.
For example, \citet{farmer2015stock} building on previous work by \cite{fitoussi2000roots, phelps1999behind} tests only the relationship between the employment index and the composite stock index, using a 
bivariate VAR model. On the other hand, our framework enables us to consider the components of the S\&P 100 index and the ``medium" list of macroeconomic variables considered in the work of \citet{stock2005empirical}.

Next, we provide a brief description of the data used. The stock data consist of monthly observations of 71 stocks corresponding to a stable set of historical components comprising the  S\&P 100 index for the 2001-2016 period. The macroeconomic variables are chosen from the ``medium" list in  \citet{stock2005empirical}; that is, the 3 core economic indicators (Fed Funds Rate, Consumer Price Index and Gross Domestic Product Growth Rate), plus 17 additional variables with aggregate information (e.g., exchange rate, employment, housing, etc.). However, in our study, we exclude variables that exhibit a significant change after the financial crisis of 2008 (e.g. total reserves/reserves of depository institutions). We process the macroeconomic variables to ensure stationarity following the recommendations in \citet{stock2005empirical}. As a general guideline, for real quantitative variables (e.g., GDP, money supply M2), we use the relative first difference, which corresponds to their growth rate, while for rate-related variables (e.g., Federal Funds Rate, unemployment rate), we use their first differences directly. The complete lists of stocks and macroeconomic variables used in this study are given in Appendix~\ref{appendix:Dictionary}.

\medbreak
We start the analysis by using the VAR model for the stock log-returns to study their evolution over the 2001-2016 period. Analogously to the strategy employed by \citet{billio2012econometric}, we consider 36-month long rolling-windows for fitting the model $X_t = AX_{t-1} + U_t$, for a total of 143 estimates of the transition matrix $A$. VAR models involving more than 1 lag were also fitted to the data, but did not indicate temporal dependence beyond lag 1. 

To obtain the final estimates across all 143 subsamples, we employ {\em stability selection} \citep{meinshausen2010stability}, with the threshold set at 0.6 for including an edge in $A$.\footnote{The threshold is set at a relatively low level to compensate for the relative small rolling window size.} Figure~\ref{fig:transA} depicts the global clustering coefficient \citep{luce1949method} of the skeleton of the estimated $A$ over all 143 rolling windows, with the time stamps on the horizontal axis specifying the starting time of the corresponding window. 
\begin{figure}[!hbtp]
	\centering
	\caption{Global clustering coefficient of estimated $A$ over different periods}\label{fig:transA}
	\includegraphics[scale=0.4]{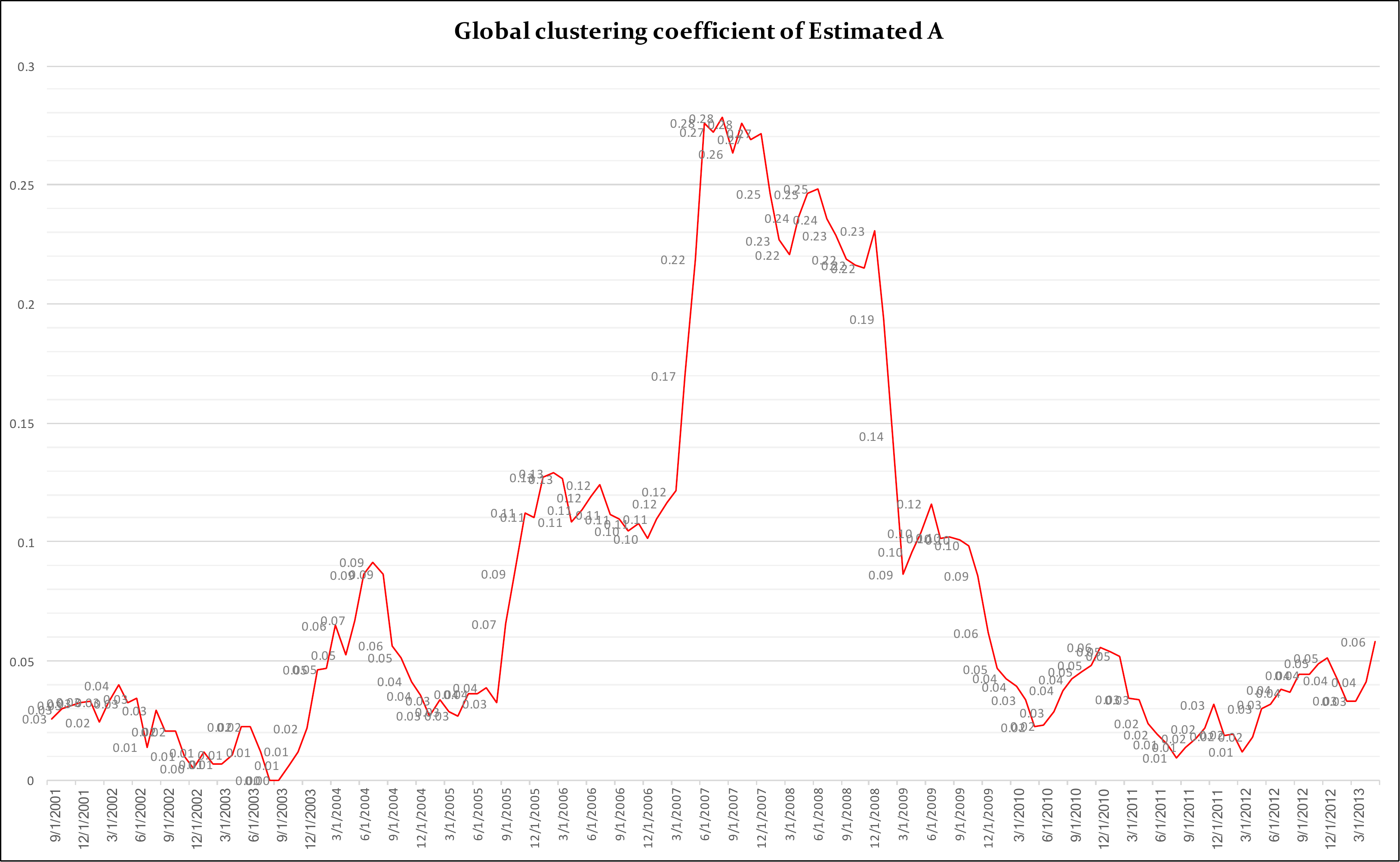}
\end{figure}

The results clearly indicate strong connectivity in lead-lag stock relationships during the financial crisis period March 2007-June 2009.
It is of interest that the data exhibit such sharp changes at time points that correspond to well documented events in the literature;
namely, March 2007 when several subprime lenders declared bankruptcy, put themselves up for sale or reported significant losses and
June 2009 that the National Bureau of Economic Research declared as the end point of the crisis. Similar patterns were broadly observed in \citet{billio2012econometric,brunetti2015interconnectedness}, albeit for a different set of stocks and using a different statistical methodology. Specifically, \citet{billio2012econometric} considered financial sector stocks (banks, brokerages, insurance companies), while 
\citet{brunetti2015interconnectedness} considered European banking stocks, and both studies used bivariate VAR models to obtain the results, thus ignoring the influence of all other components on the pairwise Granger causal relationship estimated and hence producing potentially biased estimates of connectivity. 

Next, we present the analysis based on the VAR-X component of our model, given by $Z_t = BX_{t-1} + CZ_{t-1} + V_t$ with the stock log-returns corresponding to the $X_t$ block and the (stationary) macroeconomic variables to the $Z_t$ block. As before, we fit the data within each rolling window, with the tuning parameters based on a search over a $10\times 10$ lattice (with $(\lambda_B,\lambda_C)\in [0.5,4]\times[0.2,2]$, equal-spaced) using the BIC. It should be noted that for the majority of the rolling windows, the rank of $B$ is 1 (data not shown). The sparsity level of the estimated $C$ over the 143 rolling windows is depicted in Figure~\ref{fig:sparsityC}.
\begin{figure}[!hbtp]
	\centering
	\caption{Sparsity of estimated $C$ over different periods}\label{fig:sparsityC}
	\includegraphics[scale=0.4]{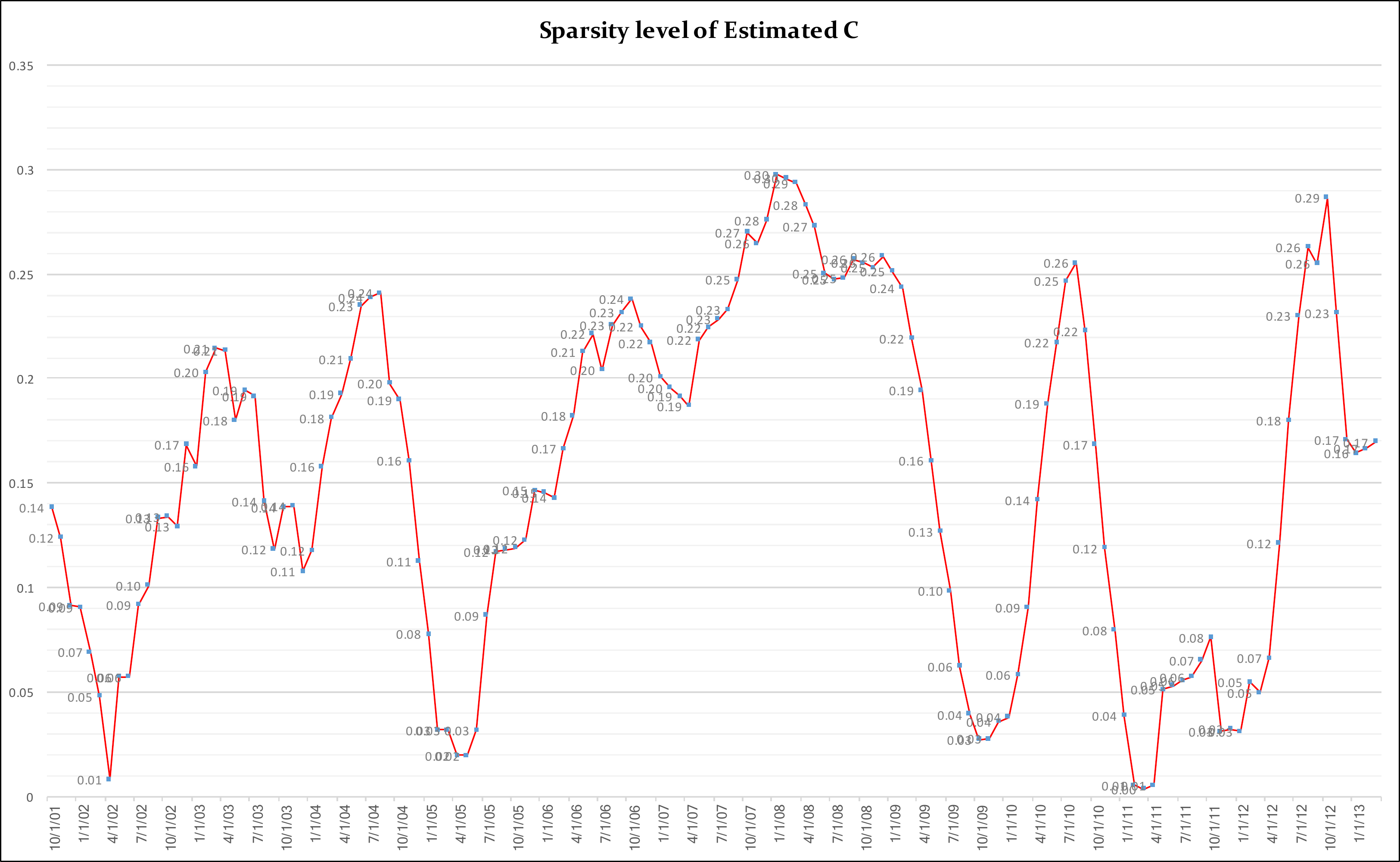}
\end{figure}
The connectivity patterns in $C$ show more complex and nuanced patterns than for stocks. Several local peaks depicted correspond to the following events: (i) March-April 2003, when the Federal Reserve cut the Effective Federal Funds Rate aggressively driving it down to 1\%, the lowest level in 45 years up to that point, (ii) January-March 2008, a significant decline in the major US stock indexes, coupled with
the freezing up of the market for auctioning rate securities with investors declining to bid, (iii) January-April 2009, characterizes the unfolding of the European debt crisis with a rescue package put together for Greece and significant downgrades of Spanish and Portuguese sovereign debt by rating agencies and (iv) July 2010, that correspond to the enactment of the Dodd-Frank Wall Street Reform and Consumer Protection Act and the acceleration of quantitative easing by the Federal Reserve Board.

%
%
\medbreak
Based on the previous findings, we partition the time frame spanning 2001-2016 into the following periods: pre- (2001/07--2007/03), during- 
(2007/01--2009/12) and post-crisis (2010/01-2016/06) one. We estimate the model parameters using the data within the entire sub-period(s). 

The estimation procedure of the transition matrix $A$ for different periods is identical to that described above using subsamples over rolling-windows. For the pre- and post- crisis periods, since we have 76 and 77 samples respectively, the stability selection threshold is set at 
0.75, whereas for the during-crisis period, at 0.6 to compensate for the small sample size (36). Table~\ref{table:stock} shows the average R-square for all 71 stocks, as well as its standard deviation, which is calculated based on in-sample fit; i.e.,the proportion of variation explained by using the $\VAR(1)$ model to fit the data.  The overall sparsity level and the spectral radius of the estimated transition matrices $A$ are also presented. The results are consistent with the previous finding of increased connectivity during the crisis. Further, for all periods the estimate of the spectral radius is fairly large, indicating strong temporal dependence of the log-returns.
\begin{table}[H]
	\centering
	\caption{Summary for estimated $A$ within different periods.}\label{table:stock}
	\begin{tabular}{l|ccc}
		\specialrule{.1em}{.05em}{.05em} 
		&  2001/07--2007/03 & 2007/01--2009/12 & 2010/01--2016/06 \\ \hline
		Averaged R sq & 0.31 & \underline{0.72} & 0.28  \\  		
		Sd of R sq &  0.103 & 0.105 & 0.094 \\
		Sparsity level of $\widehat{A}$ & 	0.17	& 0.23 & 0.19 \\ 
		Spectral radius of $\widehat{A}$& 0.67 & 0.90 & 0.75  \\
		\specialrule{.1em}{.05em}{.05em} 
	\end{tabular}
\end{table}		
Figures~\ref{fig:EstimatedA_01to07} to \ref{fig:EstimatedA_10to16} depict the estimated transition matrices $A$ for different periods, as a network, with edges thresholded based on their magnitude for ease of presentation. The node or edge coloring red/blue indicates the sign positive/negative of the corresponding entry in the transition matrix. Further, node size is proportional to the out-degree, thus indicating
which stocks influence other stocks in the next time period. The most striking feature is the outsize influence exerted by the insurance company AIG and the investment bank Goldman Sachs, whose role during the financial crisis has been well documented \citep{financial2011financial}. On the other hand, the pre- and post-crisis periods are characterized by more sparse and balanced networks, in terms of in- and out-degree magnitude. 

Next, we focus on the key motivation for developing the proposed modeling framework, namely the
inter-dependence of stocks and macroeconomic variables over the specified three sub-periods. The $p$-value for testing the hypothesis of
lack of block ``Granger causality" $H_0:B=0$, together with the spectral radius and the sparsity level for the estimated $C$ transition matrices are listed in Table~\ref{table:summaryBC}. Specifically, for all three periods, the rank of estimated $B$ is 1, indicating that the stock market as captured by its leading stocks, ``Granger-causes" the temporal evolution of the macroeconomic variables. The fact that
the rank of $B$ is 1, indicates that the inter-block influence can be captured as a single portfolio acting in unison. To investigate the relative importance of each sector in the portfolio, we group the stocks by sectors. The proportion of each sector (up to normalization) is obtained by summing up the loadings (first right singular vector of the estimated $B$) of the stocks within this sector, weighted by their market capitalization. Further, the estimated transition matrices $C$'s are depicted in network form, in Figures~\ref{fig:precrisis} to~\ref{fig:postcrisis}. It is worth noting that the temporal correlation of the macroeconomic variables significantly increased during the crisis.

Note that the proportion of various sectors in the portfolio is highly consistent with their role in stock market. For example, before crisis the financial sector had a large market capitalization (roughly 20\%), while it shrunk (to roughly 12\%) after the crisis. Also, the Information Technology (IT) and Financial (FIN) sectors are the ones exhibiting highest volatility (high beta) relative to the market, while the Utilities is the one with low volatility (low beta) , a well established stylized fact in the literature for the time period under consideration.

Next, we discuss some key relationships emerging from the model. We start with total employment (ETTL), whose dynamics are only influenced by its own past values as seen by the lack of an incoming arrow in Figure~\ref{fig:duringcrisis}. Further, an examination of the left singular vector (see Table~\ref{table:leftSV}) of $B$ strongly indicates the impact of the stock market on total employment. This finding is consistent with the analysis in \citet{farmer2015stock}, which argues that the crash of the stock market provides a plausible explanation for the great recession. However, the analysis in \citet{farmer2015stock} is based on bivariate VAR models involving only employment and the stock index. Therefore, there is a possibility that the stock market is reacting to some other information captured by other macroeconomic variables, such as GDP, capital spending, inflation, interest rates, etc. However, our high-dimensional VAR model simultaneously analyzes a key set of macroeconomic variables and also accounts for the influence of the largest stocks in the market. Hence, it automatically overcomes the criticism leveraged by \cite{sims1992interpreting} about misinterpretations of findings from small scale VAR models due to the omission of important variables, and further echoed in the discussion in~\cite{bernanke2005measuring}.

Another interesting finding is the strong influence of the stock market on GDP in the pre- and post-crisis period, consistent with the popular view of being a leading indicator for GDP growth. Further, capital utilization is positively impacted during the crisis period by GDP growth and total employment---which are both falling and hence reducing capital utilization---and further accentuated by the impact of the stock market---also falling---thus reinforcing the lack of available capital goods and resources.

In summary, the brief analysis presented above provides interesting insights into the interaction of the stock market with the real economy, identifies a number of interesting developments during the crisis period and reaffirms a number of findings studied in the literature, while ensuring that a much larger information set is utilized (a larger number of variables included) than in previous analysis.
Therefore, high-dimensional multi-block VAR models are useful for analyzing complex temporal relationships and provide insights into their dynamics.
\begin{table}[!htbp]
	\centering
	\caption{Summary for estimated $B$ and $C$ within different periods.}\label{table:summaryBC}
	\begin{tabular}{l|ccc}
		\specialrule{.1em}{.05em}{.05em} 
		& 2001/07--2007/03 & 2007/01--2009/12 & 2010/01--2016/16 \\ \hline
		$p$-value for testing $H_0:B=0$ & 0.075 & 0.009 & 0.044 \\
		\specialrule{.05em}{.05em}{.05em}
		Sparsity level of $\widehat{C}$ & 	0.06	& 0.25 & 0.06 \\ 
		Spectral radius of $\widehat{C}$& 0.35 & 0.76 & 0.40  \\ 
		\specialrule{.1em}{.05em}{.05em} 
	\end{tabular}
\end{table}

\begin{table}[!htbp]
	\centering
	\caption{Left Singular Vectors of Estimated $B$ for different periods}\label{table:leftSV}
	\begin{tabular}{l|ccc}
		\specialrule{.1em}{.05em}{.05em} 
		&	Pre-Crisis	&	During-Crisis	&	Post-Crisis	\\ \hline
		FFR	&	-0.24	&	-0.26	&	-0.23	\\
		T10yr	&	-0.09	&	0.14	&	0.16	\\
		UNEMPL	&	-0.07	&	0.01	&	-0.07	\\
		IPI	&	-0.43	&	0.34	&	0.26	\\
		ETTL	&	0.33	&	0.24	&	0.13	\\
		M1	&	0.23	&	-0.12	&	-0.47	\\
		AHES	&	-0.01	&	0.30	&	0.17	\\
		CU	&	-0.49	&	0.32	&	0.27	\\
		M2	&	0.10	&	-0.04	&	-0.32	\\
		HS	&	0.51	&	-0.02	&	-0.02	\\
		EX	&	-0.18	&	0.41	&	0.06	\\
		PCEQI	&	-0.07	&	-0.18	&	0.41	\\
		GDP	&	0.10	&	-0.02	&	0.05	\\
		PCEPI	&	0.00	&	0.14	&	-0.01	\\
		PPI	&	-0.15	&	0.00	&	0.06	\\
		CPI	&	0.01	&	0.15	&	-0.31	\\
		SP.IND	&	-0.06	&	-0.53	&	0.38	\\
		\specialrule{.1em}{.05em}{.05em} 
	\end{tabular}
\end{table}

\begin{remark}
	We also applied our multi-block model with the first block $X_t$ corresponding to the macro-economic variables and the second block $Z_t$ the stocks variables (results not shown). The key question is whether there is also ``Granger causality" from the broader economy to the stock market. The results are inconclusive due to sample size issues that do not allow us to properly test for the key hypothesis whether $B=0$ or not. Specifically, the length of the sub-periods is short compared to the dimensionality required for the test procedure. A similar issue arises, which is related to the detection boundary for the sparse testing procedure during the crisis period. Further, for a sparse $B$, an examination of its entries shows that Employment Total did not impact the stock market, which is in line with the conclusion reached at the aggregate level by \citet{farmer2015stock}. On the other hand, GDP negatively impacts stock log-returns, which may act as a leading indicator for suppressed investment and business growth and hence future stock returns. 
\end{remark}

\section{Discussion.}\label{sec:Discussion}

We briefly discuss generalizations of the model to the case of more than two blocks, as mentioned in the introductory section.
For the sake of concreteness, consider a triangular recursive linear dynamical system given by:
\begin{equation}\label{model}
\begin{split}
X^{(1)}_t & = A_{11} X^{(1)}_{t-1} + \epsilon^{(1)}_t, \\
X^{(2)}_t & = A_{12} X^{(1)}_{t-1} + A_{22} X^{(2)}_{t-1} + \epsilon^{(2)}_t, \\
X^{(3)}_t & = A_{13}X^{(1)}_{t-1} + A_{23}X^{(2)}_{t-1} + A_{33}X^{(3)}_{t-1} + \epsilon^{(3)}_t, \\
& \vdots 
\end{split}
\end{equation}
where $X^{(j)}\in\mathbb{R}^{p_j}$ denotes the variables in group $j$, $A_{ij}~(i<j)$ encodes the dependency of $X^{(j)}$ on the past values of variables in group $i$, and $A_{jj}$ encodes the dependency on its own past values. Further, $\{\epsilon^{(j)}_t\}$ is the innovation process that is neither temporally, nor cross-sectionally correlated, i.e.,
\begin{equation*}
\Cov(\epsilon_t^{(j)} ,\epsilon_s^{(j)}) = 0~(s\neq t), \quad  \Cov(\epsilon_t^{(i)} ,\epsilon_s^{(j)}) = 0~(i\neq j,~\forall~(s,t)),\quad \Cov(\epsilon_t^{(j)}, \epsilon_t^{(j)}) = \big(\Omega^{(j)}\big)^{-1},
\end{equation*}
with $\Omega^{(j)}$ capturing the conditional contemporaneous dependency of variables within group $j$. The model in~\eqref{model} can also be viewed from a multi-layered time-varying network perspective: nodes in each layer are ``Granger-caused" by nodes from its previous layers, and are also dependent on its own past values. As previously mentioned, in various real applications, it is of interest to obtain estimates of the transition matrices, and/or test if ``Granger-causality" is present between interacting blocks; i.e., to test $A_{ij}=0$ for some $i\neq j$.

The triangular structure of the system decouples the estimation of the transition matrices from each equation, and hence a
straightforward extension of the estimation procedure presented in Section~\ref{sec:estimation} becomes applicable. Specifically, to obtain estimates of the transition matrices $A_{ij}$'s for fixed $j$ and $1\leq i\leq j$, and the inverse covariance $\Omega^{(j)}$, the optimization problem is formulated as follows:
\small
\begin{align}
(\{\widehat{A}_{ij}\}_{i\leq j},\widehat{\Omega}^{(j)}) = \argmin\limits\limits_{A_{ij},\Omega^{(j)}}\Big\{  -\log\det\Omega^{(j)} + &\frac{1}{T}\sum_{t=1}^T\big( x_t^{(j)} - \sum_{i=1}^j A_{ij}x_{t-1}^{(i)}\big)'\Omega^{(j)}\big(x_t^{(j)} - \sum_{1\leq i\leq j}A_{ij}x_{t-1}^{(i)}\big)\nonumber\\
& + \sum_{i=1}^j \mathcal{R}(A_{ij}) + \rho^{(j)} \|\Omega^{(j)}\|_{1,\text{off}}  \Big\}, \label{eqn:opt-generalized}
\end{align}
\normalsize
where the exact expression for the $\mathcal{R}(A_{ij})$ adapts to the structural assumption imposed on the corresponding transition matrix (sparse/low-rank). Solving~\eqref{eqn:opt-generalized} again requires an iterative algorithm involving the alternate update between transition matrices and the inverse covariance matrices. Further, for updating the values of the transition matrices, a
cyclic block-coordinate updating procedure is used. 

Consistency results can be established analogously to those provided in Section~\ref{sec:theory}, under the posited conditions of restricted strong convexity (RSC) and a deviation bound. With a larger number of interacting blocks of variables, lower bounds for the lower extremes of the spectra involve all corresponding transition matrices. The error rates that can be obtained are as follows: (i) if equation $k$ only involves sparse transition matrices, then the finite-sample bounds of the transition matrices in this layer in Frobenius norm are of the order $O\big(\sqrt{\tfrac{\log p_k+\log \sum\nolimits_{i\leq k}p_k}{T}}\big)$, while (ii) if some of the transition matrices are assumed low rank, then the corresponding finite sample bounds are of the order $O\big(\sqrt{\tfrac{p_k+\sum\nolimits_{i\leq k}p_k}{T}}\big)$.

Another generalization that can be handled algorithmically with the same estimation procedure discussed above is the presence of
$d$-lags in the specification of the linear dynamical system. Based on the consistency results developed in this work, together with the theoretical findings for VAR($d$) models presented in \citet{basu2015estimation}, we expect all the established theoretical properties of the transition matrices estimates to go through under appropriate RSC and deviation bound conditions.

\clearpage
\begin{figure}[H]
	\centering
	\caption{Sector proportion and Estimated C for pre-crisis period}\label{fig:precrisis}\vspace*{-2mm}
	\begin{boxedminipage}{7cm}
		\includegraphics[scale=0.38]{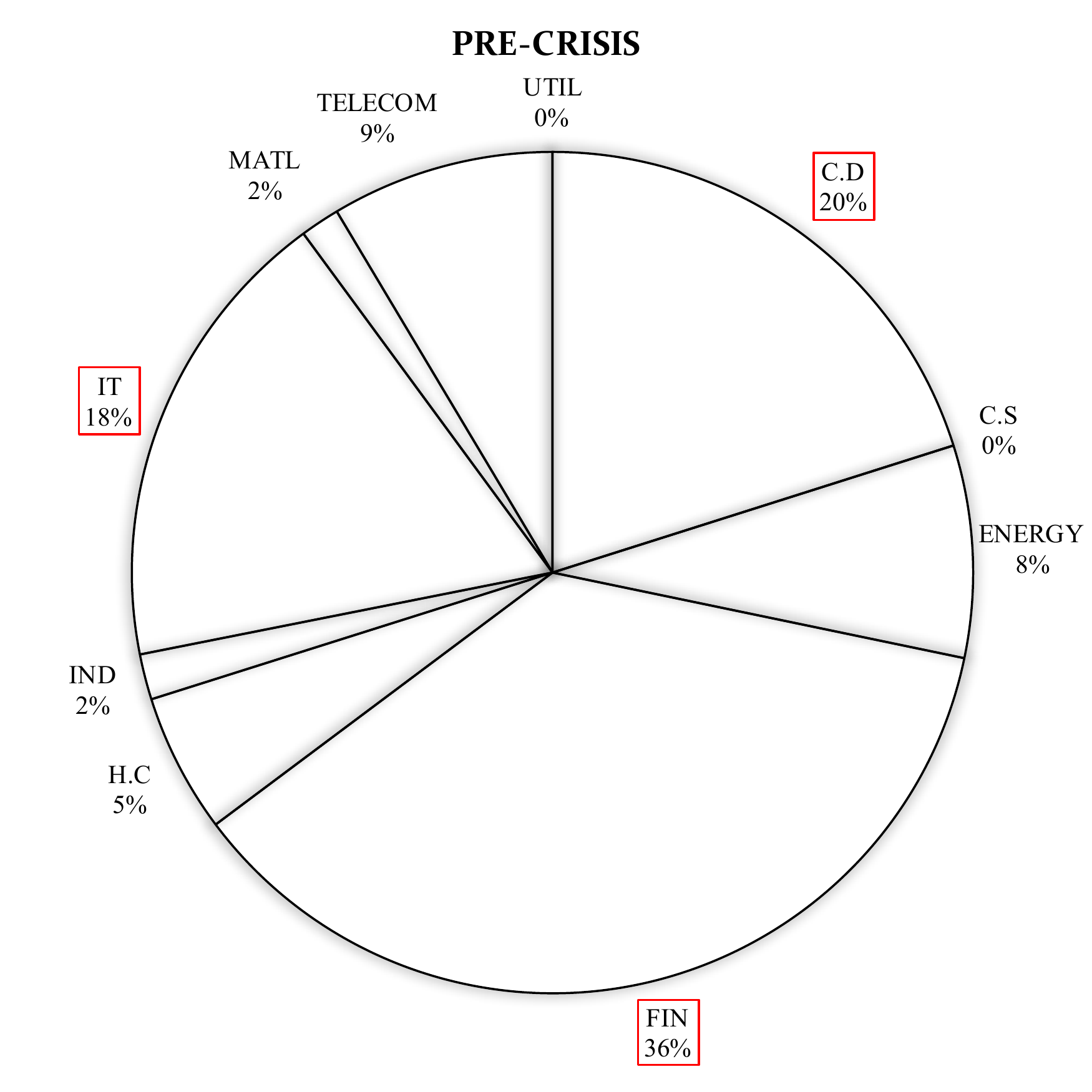}
	\end{boxedminipage}
	\begin{boxedminipage}{7cm}	
		\includegraphics[scale=0.38]{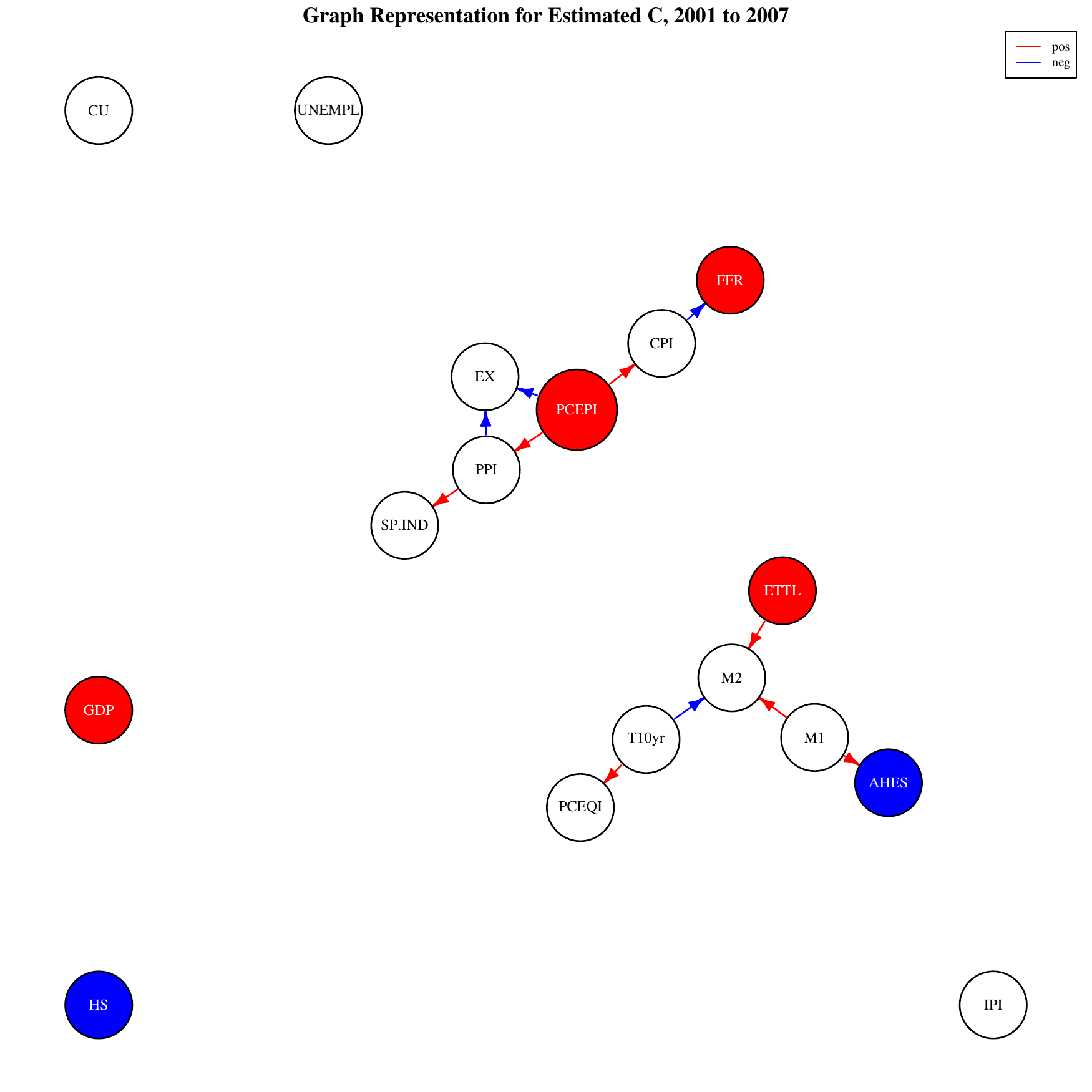}
	\end{boxedminipage}
\end{figure}
\begin{figure}[H]
	\centering
	\caption{Sector proportion and Estimated C for during-crisis period}\label{fig:duringcrisis}\vspace*{-2mm}
	\begin{boxedminipage}{7cm}
		\includegraphics[scale=0.38]{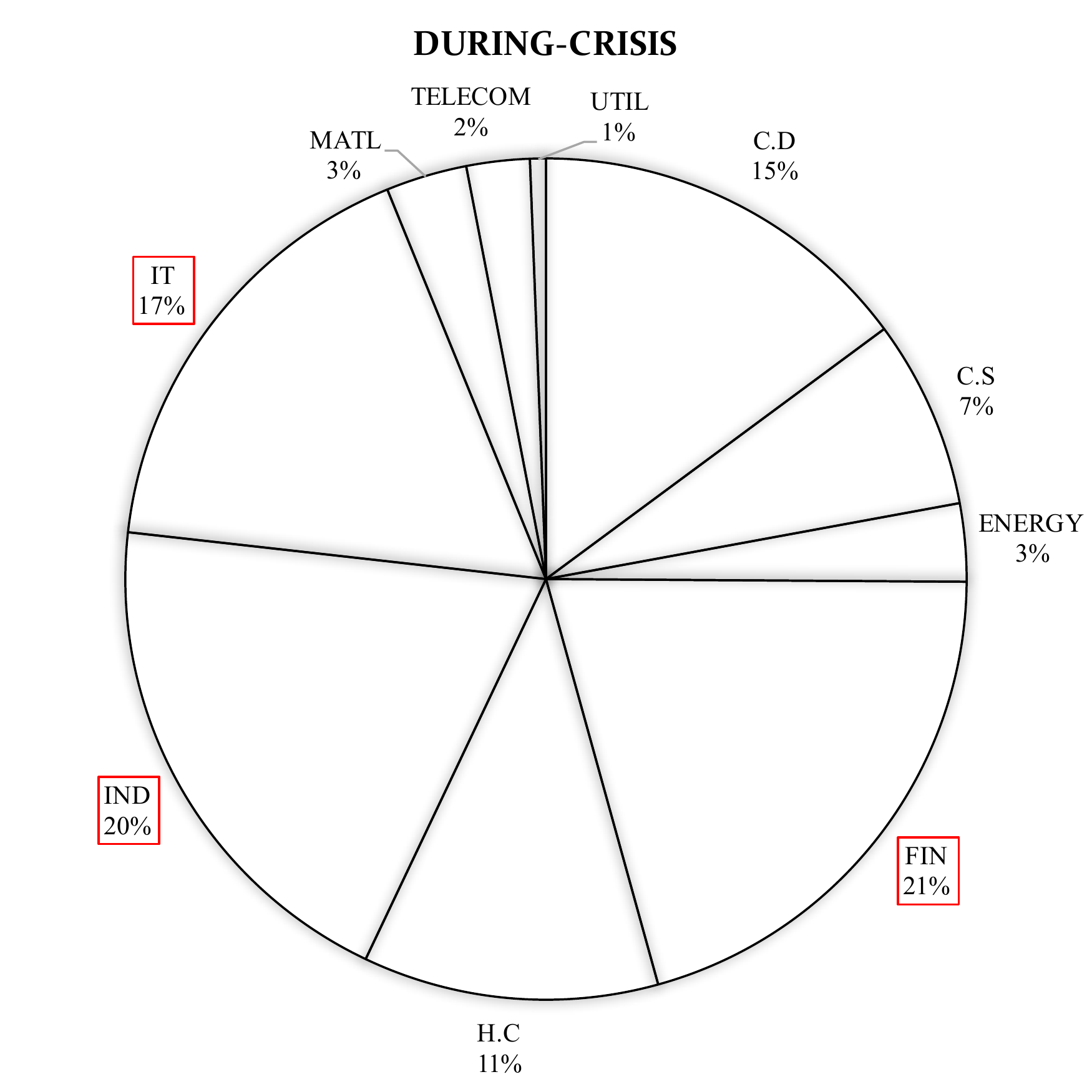}
	\end{boxedminipage}
	\begin{boxedminipage}{7cm}
		\includegraphics[scale=0.38]{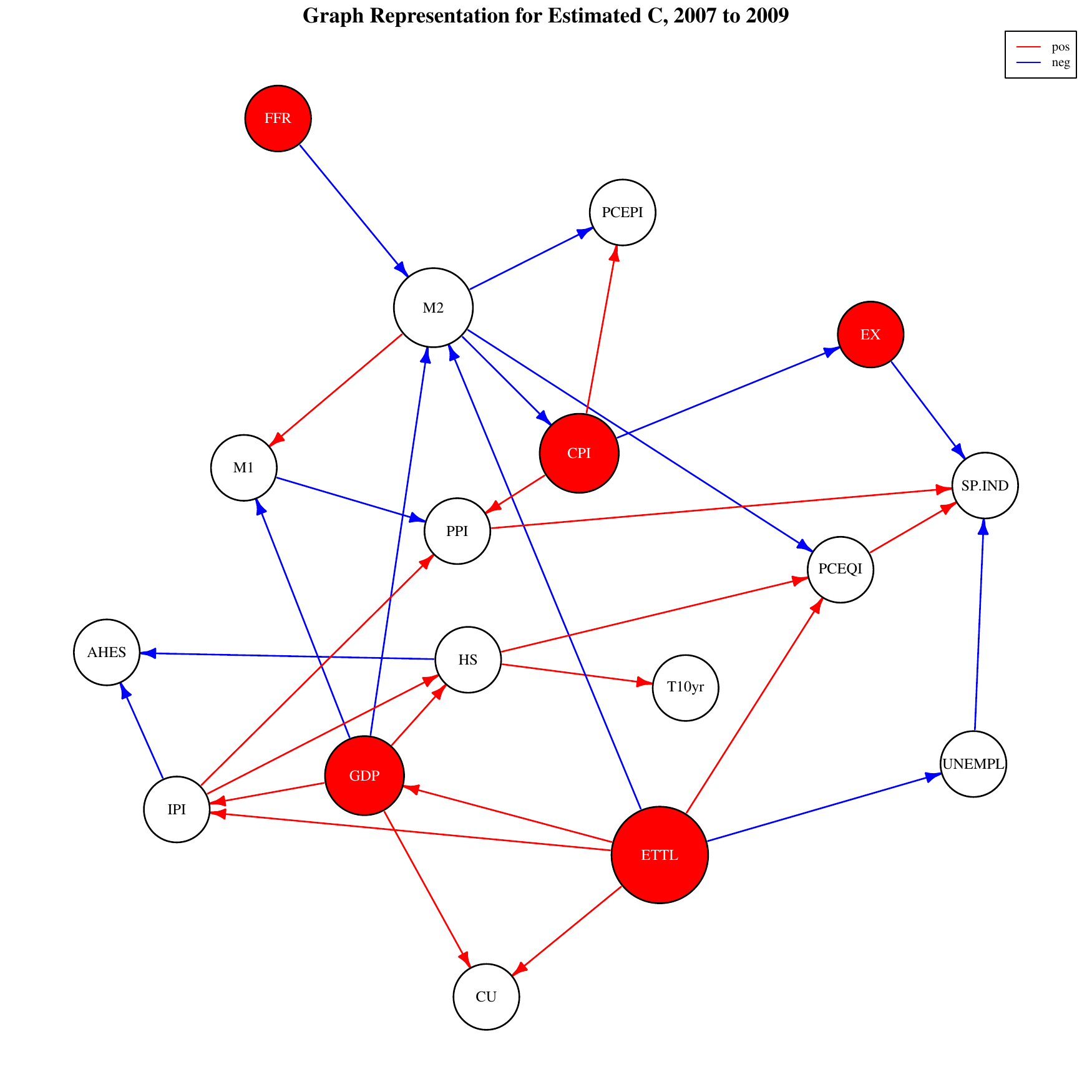}
	\end{boxedminipage}
\end{figure}
\begin{figure}[H]
	\centering
	\caption{Sector proportion and Estimated C for post-crisis period}\label{fig:postcrisis}\vspace*{-2mm}
	\begin{boxedminipage}{7cm}
		\includegraphics[scale=0.38]{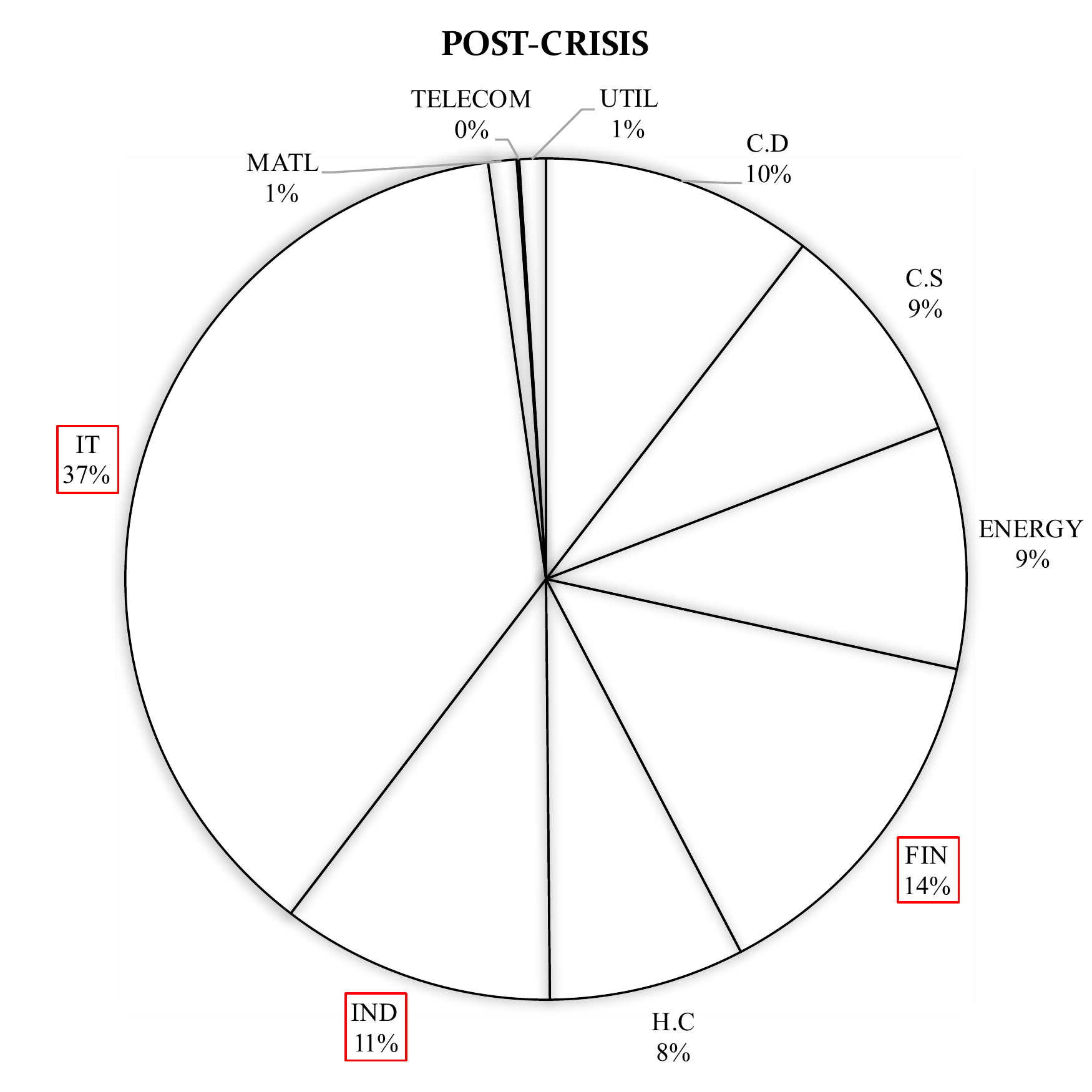}
	\end{boxedminipage}
	\begin{boxedminipage}{7cm}
		\includegraphics[scale=0.38]{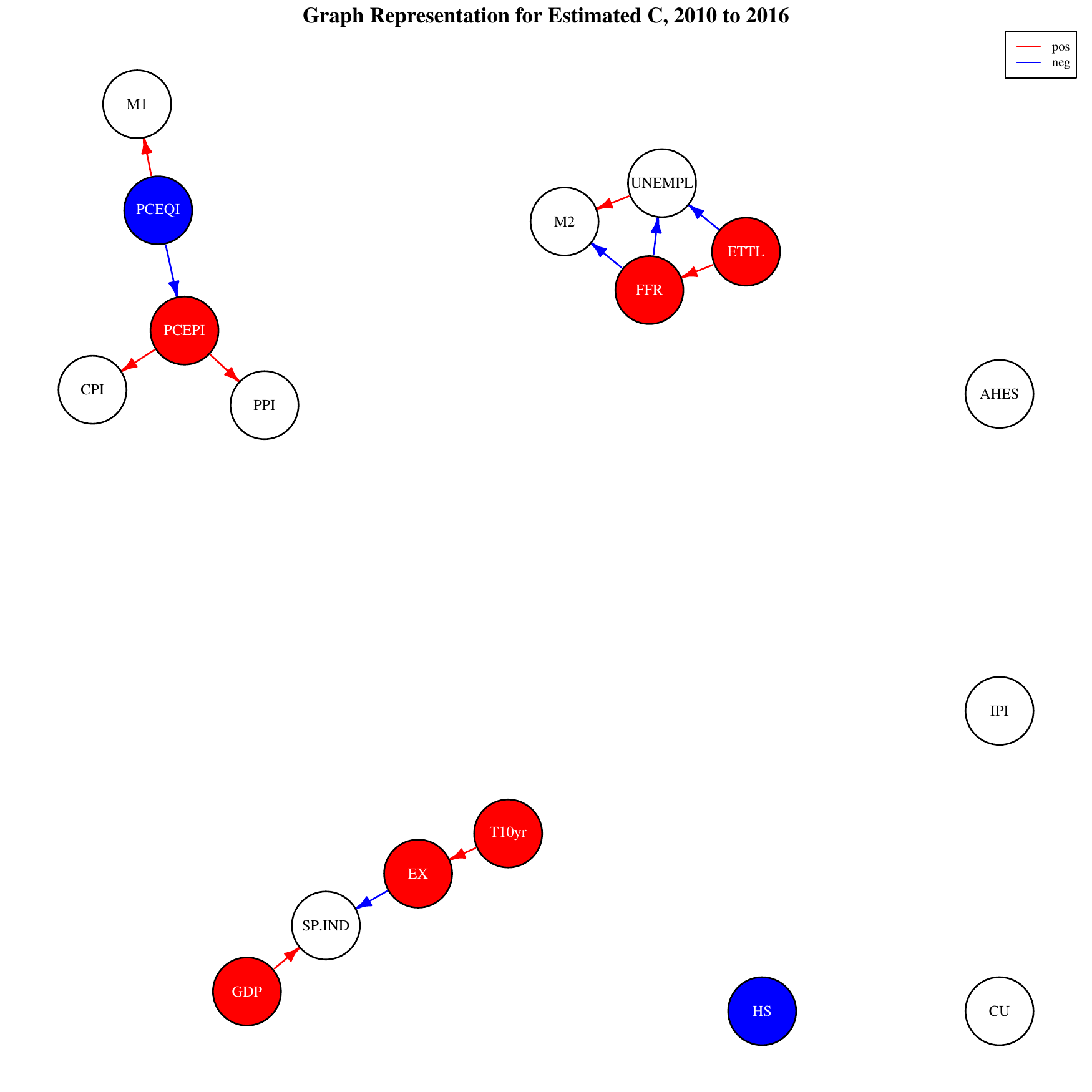}
	\end{boxedminipage}
\end{figure}
\begin{figure}[h]
	\centering	
	\caption{Estimated transition matrix for stock dynamics between 2001 to 2007}\label{fig:EstimatedA_01to07}
	\begin{boxedminipage}{15cm}
		\includegraphics[scale=.8]{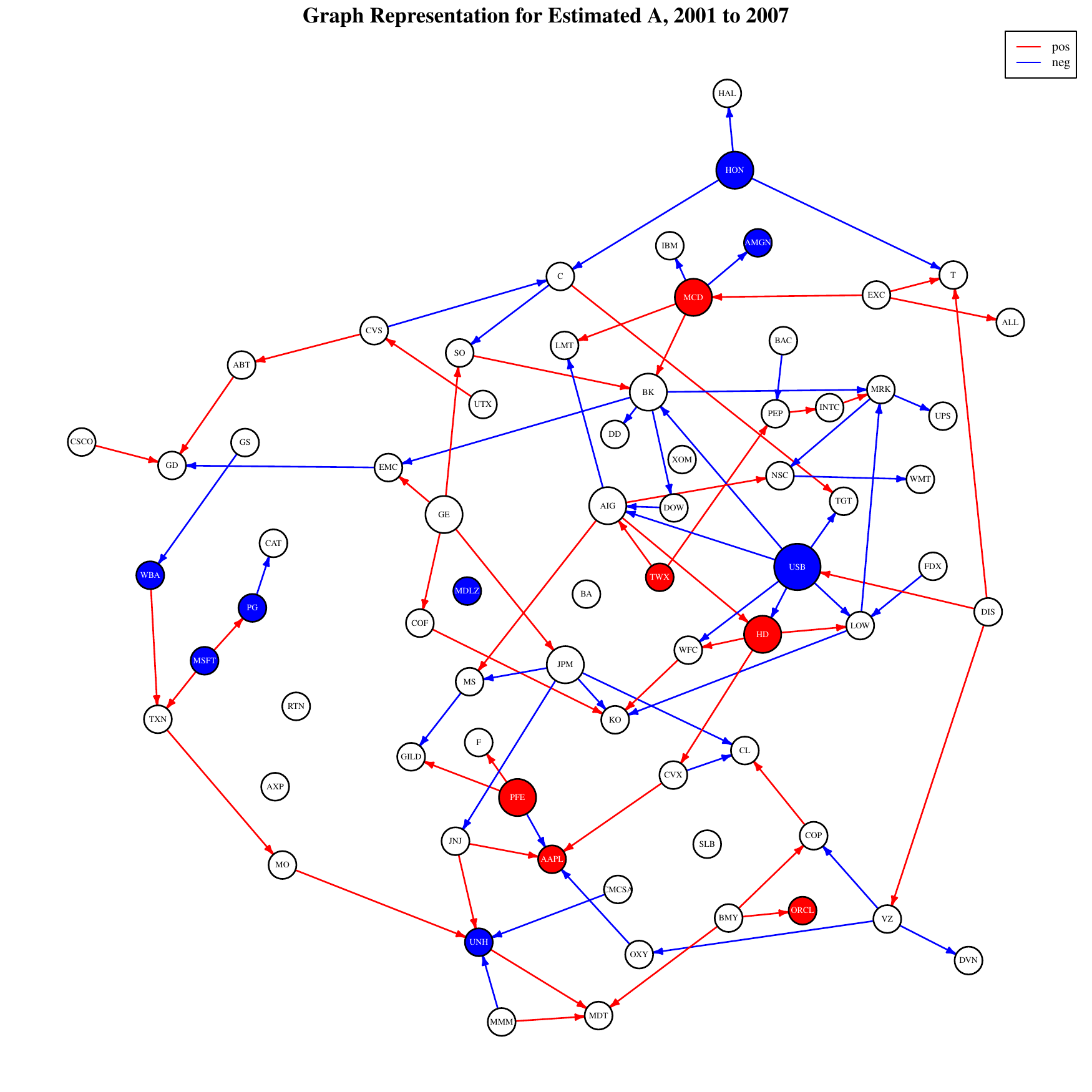}
	\end{boxedminipage}
\end{figure}
\begin{figure}
	\centering
	\caption{Estimated transition matrix for stock dynamics between 2007 to 2009}\label{fig:EstimatedA_07to09}
	\begin{boxedminipage}{15cm}
		\includegraphics[scale=.8]{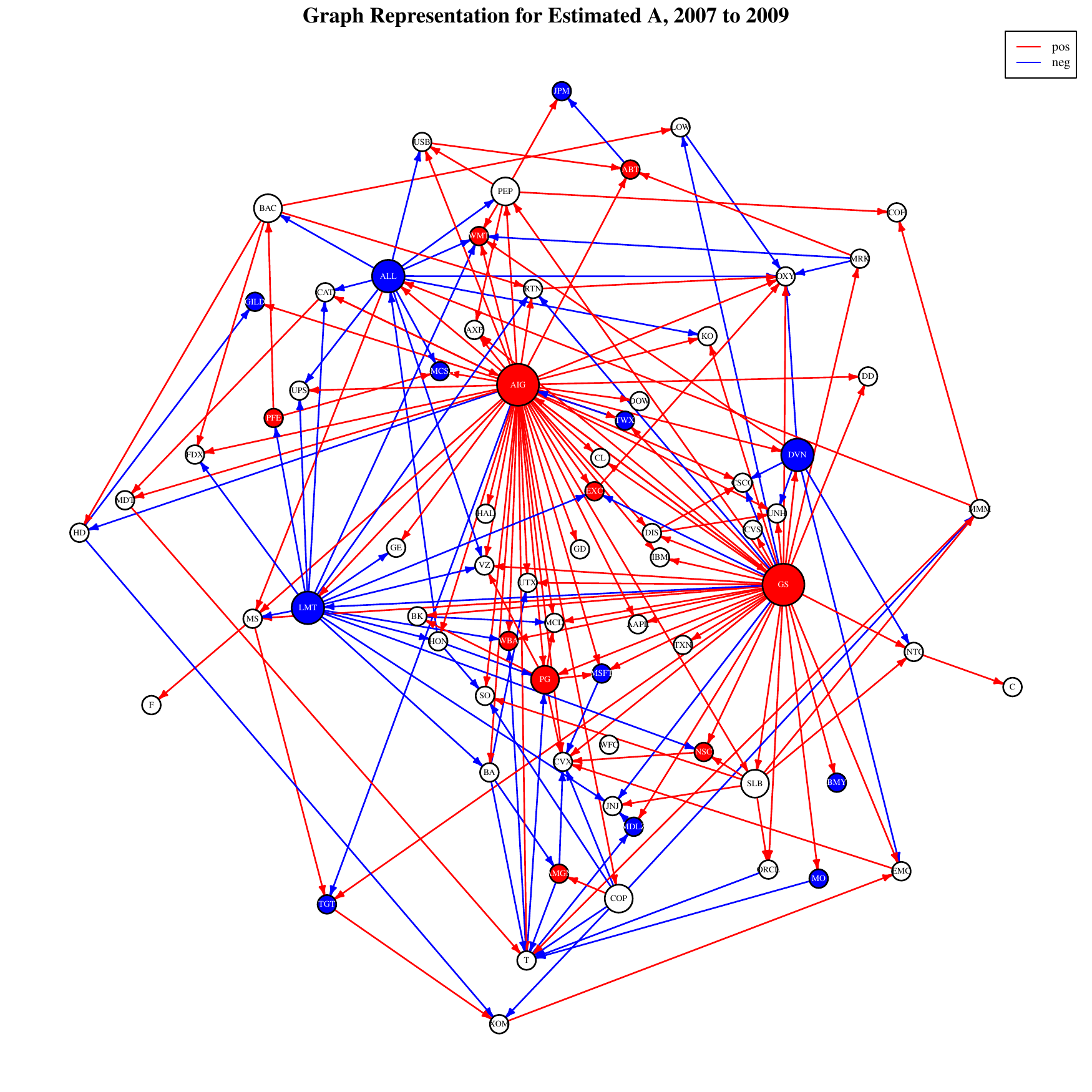}
	\end{boxedminipage}
\end{figure}
\begin{figure}
	\centering
	\caption{Estimated transition matrix for stock dynamics between 2010 to 2016}\label{fig:EstimatedA_10to16}
	\begin{boxedminipage}{15cm}
		\includegraphics[scale=.8]{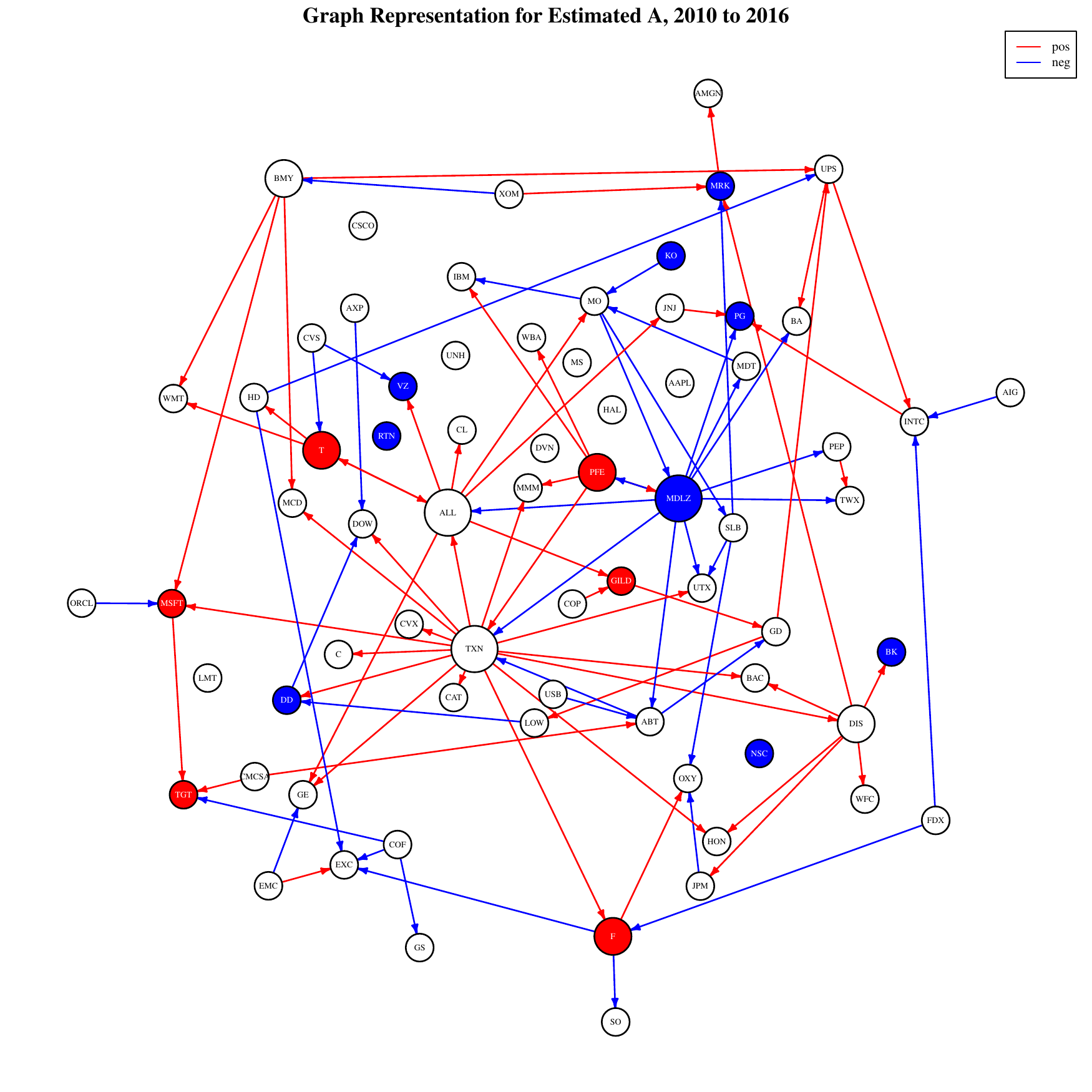}
	\end{boxedminipage}
\end{figure}

\clearpage
\appendix\label{sec:appendix}
\section{Additional Theorems and Proofs for Theorems.}\label{appendix:theorems}

In this section, we introduce two additional theorems that respectively establish the consistency properties for the initializers $\widehat{A}^{(0)}$ and $(\widehat{B}^{(0)},\widehat{C}^{(0)})$, for {\em fixed} realizations of the processes $\{X_t\}$ and $\{Z_t\}$. Specifically, $\widehat{A}^{(0)}$ and $(\widehat{B}^{(0)},\widehat{C}^{(0)})$ are solutions to the following optimization problems: 
\begin{align}
\widehat{A}^{(0)} & := \argmin\limits_{A}\big\{ \tfrac{1}{T} \vertiii{\mathcal{X}^T - \mathcal{X}A'}_F + \lambda_A \|A\|_1 \big\}, \label{eqn:optA0}\\
(\widehat{B}^{(0)},\widehat{C}^{(0)}) & := \argmin\limits_{B,C}\big\{ \tfrac{1}{T} \vertiii{\mathcal{Z}^T - \mathcal{X}B' - \mathcal{Z}C'}_F + \lambda_B \vertiii{B}_* + \lambda_C \vertii{C}_1\big\}. \label{eqn:optBC0} 
\end{align}
Note that they also correspond to estimators of the setting where there is no contemporaneous dependence among the idiosyncratic error processes. If we additionally introduce operators $\mathfrak{X}_0$ and $\mathfrak{W}_0$ defined as 
\begin{equation*}
\begin{split}
\mathfrak{X}_0:&~~~\mathfrak{X}_0(\Delta) = \mathcal{X}'\Delta,~~\text{for }\Delta\in\mathbb{R}^{p_1\times p_1}, \\
\mathfrak{W}_0:&~~~\mathfrak{W}_0(\Delta) = \mathcal{W}'\Delta,~~\text{for }\Delta\in\mathbb{R}^{p_2\times (p_1+p_2)} \quad \text{where }\mathcal{W}:=[\mathcal{X},\mathcal{Z}],
\end{split}
\end{equation*}
then~\eqref{eqn:optA0} and~\eqref{eqn:optBC0} can be equivalently written as
\begin{align*}
\widehat{A}^{(0)} & := \argmin\limits_{A}\big\{ \tfrac{1}{T} \vertiii{\mathcal{X}^T - \mathfrak{X}_0(A)}_F + \lambda_A \|A\|_1 \big\}, \\
(\widehat{B}^{(0)},\widehat{C}^{(0)}) & := \argmin\limits_{B,C}\big\{ \tfrac{1}{T} \vertiii{\mathcal{Z}^T - \mathfrak{W}_0(B_{\text{aug}},C_{\text{aug}})}_F + \lambda_B \vertiii{B}_* + \lambda_C \vertii{C}_1\big\},
\end{align*}
where $B_{\text{aug}}:=[B,O_{p_2\times p_2}],C_{\text{aug}}:=[O_{p_2\times p_1},C]$.

\begin{theorem}[Error bounds for $\widehat{A}^{(0)}$]\label{thm:consistencyA} Suppose the operator $\mathfrak{X}_0$ satisfies the RSC condition with norm $\Phi(\Delta)=\|\Delta\|_1$, curvature $\alpha_{\text{RSC}}>0$ and tolerance $\tau>0$, so that 
	\begin{equation*}
	s_A^\star\tau \leq \alpha_{\text{RSC}}/32.
	\end{equation*}
	Then, with regularization parameter $\lambda_A$ satisfying $\lambda_A \geq 4\| \mathcal{X}'\mathcal{U}/T\|_\infty$, the solution to~\eqref{eqn:optA0} satisfies the following bounds:
	\begin{equation*}
	\smalliii{\widehat{A}^{(0)}-A^\star}_F\leq 12\sqrt{s^\star_A}\lambda_A/\alpha_{\text{RSC}} \qquad \text{and} \qquad \|\widehat{A}-A^\star\|_1\leq 48s^\star_A\lambda_A/\alpha_{\text{RSC}}.
	\end{equation*}
\end{theorem}

\begin{theorem}[Error bound for $(\widehat{B}^{(0)},\widehat{C}^{(0)})$]\label{thm:consistencyBC} Let $\mathcal{J}_{C^\star}$ be the support set of $C^\star$ and $s^\star_C$ denote its cardinality. Let $r^\star_B$ be the rank of $B^\star$.  Assume that $\mathfrak{W}_0$ satisfies the RSC condition with norm 
	\begin{equation*}
	\Phi(\Delta) := \inf\limits_{B_{\text{aug}}+C_{\text{aug}}=\Delta} \mathcal{Q}(B,C), \qquad \text{where }~\mathcal{Q}(B,C):=\vertiii{B}_*+\tfrac{\Lambda_C}{\lambda_B}\|C\|_1,
	\end{equation*}	
	curvature $\alpha_{\text{RSC}}$ and tolerance $\tau$ such that 
	\begin{equation*}
	128\tau r^\star_B < \alpha_{\text{RSC}}/4 \qquad \text{and} \qquad 64\tau s^\star_C(\lambda_C/\lambda_B)^2 < \alpha_{\text{RSC}}/4.
	\end{equation*}
	Then, with regularization parameters $\lambda_B$ and $\lambda_C$ satisfying 
	\begin{equation*}
	\lambda_B\geq 4\vertiii{\mathcal{W}'\mathcal{V}/T}_{\text{op}} \quad \text{and} \quad \lambda_C \geq 4\vertii{\mathcal{W}'\mathcal{V}/T}_\infty, 
	\end{equation*}
	the solution to~\eqref{eqn:optBC0} satisfies the following bounds:
	\begin{equation*}
	\smalliii{\widehat{B}^{(0)}-B^\star}_F^2 + \smalliii{\widehat{C}^{(0)}-C^\star}_F^2 \leq 4 \left(2r^\star_B\lambda_B^2 + s^\star_C\lambda_C^2\right)/\alpha^2_{\text{RSC}}.
	\end{equation*}	
\end{theorem}

In the rest of this subsection, we first prove Theorem~\ref{thm:consistencyA} and~\ref{thm:consistencyBC}, then prove Theorem~\ref{thm:algo1} and~\ref{thm:algo2}, whose  statements are given in~Section~\ref{sec:ConsistencyProp}. 

%
%
\begin{proof}[\textbf{Proof of Theorem~\ref{thm:consistencyA}}]
	
	For the ease of notation, in this proof, we use $\widehat{A}$ to refer to $\widehat{A}^{(0)}$ whenever there is no ambiguity. Let $\beta^\star_A = \mathrm{vec}(A^\star)$ and denote the residual matrix and its vectorized version by $\Delta_A = \widehat{A}-A^\star$ and $\Delta_{\beta_A} = \widehat{\beta}_A - \beta_A^\star$, respectively. By the optimality of $\widehat{A}$ and the feasibility of $A^\star$, the following {\em basic inequality} holds:
	\begin{equation*}
	\tfrac{1}{T} \vertiii{\mathfrak{X}_0(\Delta_A)}_F^2 \leq \tfrac{2}{T} \llangle \Delta_A, \mathcal{X}'\mathcal{U} \rrangle + \lambda_A \left\{ \vertii{A^\star}_1- \vertii{A^\star + \Delta_A}_1\right\},
	\end{equation*}
	which is equivalent to:
	\begin{equation}\label{A:basic}
	\Delta_{\beta_A}' \widehat{\Gamma}^{(0)}_{X} \Delta_{\beta_A} \leq \tfrac{2}{T} \langle \Delta_{\beta_A}, \mathrm{vec}(\mathcal{X}'\mathcal{U}) \rangle + \lambda_A \left\{ \vertii{\beta_A^\star}_1 - \vertii{\beta_A^\star + \Delta_{\beta_A}}_1  \right\},
	\end{equation}
	where $\widehat{\Gamma}^{(0)}_X=\mathrm{I}_{p_1}\otimes \tfrac{\mathcal{X}'\mathcal{X}}{T}$. By H\"{o}lder's inequality and the triangle inequality, an upper bound for the right-hand-side of~\eqref{A:basic} is given~by
	\begin{equation}\label{A:upperbound1}
	\tfrac{2}{T} \vertii{\Delta_{\beta_A}}_1 \vertii{\mathcal{X}'\mathcal{U}}_\infty + \lambda_A \|\Delta_{\beta_A}\|_1.
	\end{equation}
	Now with the specified choice of $\lambda_A$, by Lemma~\ref{lemma:J3cone}, $\|\Delta_{\beta_A|\mathcal{J}_{A^\star}}\|_1\leq 3\|\Delta_{\beta_A|\mathcal{J}_{A^\star}^c}\|_1$ i.e., $\Delta_{\beta_A} \in \mathcal{C}(\mathcal{J}_{A^\star},3)$, hence $\|\Delta_{\beta_A}\|_1 \leq 4\| \Delta_{\beta_A|\mathcal{J}_{A^\star}}\|_1 \leq 4\sqrt{s_A^\star}\|\Delta_{\beta_A}\|$. By choosing $\lambda_A\geq 4\|\mathcal{X}'\mathcal{U}/T\|_\infty$, (\ref{A:upperbound1}) is further upper bounded~by
	\begin{equation}\label{A:upperbound2}
	\frac{3}{2} \lambda_A \|\Delta_{\beta_A}\|_1 \leq 6\sqrt{s^\star_A}\lambda_A \|\Delta_{\beta_A}\|.
	\end{equation}
	Combined with the RSC condition and the upper bound given in~\eqref{A:upperbound2}, we have 
	\begin{equation*}
	\begin{split}
	\frac{\alpha_{\text{RSC}}}{2}\|\Delta_{\beta_A}\|^2 - \frac{\tau}{2}\|\Delta_{\beta_A}\|^2_1 \leq \frac{1}{2}\Delta_{\beta_A}'\widehat{\Gamma}^{(0)}_{X}\Delta_{\beta_A} \leq 3\sqrt{s^\star_A}\lambda_A \|\Delta_{\beta_A}\|,\\
	\frac{\alpha_{\text{RSC}}}{4} \|\Delta_{\beta_A}\|^2 \leq \left( \frac{\alpha_{\text{RSC}}}{2} -  \frac{16s^\star_A\tau}{4} \right)   \|\Delta_{\beta_A}\|^2 \leq 3\sqrt{s^\star_A}\lambda_A \|\Delta_{\beta_A}\|,
	\end{split}
	\end{equation*}
	which implies 
	\begin{equation*}
	\|\Delta_{\beta_A}\| \leq 12\sqrt{s^\star_A}\lambda_A/\alpha_{\text{RSC}} \qquad \text{and} \qquad \|\Delta_{\beta_A}\|_1 \leq 48s^\star_A\lambda_A/\alpha_{\text{RSC}}.
	\end{equation*}
	It is easy to see that these bounds also hold for $\|\Delta_A\|_F$ and $\|\Delta_A\|_1$, respectively. 
\end{proof}

%
%
Next, to prove Theorem~\ref{thm:consistencyBC}, we introduce the following two sets of subspaces $\{\mathcal{S}_\Theta,\mathcal{S}_\Theta^\perp\}$ and $\{\mathcal{R}_\Theta,\mathcal{R}_\Theta^c\}$ associated with some generic matrix $\Theta\in\mathbb{R}^{m_1\times m_2}$, in which the nuclear norm and the $\ell_1$-norm are decomposable, respectively \citep[see][]{negahban2009unified}.
Specifically, let the singular value decomposition of $\Theta$ be $\Theta = U\Sigma V'$ with $U$ and $V$ being orthogonal matrices. Let $r=\text{rank}(\Theta)$, and we use $U^r$ and $V^r$ to denote the first $r$ columns of $U$ and $V$ associated with the $r$ singular values of $\Theta$. Further, define  
\begin{equation}\label{defn:basisspace}
\begin{split}
\mathcal{S}_\Theta& : = \left\{ \Delta\in\mathbb{R}^{m_1\times m_2} |  \text{row}(\Delta) \subseteq V^r ~~~\text{and}~~~ \text{col}(\Delta)\subseteq U^r \right\}, \\
\mathcal{S}^\perp_\Theta& : = \left\{ \Delta\in\mathbb{R}^{m_1\times m_2} |  \text{row}(\Delta) \perp V^r ~~~\text{and}~~~ \text{col}(\Delta)\perp U^r \right\}.
\end{split}
\end{equation}
Then, for an arbitrary (generic) matrix $M\in\mathbb{R}^{m_1\times m_2}$, its restriction on the subspace $\mathcal{S}(\Theta)$ and $\mathcal{S}^\perp(\Theta)$, denoted by $M_{\mathcal{S}(\Theta)}$ and $M_{\mathcal{S}^\perp(\Theta)}$ respectively, are given by:
\begin{equation*}
M_{\mathcal{S}_\Theta} = U \begin{bmatrix}
\widetilde{M}_{11} & \widetilde{M}_{12} \\ \widetilde{M}_{21} & O
\end{bmatrix}V' \qquad \text{and} \qquad M_{\mathcal{S}^\perp_\Theta} = U \begin{bmatrix}
O & O \\ O & \widetilde{M}_{22}
\end{bmatrix}V',
\end{equation*}
where $\Theta = U\Sigma V'$ and $\widetilde{M}$ is defined and partitioned as
\begin{equation*}
\widetilde{M} = U' M V = \begin{bmatrix}
\widetilde{M}_{11} & \widetilde{M}_{12} \\ \widetilde{M}_{21} & \widetilde{M}_{22}
\end{bmatrix}, \quad \text{where}~~\widetilde{M}_{11}\in\mathbb{R}^{r\times r}.
\end{equation*}
Note that by Lemma~\ref{lemma:decomposition}, $M_{\mathcal{S}_\Theta} + M_{\mathcal{S}^\perp_\Theta} = M$. Moreover, when $\Theta$ is restricted to the subspace induced by itself $\Theta_{\mathcal{S}_\Theta}$ (and we write $\Theta_{\mathcal{S}}$ for short for this specific case), the following decomposition for the nuclear norm holds:
\begin{equation*}
\vertiii{\Theta}_* = \vertiii{\Theta_\mathcal{S} + \Theta_{\mathcal{S}^\perp}}_* =  \vertiii{\Theta_\mathcal{S}}_* + \vertiii{\Theta_{\mathcal{S}^\perp}}_*.
\end{equation*}
Let $\mathcal{J}(\Theta)$ be the set of indexes in which $\Theta$ is nonzero. Analogously, we define 
\begin{equation}\label{defn:supportspace}
\begin{split}
\mathcal{R}_\Theta &:= \left\{ \Delta\in\mathbb{R}^{m_1\times m_2} | \Delta_{ij}=0~~\text{for}~~(i,j)\notin \mathcal{J}(\Theta) \right\}, \\
\mathcal{R}^c_\Theta & : = \left\{ \Delta\in\mathbb{R}^{m_1\times m_2} | \Delta_{ij}=0~~\text{for}~~(i,j)\in \mathcal{J}(\Theta) \right\}.
\end{split}
\end{equation}
Then, for an arbitrary matrix $M$, $M_{\mathcal{J}_\Theta} \in \mathcal{R}_\Theta$ is obtained by setting the entries of $M$ whose indexes are not in $\mathcal{J}(\Theta)$ to 0, and $M_{\mathcal{J}^c_\Theta}\in \mathcal{R}^c_\Theta$ is obtained by setting the entries of $M$ whose indexes are in $\mathcal{J}(\Theta)$ to 0. Then, the following decomposition holds:
\begin{equation*}
\vertii{ M_{\mathcal{J}_\Theta} + M_{\mathcal{J}^c_\Theta}}_1 = \vertii{ M_{\mathcal{J}_\Theta}}_1 + \vertii{M_{\mathcal{J}^c_\Theta}}_1.
\end{equation*}
\begin{proof}[\textbf{Proof of Theorem~\ref{thm:consistencyBC}}] Again for the ease of notation, in this proof, we drop the superscript and use $(\widehat{B}^{(0)},\widehat{C}^{(0)})$ to denote $(\widehat{B},\widehat{C})$ whenever there is no ambiguity. Define $\mathcal{Q}$ to be the weighted regularizer:
	\begin{equation*}
	\mathcal{Q}(B,C) = \vertiii{B}_* + \tfrac{\lambda_C}{\lambda_B}\vertii{C}_1.
	\end{equation*}
	Note that $(B^\star,C^\star)$ is always feasible, and by the optimality of $(\widehat{B},\widehat{C})$, the following inequality holds:
	\small
	\begin{equation*}
	\frac{1}{T}\smalliii{\mathcal{Z}^T - \mathfrak{W}_0(\widehat{B}_{\text{aug}} + \widehat{C}_{\text{aug}})}_F^2 + \lambda_B\smalliii{\widehat{B}}_* + \lambda_C\smallii{\widehat{C}}_1 \leq \frac{1}{T}\smalliii{\mathcal{Z}^T -\mathfrak{W}_0(B^\star + C^\star)}_F^2 + \lambda_B\vertiii{B^\star}_* + \lambda_C \smallii{C^\star}_1,
	\end{equation*}
	\normalsize
	By defining $\Delta^B_{\text{aug}}=\widehat{B}_{\text{aug}}-B_{\text{aug}}^\star = [\Delta^B, O]$, $\Delta_{\text{aug}}^C = \widehat{C}_{\text{aug}}-C_{\text{aug}}^\star=[O,\Delta^C]$, we obtain the following \textit{basic inequality}:
	\begin{equation}\label{BC:basic}
	\tfrac{1}{T}\vertiii{\mathfrak{W}_0(\Delta^B_{\text{aug}} + \Delta^C_{\text{aug}})}_F^2 \leq \tfrac{2}{T}\llangle \Delta^B_{\text{aug}} + \Delta^C_{\text{aug}}, \mathcal{W}'\mathcal{V} \rrangle + \lambda_B\mathcal{Q}(B^\star,C^\star) - \lambda_B \mathcal{Q}(\widehat{B},\widehat{C}).
	\end{equation}
	By H\"{o}lder's inequality and Lemma~\ref{lemma:inequality1}, we have 
	\small
	\begin{align}
	\tfrac{1}{T}\vertiii{\mathfrak{W}_0(\Delta^B_{\text{aug}} + \Delta^C_{\text{aug}})}_F^2 &\leq \tfrac{2}{T}\big(\vertiii{ \Delta^B_{\mathcal{S}_{B^\star}}}_* + \smalliii{\Delta^B_{\mathcal{S}_{B^\star}^\perp}}_*\big) \smalliii{\mathcal{W}'\mathcal{V}}_{\text{op}} + \tfrac{2}{T}\big(\smallii{\Delta^C_{\mathcal{J}^c_{C^\star}}}_1 + \smallii{\Delta^C_{\mathcal{J}^c_{C^\star}}}_1\big)\vertii{\mathcal{W}'\mathcal{V}}_\infty \nonumber \\
	& \quad + \lambda_B\mathcal{Q}(\Delta^B_{\mathcal{S}_{B^\star}},\Delta^C_{\mathcal{J}_{C^\star}}) -\lambda_B \mathcal{Q}(\Delta^B_{\mathcal{S}^\perp_{B^\star}},\Delta^C_{\mathcal{J}^c_{C^\star}} ). \label{BC:upperbound1}
	\end{align}
	\normalsize
	With the specified choice of $\lambda_B$ and $\lambda_C$, after some algebra, \eqref{BC:upperbound1} is further bounded by
	\begin{equation*}
	\frac{3\lambda_B}{2} \mathcal{Q}(\Delta^B_{\mathcal{S}_{B^\star}},\Delta^C_{\mathcal{J}_{C^\star}}) -\frac{\lambda_B}{2} \mathcal{Q}(\Delta^B_{\mathcal{S}^\perp_{B^\star}},\Delta^C_{\mathcal{J}^c_{C^\star}}).
	\end{equation*}
	By Lemma~\ref{lemma:RSC-result}, and using this upper bound, we obtain 
	\begin{equation*}
	\frac{\alpha_{\text{RSC}}}{2} (\vertiii{\Delta^B}_F^2 + \vertiii{\Delta^C}_F^2) - \frac{\lambda_B}{2}\mathcal{Q}(\Delta^B,\Delta^C) \leq \frac{3\lambda_B}{2} \mathcal{Q}(\Delta^B_{\mathcal{S}_{B^\star}},\Delta^C_{\mathcal{J}_{C^\star}}) -\frac{\lambda_B}{2} \mathcal{Q}(\Delta^B_{\mathcal{S}^\perp_{B^\star}},\Delta^C_{\mathcal{J}^c_{C^\star}}).
	\end{equation*}
	By the triangle inequality, $\mathcal{Q}(\Delta^B,\Delta^C)\leq \mathcal{Q}(\Delta^B_{\mathcal{S}_{B^\star}},\Delta^C_{\mathcal{J}_{C^\star}})+\mathcal{Q}(\Delta^B_{\mathcal{S}^\perp_{B^\star}},\Delta^C_{\mathcal{J}^c_{C^\star}})$, rearranging gives
	\begin{equation}\label{BC:inequality}
	\frac{\alpha_{\text{RSC}}}{2}  (\vertiii{\Delta^B}_F^2 + \vertiii{\Delta^C}_F^2)   \leq  2\lambda_B \mathcal{Q}(\Delta^B_{\mathcal{S}_{B^\star}},\Delta^C_{\mathcal{J}_{C^\star}}).
	\end{equation}
	By Lemma~\ref{lemma:decomposition}, with $N=B^\star$, $M_1=\Delta^B_{\mathcal{S}_{B^\star}}$, $M_2=\Delta^B_{\mathcal{S}^\perp_{B^\star}}$, we get 
	\begin{equation*}
	\text{rank}(\Delta^B_{\mathcal{S}_{B^\star}}) \leq 2r_B^\star \qquad \text{and} \qquad \llangle \Delta^B_{\mathcal{S}_{B^\star}}, \Delta^B_{\mathcal{S}^\perp_{B^\star}}\rrangle = 0,
	\end{equation*}
	which implies $\smalliii{\Delta^B_{\mathcal{S}_{B^\star}}}_* \leq (\sqrt{2r_B^\star})\smalliii{\Delta^B_{\mathcal{S}_{B^\star}}}_F\leq (\sqrt{2r_B^\star})\smalliii{\Delta^B}_F$. Since $\Delta^C_{\mathcal{J}_{C^\star}}$ has at most $s^\star_C$ nonzero entries, it follows that $\|\Delta^C_{\mathcal{J}_{C^\star}}\|_1\leq \sqrt{s^\star_C}\smalliii{\Delta^C_{\mathcal{J}_{C^\star}}}_F \leq \sqrt{s_C^\star}\smalliii{\Delta^C}_F$. Therefore,
	\begin{equation*}
	\mathcal{Q}(\Delta^B_{\mathcal{S}_{B^\star}},\Delta^C_{\mathcal{J}_{C^\star}}) = \lambda_B \vertiii{\Delta^B_{\mathcal{S}_{B^\star}}}_* + \lambda_C\vertii{\Delta^C_{\mathcal{J}_{C^\star}} }_1 \leq \lambda_B(\sqrt{2r_B^\star})\vertiii{\Delta^B}_F + \lambda_C(\sqrt{s_C^\star})\vertiii{\Delta^C}_F
	\end{equation*}
	With an application of the Cauchy-Schwartz inequality, \eqref{BC:inequality} yields:
	\begin{equation*}
	\frac{\alpha_{\text{RSC}}}{2}  (\vertiii{\Delta^B}_F^2 + \vertiii{\Delta^C}_F^2) \leq \sqrt{2r^\star_B\lambda_B^2 + s^\star_C\lambda_C^2} * \sqrt{\vertiii{\Delta^B}_F^2 + \vertiii{\Delta^C}_F^2 } 
	\end{equation*}
	and we obtain the following bound:
	\begin{equation*}
	\vertiii{\Delta^B}_F^2 + \vertiii{\Delta^C}_F^2 \leq 4 \left(2r_B^\star\lambda_B^2 + s_C^\star\lambda_C^2\right)/\alpha^2_{\text{RSC}}.
	\end{equation*}
\end{proof}

%
%
\begin{proof}[\textbf{Proof of Theorem~\ref{thm:algo1}}] At iteration 0, $\widehat{A}^{(0)}$ solves the following optimization problem:
	\begin{equation*}
	\widehat{A}^{(0)} = \argmin\limits_{A\in\mathbb{R}^{p_1\times p_1}}\big\{\tfrac{1}{T}\vertiii{\mathcal{X}^T - \mathcal{X}A'}_F^2 + \lambda_A\vertiii{A}_* \big\}.
	\end{equation*}
	By Theorem~\ref{thm:consistencyA}, its error bound is given by 
	\begin{equation*}
	\smallii{\widehat{A}^{(0)} - A^\star}_1  \leq 48 s_A^\star \lambda_A/\alpha_{\text{RSC}},
	\end{equation*}
	provided that $\widehat{\Gamma}^{(0)}_X = \mathrm{I}_{p_1}\otimes \mathcal{X}'\mathcal{X}/T$ satisfies the RSC condition, and the regularization parameter $\lambda_A$ satisfies $\lambda_A\geq 4\|\mathcal{X}'\mathcal{U}/T\|_\infty$. For random realizations $\mathcal{X}$ and $\mathcal{U}$, by Lemma~\ref{lemma:REX} and Lemma~\ref{lemma:DeviationBound1}, there exist constants $c_i$ and $c_i'$ such that for sample size $T\succsim s_A^\star \log p_1$, with probability at least $1-c_1\exp(-c_2 T\min\{1,\omega^{-2}\})$, where $\omega = c_3\mu_{\max}(\mathcal{A})/\mu_{\min}(\mathcal{A})$
	\begin{equation*}
	(\mathbf{E_1})\qquad \text{$\widehat{\Gamma}^{(0)}_X$ satisfies RSC condition with $\alpha_{\text{RSC}}=\Lambda_{\min}(\Sigma_u^\star)/(2\mu_{\max}(\mathcal{A}))$ },
	\end{equation*}
	and with probability at least $1-c'_1\exp(-c'_2\log p_1)$,
	\begin{equation*}
	(\mathbf{E_2}) \qquad \|\mathcal{X}'\mathcal{U}/T\|_\infty \leq C_0\sqrt{\frac{\log p_1}{T}}, \qquad \text{for some constant $C_0$.}
	\end{equation*}
	Hence with probability at least $1-c_1\exp(-c_2T)-c_1'\exp(-c_2'\log p_1)$, $$\|\widehat{A}^{(0)}-A^\star\|_1 = O\Big(s_A^\star\sqrt{\tfrac{\log p_1}{T}}\Big).$$ Moving onto $\widehat{\Omega}_u^{(0)}$, which is given by
	\begin{equation*}
	\widehat{\Omega}_u^{(0)} = \argmin\limits_{\Omega_u\in\mathbb{S}_{++}^{p_1\times p_1}} \left\{\log\det\Omega_u - \tr\big(\widehat{S}_u^{(0)}\Omega_u \big) + \rho_u \|\Omega_u\|_{1,\text{off}}\right\},
	\end{equation*}
	where $\widehat{S}_u^{(0)} = \tfrac{1}T(\mathcal{X}^T - \mathcal{X}\widehat{A}^{(0)'})'(\mathcal{X}^T - \mathcal{X}\widehat{A}^{(0)'})$. By Theorem 1 in \citet{ravikumar2011high}, the error bound for $\widehat{\Omega}_u^{(0)}$ relies on how well $\widehat{S}_u^{(0)}$ concentrates around $\Sigma_u^\star$, more specifically, $\|\widehat{S}_u^{(0)}-\Sigma_u^\star\|_\infty$. Note that 
	\begin{equation*}
	\|\widehat{S}_u^{(0)}-\Sigma_u^\star\|_\infty \leq  \|S_u-\Sigma_u^\star\|_\infty +  \|\widehat{S}_u^{(0)}-S_u\|_\infty,
	\end{equation*}
	where $S_u=\mathcal{U}'\mathcal{U}/T$ is the sample covariance based on true errors. For the first term, by \citet{ravikumar2011high}, there exists constant $\tau_0>2$ such that with probability at least $1-1/p_1^{\tau_0-2} = 1-\exp(-\tau\log p_1)~(\tau>0)$, the following bound holds:
	\begin{equation*}
	(\mathbf{E_3}) \qquad \|S_u - \Sigma_u^\star\|_\infty \leq C_1\sqrt{\frac{\log p_1}{T}}, \qquad \text{for some constant $C_1$.}
	\end{equation*}
	For the second term, 
	\begin{equation*}
	\widehat{S}_u^{(0)}-S_u = \tfrac{2}T\mathcal{U}'\mathcal{X}(A^\star - \widehat{A}^{(0)})' + (A^\star - \widehat{A}^{(0)})\left(\tfrac{\mathcal{X}'\mathcal{X}}T\right)(A^\star - \widehat{A}^{(0)})' := I_1 + I_2,
	\end{equation*}
	For $I_1$, based on the analysis of $\smallii{A^\star - \widehat{A}^{(0)}}_1$ and $\smallii{\mathcal{X}'\mathcal{U}/T}_\infty$,
	\begin{equation*}
	\|I_1\|_\infty \leq  2 \smalliii{A^\star - \widehat{A}^{(0)}}_\infty \smallii{\tfrac{1}T\mathcal{X}'\mathcal{U}}_\infty \leq 2\smallii{A^\star - \widehat{A}^{(0)}}_1 \smallii{\tfrac{1}T\mathcal{X}'\mathcal{U}}_\infty = O\Big(\tfrac{s_A^\star\log p_1}{T}\Big)
	\end{equation*}
	For $I_2$, 
	\begin{equation*}
	\begin{split}
	\|(A^\star - \widehat{A}^{(0)})\left(\tfrac{\mathcal{X}'\mathcal{X}}T\right)(A^\star - \widehat{A}^{(0)})'\|_\infty & \leq \smalliii{A^\star - \widehat{A}^{(0)}}_\infty\smalliii{A^\star - \widehat{A}^{(0)}}_1 \smallii{\tfrac{\mathcal{X}'\mathcal{X}}T}_\infty \\
	&\leq \smallii{A^\star - \widehat{A}^{(0)}}_1^2 \smallii{\tfrac{\mathcal{X}'\mathcal{X}}{T}}_\infty,
	\end{split}
	\end{equation*}
	where by Proposition 2.4 in \citet{basu2015estimation} and then taking a union bound, with probability at least $1-c_1''\exp(-c_2''\log p_1)~(c_1'',c_2''>0)$,
	\begin{equation*}
	(\mathbf{E_4})\qquad \smallii{\tfrac{\mathcal{X}'\mathcal{X}}{T}}_\infty \leq C_2\sqrt{\frac{\log p_1}T} + \Lambda_{\max}(\Gamma_X), \qquad \text{for some constant $C_2$.} 
	\end{equation*}
	Hence,
	\begin{equation*}
	\|I_2\|_\infty = O\Big( (s_A^\star)^2\big(\tfrac{\log p_1}{T}\big)^{3/2}\Big) + O\Big((s_A^\star)^2\tfrac{\log p_1}{T}\Big)
	\end{equation*}
	Combining all terms, and since we assume that $T^{-1}\log p_1$ is small,  $O(\sqrt{T^{-1}\log p_1})$ becomes the leading term, and the following bound holds with probability at least $1-c_1\exp(-c_2T) - c_1'\exp(-c_2'\log p_1) - c_1''\exp(-c_2''\log p_1) - \exp(-\tau \log p_1)$:
	\begin{equation*}
	\|\widehat{S}_u^{(0)}-\Sigma_u^\star\|_\infty  = O\Big(\sqrt{\tfrac{\log p_1}T}\Big).
	\end{equation*} 
	Consequently, 
	\begin{equation*}
	\|\widehat{\Omega}_u^{(0)} - \Omega_u^\star\|_\infty = O\Big(\sqrt{\tfrac{\log p_1}T}\Big).
	\end{equation*}
	At iteration 1, the vectorized $\widehat{A}^{(1)}$ solves 
	\begin{equation*}
	\widehat{\beta}_{A}^{(1)} = \argmin\limits_{\beta\in\mathbb{R}^{p_1^2}} \big\{ -2\beta' \widehat{\gamma}_X^{(1)} + \beta'\widehat{\Gamma}_X^{(1)}\beta + \lambda_A\|\beta\|_1 \big\},
	\end{equation*}
	where
	\begin{equation*}
	\widehat{\gamma}^{(1)}_X = \tfrac{1}T\big(\widehat{\Omega}_u^{(0)} \otimes \mathcal{X}'\big) \text{vec}(\mathcal{X}^T), \qquad \widehat{\Gamma}_X^{(1)} =  \widehat{\Omega}_u^{(0)}\otimes \tfrac{\mathcal{X}'\mathcal{X}}T.
	\end{equation*}
	The error bound for $\widehat{\beta}_{A}^{(1)}$ relies on (1) $\widehat{\Gamma}_X^{(1)}$ satisfying the RSC condition, which holds for sample size $T\succsim (d^{\max}_{\Omega_u^\star})^2\log p_1$ upon $\|\widehat{\Omega}_u^{(0)} - \Omega_u^\star\|_\infty=O(\sqrt{T^{-1}\log p_1})$; and (2) a bound for $\|\mathcal{X}'\mathcal{U}\widehat{\Omega}_u^{(0)}/T\|_\infty$. For 
	$\|\mathcal{X}'\mathcal{U}\widehat{\Omega}_u^{(0)}/T\|_\infty$, 
	\begin{equation*}
	\tfrac{1}T\mathcal{X}'\mathcal{U}\widehat{\Omega}_u^{(0)} = \tfrac{1}T\mathcal{X}'\mathcal{U}\Omega_u^\star + \tfrac{1}T\mathcal{X}'\mathcal{U}(\widehat{\Omega}_u^{(0)} - \Omega_u^\star) : = I_3 + I_4. 
	\end{equation*}
	For $I_3$, by Lemma 3 in \citet{lin2016penalized} and with the aid of Proposition 2.4 in \citet{basu2015estimation}, again with probability at least $1-c_1'''\exp(-c_2'''\log p_1)$ we get
	\begin{equation*}
	(\mathbf{E_5}) \qquad \vertii{\tfrac{1}T\mathcal{X}'\mathcal{U}\Omega_u^\star}_\infty \leq C_3\sqrt{\frac{\log p_1}T}, \qquad \text{for some constant $C_3$. }
	\end{equation*}
	For $I_4$, by Corollary 3 in \citet{ravikumar2011high}, we get
	\begin{equation*}
	\vertii{\tfrac{1}T\mathcal{X}'\mathcal{U}(\widehat{\Omega}_u^{(0)} - \Omega_u^\star)}_\infty \leq d_{\Omega^\star_u}^{\max}\smallii{\tfrac{1}T\mathcal{X}'\mathcal{U}}_\infty \smallii{\widehat{\Omega}_u^{(0)} - \Omega_u^\star}_\infty = O\Big(\tfrac{\log p_1}T\Big).
	\end{equation*}
	Combining all terms and taking the leading one, once again we have 
	\begin{equation*}
	\|\widehat{A}^{(1)} - A^\star\|_1 = O\Big(s_A^\star\sqrt{\tfrac{\log p_1}T}\Big),
	\end{equation*}
	which holds with probability at least $1-c_1\exp(-c_2T) - \tilde{c}_1\exp(-\tilde{c}_2\log p_1) - \exp(-\tau \log p_1)$, by letting $\tilde{c}_1=\max\{c_1',c_1'',c_1'''\}$ and $\tilde{c}_1=\min\{c_2',c_2'',c_2'''\}$. 
	It should be noted that up to this step, all sources of randomness from the random realizations have been captured by events from $\mathbf{E_1}$ to $\mathbf{E_5}$; thus, for $\widehat{\Omega}_u^{(1)}$ and iterations thereafter, the probability for which the bounds hold will no longer change, and the same holds for the error bounds for $\widehat{A}^{(k)}$ and $\widehat{\Omega}_u^{(k)}$ in terms of the relative order with respect to the dimension $p_1$ and sample size $T$. Therefore, we conclude that with high probability, for all iterations~$k$, 
	\begin{equation*}
	\|\mathcal{X}'\mathcal{U}\widehat{\Omega}_u^{(k)}/T\|_\infty = O\Big(\sqrt{\tfrac{ \log p_1}T}\Big), \qquad \| \widehat{S}_u^{(k)} - \Sigma_u^\star\|_\infty = O\Big( \sqrt{\tfrac{\log p_1}T}\Big).
	\end{equation*}
	With the aid of Theorem~\ref{thm:consistencyA}, it then follows that 
	\begin{equation*}
	\smalliii{\widehat{A}^{(k)}-A^\star}_F = O\Big(\sqrt{\tfrac{ s_A^\star\log p_1}T}\Big), \qquad \smalliii{\widehat{\Omega}_u^{(k)}-\Omega_u^\star}_F = O\Big( \sqrt{\tfrac{(s_{\Omega_u^\star}+p_1)\log p_1}T}\Big).
	\end{equation*}
\end{proof}

%
%
\begin{proof}[\textbf{Proof of Theorem~\ref{thm:algo2}}] At iteration 0, $(\widehat{B}^{(0)},\widehat{C}^{(0)})$ solves the following optimization:
	\begin{equation*}
	(\widehat{B}^{(0)},\widehat{C}^{(0)}) = \argmin\limits_{(B,C)}\left\{ \tfrac{1}{T}\vertiii{\mathcal{Z}^T - \mathcal{X}B' - \mathcal{Z}C'}_F^2 + \lambda_B\vertiii{B}_* + \lambda_C \vertii{C}_1 \right\}.
	\end{equation*}	
	Let $W_t=(X_t',Z_t')'\in\mathbb{R}^{p_1+p_2}$ be the joint process and $\mathcal{W}$ be the realizations, with operators $\mathfrak{W}_0$ identically defined to that in Theorem~\ref{thm:consistencyBC} . By Theorem~\ref{thm:consistencyBC}, 
	\begin{equation*}
	\smalliii{\widehat{B}^{(0)} - B^\star}_F^2 + \smalliii{\widehat{C}^{(0)} - C^\star}_F^2 \leq 4(2r_B^\star \lambda_B^2 + s_C^\star \lambda_C^2)/\alpha^2_{\text{RSC}},
	\end{equation*}
	provided that $\mathfrak{W}$ satisfies the RSC condition and $\lambda_B$, $\lambda_C$ respectively satisfy
	\begin{equation*}
	\lambda_B\geq 4\vertiii{\mathcal{W}'\mathcal{V}/T}_{\text{op}} \quad \text{and} \quad \lambda_C \geq 4\vertii{\mathcal{W}'\mathcal{V}/T}_\infty.
	\end{equation*}
	In particular, by Lemma~\ref{lemma:RSCW} for random realizations of $\mathcal{X}$, $\mathcal{Z}$ and $\mathcal{V}$, for sample size $T\succsim c_0(p_1+2p_2)$, with probability at least $1-c_1\exp\{-c_2(p_1+p_2)\}$, 
	\begin{equation*}
	(\mathbf{E'_1}) \qquad \mathfrak{W}_0 \text{ satisfies the RSC condition}.
	\end{equation*}
	By Lemma~\ref{lemma:operator-infinity}, for sample size $T\succsim (p_1+2p_2)$ and some constant $C_1,C_2>0$,
	\begin{equation*}
	(\mathbf{E'_2})\qquad \vertiii{\mathcal{W}'\mathcal{V}/T}_{\text{op}} \leq C_1\sqrt{\frac{p_1+2p_2}T} \quad \text{and} \quad 
	\vertii{\mathcal{W}'\mathcal{V}/T}_\infty \leq C_2\sqrt{\frac{\log(p_1+p_2)+\log p_2}T},
	\end{equation*}
	with probability at least $1-c_1'\exp\{-c_2'(p_1+2p_2)\}$ and $1-c_1''\exp\{-c_2''\log[p_2(p_1+p_2)]\}$, respectively. 
	Hence, with probability at least 
	\begin{equation*}
	1 - c_1\exp\{-c_2(p_1+p_2)\} - c_1'\exp\{-c_2'(p_1+2p_2)\} - c_1''\exp\{-c_2''\log[p_2(p_1+p_2)]\},
	\end{equation*} 
	the following bound holds for the initializers as long as sample size $T\succsim (p_1+2p_2)$:
	\begin{equation}\label{eqn:BC0}
	\smalliii{\widehat{B}^{(0)} - B^\star}_F^2 + \smalliii{\widehat{C}^{(0)} - C^\star}_F^2 = O\Big(\tfrac{p_1+2p_2}T\Big) + O\Big(\tfrac{\log(p_1+p_2)+\log p_2}T\Big).
	\end{equation}
	Considering the estimation of $\widehat{\Omega}_v^{(0)}$, it solves a graphical Lasso problem:
	\begin{equation*}
	\widehat{\Omega}_v^{(0)} = \argmin\limits_{\Omega_v\in\mathbb{S}_{++}^{p_2\times p_2}} \left\{\log\det\Omega_v - \tr\big(\widehat{S}_u^{(0)}\Omega_v \big) + \rho_v \|\Omega_v\|_{1,\text{off}}\right\},
	\end{equation*}
	where $\widehat{S}_v^{(0)} = \tfrac{1}T(\mathcal{Z}^T - \mathcal{X}\widehat{B}^{(0)'} - \mathcal{Z}\widehat{C}^{(0)'})'(\mathcal{Z}^T - \mathcal{X}\widehat{B}^{(0)'} - \mathcal{Z}\widehat{C}^{(0)'})$. Similar to the proof of Theorem~\ref{thm:algo1}, the error bound for $\widehat{\Omega}_v^{(0)}$ depends on $\|\widehat{S}_v^{(0)}-\Sigma_v^\star\|_\infty$, which can be decomposed~as
	\begin{equation*}
	\|\widehat{S}_v^{(0)}-\Sigma_v^\star\|_\infty \leq  \|S_v-\Sigma_v^\star\|_\infty +  \|\widehat{S}_v^{(0)}-S_v\|_\infty,
	\end{equation*}
	where $S_v=\mathcal{V}'\mathcal{V}/T$ is the sample covariance based on the true errors. For the first term, by Lemma 1 in \citet{ravikumar2011high}, there exists constant $\tau_0>2$ such that with probability at least $1-1/p_2^{\tau_0-2}=1-\exp(-\tau \log p_2)~(\tau>0)$, the following bound holds:
	\begin{equation*}
	(\mathbf{E'_3}) \qquad \|S_v - \Sigma_v^\star\|_\infty \leq C_3\sqrt{\frac{\log p_1}T}, \qquad \text{for some constant $C_3$.}
	\end{equation*}
	For the second term, let $\Pi=[B,C]\in\mathbb{R}^{p_2\times (p_1+p_2)}$, then
	\begin{equation*}
	\widehat{S}_v^{(0)}-S_v = \frac{2}T\mathcal{V}'\mathcal{W}(\Pi^\star - \widehat{\Pi}^{(0)})' + (\Pi^\star - \widehat{\Pi}^{(0)})\left(\tfrac{\mathcal{W}'\mathcal{W}}T\right)(\Pi^\star - \widehat{\Pi}^{(0)})' := I_1 + I_2,
	\end{equation*}
	For $I_1$, we have 
	\begin{equation*}
	\begin{split}
	\smallii{\tfrac{2}T\mathcal{V}'\mathcal{W}(\Pi^\star - \widehat{\Pi}^{(0)})'}_\infty \leq \smallii{\tfrac{2}T\mathcal{V}'\mathcal{W}(\Pi^\star - \widehat{\Pi}^{(0)})'}_F \leq 2  \smalliii{\tfrac{1}T\mathcal{W}'\mathcal{V}}_{\text{op}} \smalliii{\Pi^\star - \widehat{\Pi}^{(0)}}_F.
	\end{split}
	\end{equation*}
	Consider the leading term of $\smalliii{\Pi^\star - \widehat{\Pi}^{(0)}}_F$ as in~\eqref{eqn:BC0}, whose rate is 
	$O(\sqrt{T^{-1}(p_1+2p_2)})$. We therefore obtain 
	\begin{equation*}
	\|I_1\|_\infty \leq \|I_1\|_F  = O\Big(\tfrac{p_1+2p_2}T\Big).
	\end{equation*}
	Similarly for $I_2$, 
	\begin{equation*}
	\|I_2\|_\infty \leq \|I_2\|_F \leq \smalliii{\Pi^\star - \widehat{\Pi}^{(0)}}_F^2 \smalliii{\tfrac{\mathcal{W}'\mathcal{W}}T}_{\text{op}},
	\end{equation*}
	where with a similar derivation to that in Lemma~\ref{lemma:boundoperator}, for sample size $T\succsim (p_1+p_2)$, with probability at least $1-c_1'''\exp\{-c_2'''(p_1+p_2) \}$,  we get
	\begin{equation*}
	(\mathbf{E_4'})\qquad \smalliii{\tfrac{\mathcal{W}'\mathcal{W}}T}_{\text{op}}\leq C_4\sqrt{\frac{p_1+2p_2}T} + \Lambda_{\max}(\Gamma_X), \qquad \text{for some constant $C_4$.}
	\end{equation*}
	Hence,
	\begin{equation*}
	\|I_2\|_\infty \leq \|I_2\|_F \leq  = O\Big( \big(\frac{p_1+2p_2}T\big)^{3/2} \Big).
	\end{equation*}
	Combining all terms and then taking the leading one, with probability at least 
	\begin{equation*}
	\begin{split}
	1 - c_1\exp\{-c_2(p_1+p_2)\} - c_1'\exp\{-c_2'(p_1+2p_2)\} - c_1''\exp\{-c_2''\log[p_2(p_1+p_2)]\} \\
	- c_1'''\exp\{-c_2'''(p_1+p_2)\} - \exp(-\tau \log p_2),
	\end{split}
	\end{equation*}
	we obtain 
	\begin{equation*}
	\|\widehat{S}_v^{(0)}-\Sigma_v^\star\|_\infty  = O\Big(\sqrt{\tfrac{p_1+2p_2}T}\Big). 
	\end{equation*}
	Note that here with the required sample size, $(p_1+2p_2)/T$ is a small quantity, and therefore
	\begin{equation*}
	O\Big( \big(\tfrac{p_1+2p_2}T\big)^{3/2} \Big)\leq O\Big(\tfrac{p_1+2p_2}T\Big) \leq O\Big(\sqrt{\tfrac{p_1+2p_2}T}\Big).
	\end{equation*}
	At iteration 1, the bound of $\smalliii{\widehat{B}^{(1)}-B^\star}_F^2+\smalliii{\widehat{C}^{(1)}-C^\star}_F^2$ relies on the following two quantities:
	\begin{equation*}
	\vertiii{\tfrac{1}T \mathcal{W}'\mathcal{V}\widehat{\Omega}^{(0)}_v}_{\text{op}} \qquad \text{and} \qquad \vertii{\tfrac{1}T\mathcal{W}'\mathcal{V}\widehat{\Omega}^{(0)}_v}_{\infty}.
	\end{equation*}
	Using a similar derivation to that in the proof of Theorem~\ref{thm:algo1}, 
	\begin{equation}\label{eqn:term1}
	\vertii{\tfrac{1}T\mathcal{W}'\mathcal{V} \widehat{\Omega}_v^{(0)}}_{\infty} \leq \vertii{\tfrac{1}T \mathcal{W}'\mathcal{V} (\widehat{\Omega}_v^{(0)}-\Omega_v^\star)}_\infty + \vertii{\tfrac{1}T\mathcal{W}'\mathcal{V}\Omega_v^\star}_\infty,
	\end{equation}	
	where by viewing $\mathcal{V}\Omega_v^\star $ as some random realization coming from a certain sub-Gaussian process, with probability at least $1-\bar{c}_1''\exp\{-\bar{c}_2''\log[p_2(p_1+p_2)]\}$, we get
	\begin{equation*}
	(\mathbf{E_5'}) \qquad \vertii{\tfrac{1}T\mathcal{W}'\mathcal{V}\Omega_v^\star}_\infty\leq C_5\sqrt{\frac{\log(p_1+p_2)+\log p_2}T}, \qquad \text{for some constant $C_5$},
	\end{equation*}
	and
	\begin{equation*}
	\begin{split}
	\vertii{\tfrac{1}T \mathcal{W}'\mathcal{V} (\widehat{\Omega}_v^{(0)}-\Omega_v^\star)}_\infty &\leq d^{\Omega_v^\star}_{\max}\smallii{\tfrac{1}T\mathcal{W}'\mathcal{V}}_\infty \smallii{\widehat{\Omega}_v^{(0)}-\Omega_v^\star}_\infty\\
	& =O\Big(\sqrt{\tfrac{\log(p_1+p_2)+\log p_2}T}\Big)\cdot O\Big(\sqrt{\tfrac{p_1+2p_2}T}\Big).
	\end{split}
	\end{equation*}
	For $\smalliii{\tfrac{1}T\mathcal{W}' \mathcal{V}\widehat{\Omega}_v^{(0)}}_{\text{op}}$, similarly we have
	\begin{equation}\label{eqn:term2}
	\vertiii{\tfrac{1}T\mathcal{W}'\mathcal{V}\widehat{\Omega}_v^{(0)}}_{\text{op}} \leq \smalliii{\tfrac{1}T\mathcal{W}'\mathcal{V}(\widehat{\Omega}_v^{(0)}-\Omega_v^\star)}_{\text{op}} + \smalliii{\tfrac{1}T\mathcal{W}'\mathcal{V}\Omega_v^\star}_{\text{op}},
	\end{equation}
	where with probability at least $1- \bar{c}_1'\exp\{-\bar{c}_2'(p_1+p_2)\}$,
	\begin{equation*}
	(\mathbf{E_6'}) \qquad \vertiii{\tfrac{1}T\mathcal{W}'\mathcal{V}\Omega_v^\star}_{\text{op}}\leq C_6\sqrt{\frac{p_1+2p_2}T} \qquad \text{for some constant $C_6$},
	\end{equation*}
	and
	\begin{equation*}
	\begin{split}
	\vertiii{\tfrac{1}T\mathcal{W}'\mathcal{V}(\widehat{\Omega}_v^{(0)}-\Omega_v^\star)}_{\text{op}} & \leq \smalliii{\tfrac{1}T\mathcal{W}'\mathcal{V}}_{\text{op}} \smalliii{\widehat{\Omega}_v^{(0)}-\Omega_v^\star}_{\text{op}} \\
	& \leq  \smalliii{\tfrac{1}T\mathcal{W}'\mathcal{V}}_{\text{op}}\left[ d^{\Omega_v^\star}_{\max} \smallii{\widehat{\Omega}_v^{(0)}-\Omega_v^\star}_\infty\right] = O\Big(\tfrac{p_1+2p_2}T\Big),
	\end{split}
	\end{equation*}
	where the second inequality follows from Corollary 3 of \citet{ravikumar2011high}.  Combining all terms from~\eqref{eqn:term1} and~\eqref{eqn:term2}, the leading term gives the following bound:
	\begin{equation*}
	\smalliii{\widehat{B}^{(1)}-B^\star}_F^2+\smalliii{\widehat{C}^{(1)}-C^\star}_F^2  \leq C_7\Big(\frac{p_1+2p_2}T\Big)  \qquad \text{for some constant $C_7$},
	\end{equation*}
	and this error rate coincides with that in the bound of $\smalliii{\widehat{B}^{(0)}-B^\star}_F^2+\smalliii{\widehat{C}^{(0)}-C^\star}_F^2$. This implies that for $\widehat{\Omega}_v^{(1)}$ and iterations thereafter, the error rate remains unchanged. Moreover, all sources of randomness have been captured up to this step in events $\mathbf{E_1'}$ to $\mathbf{E_6'}$, and therefore the probability for the bounds to hold no longer changes. Consequently, the following bounds hold for all iterations~$k$:
	\begin{equation*}
	\|\mathcal{W}'\mathcal{V}\widehat{\Omega}_v^{(k)}/T\|_\infty = \smalliii{\mathcal{W}'\mathcal{V}\widehat{\Omega}_v^{(k)}/T}_\text{op} = O\Big(\sqrt{\tfrac{p_1+2p_2}T}\Big)
	\end{equation*}
	and
	\begin{equation*}
	\| \widehat{S}_v^{(k)} - \Sigma_v^\star\|_\infty = O\Big(\sqrt{\tfrac{p_1+2p_2}T}\Big),
	\end{equation*}
	with probability at least 
	\begin{equation*}
	1- c_0\exp\{-\tilde{c}_0(p_1+p_2)\} - c_1\exp\{-\tilde{c}_1(p_1+2p_2)\} - c_2\exp\{-\tilde{c}_2 \log[p_2(p_1+p_2)]\} - \exp\{-\tau \log p_2\}.
	\end{equation*}
	for some new positive constants $c_i,\tilde{c}_i~(i=0,1,2)$ and $\tau$.\footnote{Here we slightly abuse the notations and redefine $c_0:=\max\{c_1,c_1'''\}$, $c_1:=\max\{c_1',\bar{c}_1'\}$, $\tilde{c}_1:=\min\{c_2',\bar{c}_2'\}$, $c_2=\max\{c_1'',\bar{c}_1''\}$, $\tilde{c}_2:=\min\{c_2'',\bar{c}_2''\}$. }
	The above bounds directly imply the bound in the statement in Theorem~\ref{thm:algo2}, with the aid of Theorem~\ref{thm:consistencyBC}.
\end{proof}

\bigbreak
\section{Key Lemmas and Their Proofs.} \label{appendix:KeyLemmas}

In this section, we verify the conditions that are required for establishing the consistency results in Theorem~\ref{thm:consistencyA} and~\ref{thm:consistencyBC}, under random realizations of $\mathcal{X}$, $\mathcal{Z}$, $\mathcal{U}$ and $\mathcal{V}$. 

The following two lemmas verify the conditions for establishing the consistency properties for $\widehat{A}^{(0)}$. Specifically, Lemma~\ref{lemma:REX} establishes that with high probability, $\mathfrak{X}_0$ satisfies the RSC condition. Further, Lemma~\ref{lemma:DeviationBound1} gives a high probability upper bound for $\|\mathcal{X}'\mathcal{U}/T\|_\infty$ for random $\mathcal{X}$ and $\mathcal{U}$. 
\begin{lemma}[Verification of the RSC condition] \label{lemma:REX} For the $\VAR(1)$ model $\{X_t\}$ posited in~\eqref{eqn:model0}, there exist $c_i>0~(i=1,2,3)$ such that for sample size $T\succsim \max\{\omega^2,1\} s_A^\star \log p_1$,  with probability at least $$1-c_1\exp\left[-c_2 T\min\{1,\omega^{-2}\}\right],\qquad \omega = c_3\frac{\Lambda_{\max}(\Sigma_u)\mu_{\max}(\mathcal{A})}{\Lambda_{\min}(\Sigma_u)\mu_{\min}(\mathcal{A})},$$ the following inequality holds
	\begin{equation*}
	\frac{1}{2T}\vertiii{\mathfrak{X}_0(\Delta)}_F^2 \geq \alpha_{\text{RSC}} \vertiii{\Delta}_F^2 - \tau \vertii{\Delta}_1^2, \qquad \text{for }\Delta\in \mathbb{R}^{p_1\times p_1}, 
	\end{equation*}
	where $\alpha_{\text{RSC}}=\frac{\Lambda_{\min}(\Sigma_u)}{\mu_{\max}(\mathcal{A})}$, $\tau=4\alpha_{\text{RSC}}\max\{\omega^2,1\}\log p_1/T$.
\end{lemma}
\begin{proof}[\textbf{Proof of Lemma~\ref{lemma:REX}}]
	For the specific $\VAR(1)$ process $\{X_t\}$ given in~\eqref{eqn:model0}, using  Proposition 4.2 in \citet{basu2015estimation} with $d=1$ directly gives the result. Specifically, we note that by letting $\theta=\text{vec}(\Delta)$, 
	\begin{equation*}
	\tfrac{1}{T}\vertiii{\mathfrak{X}_0(\Delta)}_F^2 = \theta' \widehat{\Gamma}^{(0)}_X \theta,
	\end{equation*}
	where $\widehat{\Gamma}^{(0)}_X = \mathrm{I}_{p_1}\otimes (\mathcal{X}'\mathcal{X}/T)$, and $\|\theta\|_2^2 = \vertiii{\Delta}_F^2$, $\|\theta\|_1 = \vertii{\Delta}_1$. 
\end{proof}
\begin{lemma}[Verification of the deviation bound]\label{lemma:DeviationBound1}
	For the model in~\eqref{eqn:model0}, there exist constants $c_i>0, ~i=0,1,2$ such that for $T\succsim 2\log p_1$, with probability at least $1-c_1\exp(-2c_2\log p_1)$, the following bound holds:
	\begin{equation}
	\|\mathcal{X}'\mathcal{U}/T\|_\infty \leq c_0\Lambda_{\max}(\Sigma_u)\left[ 1 + \frac{1}{\mu_{\min}(\mathcal{A})} +  \frac{\mu_{\max}(\mathcal{A})}{ \mu_{\min}(\mathcal{A})}\right]\sqrt{\frac{2\log p_1}{T}}.
	\end{equation}
\end{lemma}
\begin{proof}[\textbf{Proof of Lemma~\ref{lemma:DeviationBound1}}]
	First, we note that, 
	\begin{equation*}
	\vertii{\mathcal{X}'\mathcal{U}/T}_\infty = \max\limits_{\substack{1\leq i\leq p_1\\1\leq j\leq p_1}}\left| e_i' \left(\mathcal{X}'\mathcal{U}/T\right) e_j \right|.
	\end{equation*}
	Applying Proposition~2.4(b) in \citet{basu2015estimation} for an arbitrary pair of $(i,j)$ gives:
	\begin{equation*}
	\mathbb{P}\left( \left| e_i' \left(\mathcal{X}'\mathcal{U}/T\right) e_j \right| >  \eta \left[\Lambda_{\max}(\Sigma_u)\left( 1 + \frac{1}{\mu_{\min}(\mathcal{A})} +  \frac{\mu_{\max}(\mathcal{A})}{ \mu_{\min}(\mathcal{A})}\right) \right]\right) \leq 6\exp[-cT\min\{\eta,\eta^2\}]. 
	\end{equation*}
	Setting $\eta = c_0\sqrt{2\log p_1/T}$ and taking a union bound over all $1\leq i\leq p_1,1\leq j\leq p_1$, we get that for some $c_1,c_2>0$, with probability at least $1-c_1\exp[-2c_2\log p_1]$, 
	\begin{equation*}
	\max\limits_{\substack{1\leq i\leq p_1\\1\leq j\leq p_1}}\left| e_i' \left(\mathcal{X}'\mathcal{U}/T\right) e_j \right| \leq c_0\Lambda_{\max}(\Sigma_u)\left[ 1 + \frac{1}{\mu_{\min}(\mathcal{A})} +  \frac{\mu_{\max}(\mathcal{A})}{ \mu_{\min}(\mathcal{A})}\right]\sqrt{\frac{2\log p_1}{T}}.
	\end{equation*}
\end{proof}
In the next two lemmas, Lemma~\ref{lemma:RSCW} gives an RSC curvature that holds with high probability for $\mathfrak{W}$ induced by a random $\mathcal{W}$, and Lemma~\ref{lemma:operator-infinity} gives a high probability upper bound for $\vertiii{W'\mathcal{V}/T}_{\text{op}}$ and $\vertii{W'\mathcal{V}/T}_\infty$. 
\begin{lemma}[Verification of the RSC condition]\label{lemma:RSCW}
	Consider the covariance stationary process $W_t=(X'_t,Z'_t)'\in\mathbb{R}^{p_1+p_2}$ whose spectral density exists. Suppose $\mathfrak{m}(f_W)>0$. There exist constants $c_i>0, i=1,2,3$ such that with probability at least $1-2c_1\exp(-c_2(p_1+p_2))$, the RSC condition for $\mathfrak{W}$ induced by a random $\mathcal{W}$ holds for $\alpha_{\text{RSC}}$ and tolerance 0, where 
	\begin{equation*}
	\alpha_{\text{RSC}} = \pi \mathfrak{m}(f_W)/4,
	\end{equation*}
	whenever $T\succsim c_3(p_1+p_2)$.
\end{lemma}
\begin{proof}[\textbf{Proof of Lemma~\ref{lemma:RSCW}}]
	First ,we note that the following inequality holds, for any $\mathcal{W}$:
	\begin{equation}\label{RSCinquality}
	\frac{1}{2T}\vertiii{\mathfrak{W}_0(\Delta)}_F^2 = \frac{1}{2T} \vertiii{\mathcal{W}'\Delta}_F^2  = \frac{1}{2T}\sum_{j=1}^{p_2} \vertii{[\mathcal{W}'\Delta]_j}_2^2 \geq \frac{1}{2} \Lambda_{\min}\left(\widehat{\Gamma}^{(0)}_W\right)\vertiii{\Delta}_F^2.
	\end{equation}
	where $\widehat{\Gamma}^{(0)}_W = \mathcal{W}'\mathcal{W}/T$. Applying Lemma 4 in \citet{negahban2011estimation} on $\mathcal{W}$ together with Proposition 2.3 in \citet{basu2015estimation}, the following bound holds with probability at least $1-2c_1\exp[-c_2(p_1+p_2)]$, as long as $T\succsim c_3(p_1+p_2)$:
	\begin{equation*}
	\Lambda_{\min}\left(\widehat{\Gamma}^{(0)}_W\right) \geq \frac{\Lambda_{\min}(\Gamma_W(0)) }{4} \geq \frac{\pi}{2}\mathfrak{m}(f_W),
	\end{equation*}
	where $\Gamma_W(0)=\mathbb{E}W_tW_t'$. Combining with~\eqref{RSCinquality}, the RSC condition holds with $\kappa(\mathfrak{W}) = \pi \mathfrak{m}(f_W)/4$. 
\end{proof}
\begin{lemma}[Verification of the deviation bound]\label{lemma:operator-infinity}
	There exist constants $c_i>0$ and $c_i'>0, i=1,2,3$ such that the following statements hold:
	\begin{enumerate}[(a)]
		\item With probability at least $1-c_1\exp[-c_2(p_1+2p_2)]$, as long as $T\succsim c_3(p_1+2p_2)$, 
		\begin{equation}\label{OperatorBound1}
		\vertiii{\mathcal{W}'\mathcal{V}/T}_{\text{op}} \leq c_0\left[\mathcal{M}(f_W) + \Lambda_{\max}(\Sigma_v) + \mathcal{M}(f_{W,V})\right]\sqrt{\frac{p_1+2p_2}{T}}.
		\end{equation}
		\item With probability at least $1-c_1'\exp(-c'_2\log (p_1+p_2) - c'_2\log p_2)$, as long as $T\succsim c_3'\log [(p_1+p_2)p_2]$, 
		\begin{equation}\label{DeviationBound2}
		\vertii{\mathcal{W}'\mathcal{V}/T}_\infty \leq c_0'\left[\mathcal{M}(f_W) + \Lambda_{\max}(\Sigma_v) + \mathcal{M}(f_{W,V})\right]\sqrt{\frac{\log(p_1+p_2)+\log p_2 }{T}}.
		\end{equation}
	\end{enumerate} 
\end{lemma}
\begin{proof}[\textbf{Proof of Lemma~\ref{lemma:operator-infinity}}]
	(a) is a direct application of Lemma~\ref{lemma:boundoperator} on processes $\{W_t\}\in\R^{(p_1+p_2)}$ and $\{V_t\}\in\R^{p_2}$, and (b) is a direct application of Lemma~\ref{lemma:DeviationBound1}. 
\end{proof}

\bigbreak
\section{Auxiliary Lemmas and Their Proofs.}\label{sec:aux-lemma}
\begin{lemma}\label{lemma:J3cone}
	Let $\widehat{\beta}$ be the solution to the following optimization problem: 
	\begin{equation*}
	\widehat{\beta} = \argmin\limits_{\beta \in\mathbb{R}^p} \left\{ \tfrac{1}{2n} \vertii{Y-X\beta}_2^2 + \lambda_n \|\beta\|_1 \right\},
	\end{equation*}
	where the data is generated according to $Y = X\beta^\star + E$ with $X\in\mathbb{R}^{n\times p}$ and $E\in\mathbb{R}^{n}$.
	The errors $E_i$ are assumed to be i.i.d. for all $i=1,\cdots,n$. Then, by choosing $\lambda_n\geq 2\|X'E/n\|_\infty$, the error vector $\Delta := \widehat{\beta} - \beta^\star$ satisfies $\|\Delta_{\mathcal{J}}\|_1\leq 3\|\Delta_{\mathcal{J}^c}\|_1$. 
\end{lemma}
\begin{proof}
	Note that $\beta^\star$ is always feasible. By the optimality of $\widehat{\beta}$, we have
	\begin{equation*}
	\frac{1}{2n}  \vertii{Y-X\widehat{\beta}}_2^2 + \lambda_n \|\widehat{\beta} \|_1 \leq \frac{1}{2n} \vertii{Y-X\beta^\star}_2^2 + \lambda_n \|\beta^\star\|_1 
	\end{equation*}
	which after some algebra, can be simplified to
	\begin{equation*}
	\begin{split}
	\frac{1}{2} \Delta' \left( \frac{X'X}{n}\right)\Delta & \leq \langle \Delta',  \frac{1}{n} X'E\rangle + \lambda_n \|\beta^\star\|_1 - \lambda_n \|\beta^\star + \Delta\|_1 \\
	& \leq \|\Delta\|_1 \|\frac{1}{n} X'E\|_\infty + \lambda_n \|\beta^\star\|_1 - \lambda_n \|\beta^\star + \Delta\|_1,
	\end{split} 
	\end{equation*}
	with the second inequality obtained through an application of from H\"{o}lder's inequality. With the specified choice of $\lambda_n$, if follows that 
	\begin{equation*}
	\begin{split}
	0 & \leq \frac{\lambda_n}{2}\|\Delta\|_1 + \lambda_n \left\{ \|\beta_{\mathcal{J}}^\star\|_1 - \|\beta^\star_{\mathcal{J}} + \Delta_{\mathcal{J}}  + \beta^\star_{\mathcal{J}^c} + \Delta_{\mathcal{J}^c}\|_1 \right\} \\
	& = \frac{\lambda_n}{2} (\|\Delta_{\mathcal{J}}\|_1 + \|\Delta_{\mathcal{J}^c}\|_1 ) + \lambda_n\{\| \beta^\star_{\mathcal{J}} \|_1 - \|\beta^\star_{\mathcal{J}} + \Delta_{\mathcal{J}} \|_1 -  \|\Delta_{\mathcal{J}^c}\|_1   \} \quad \text{($\beta^\star_{\mathcal{J}^c}=0$, $\ell_1$ norm decomposable)} \\
	& \leq \frac{\lambda_n}{2} (\|\Delta_{\mathcal{J}}\|_1 + \|\Delta_{\mathcal{J}^c}\|_1) + \lambda_n (\|\Delta_{\mathcal{J}} \|_1 - \|\Delta_{\mathcal{J}^c}\| ) \qquad \text{(triangle inequality)} \\
	& = \frac{3\lambda_n}{2} \|\Delta_{\mathcal{J}} \|_1 - \frac{\lambda_n}{2} \|\Delta_{\mathcal{J}^c} \|_1,
	\end{split}
	\end{equation*}
	which implies $\|\Delta_{\mathcal{J}}\|_1\leq 3\|\Delta_{\mathcal{J}^c}\|_1$.
\end{proof}

\begin{lemma}\label{lemma:DeviationBound0}
	Consider two centered stationary Gaussian processes $\{X_t\}$ and $\{Z_t\}$. Further, assume that the spectral density of the joint process $\{(X_t',Z_t')'\}$ exists. Denote their cross-covariance by $\Gamma_{X,Z}(\ell):=\mathrm{Cov}(X_t,Z_{t+\ell})$, and their cross-spectral density is defined as
	\begin{equation*}
	f_{X,Z}(\theta) := \frac{1}{2\pi} \sum_{\ell=-\infty}^\infty \Gamma_{X,Z}(\ell)e^{-i\ell\theta}, \qquad \theta\in[-\pi,\pi],
	\end{equation*}
	whose upper extreme is given by:
	\begin{equation*}
	\mathcal{M}(f_{X,Z}) = \mathrm{esssup}_{\theta\in[-\pi,\pi]} \sqrt{ \Lambda_{\max} \left( f^*_{X,Z}(\theta)f_{X,Z}(\theta)\right)}.
	\end{equation*}	
	Let $\mathcal{X}$ and $\mathcal{Z}$ be data matrices with sample size $n$. Then, there exists a constant $c>0$, such that for any $u,v\in\mathbb{R}^p$ with $\|u\|\leq 1$, $\|v\|\leq1$, we have
	\begin{equation*}
	\mathbb{P}\left[ \left|u'\left(\frac{\mathcal{X}'Z}{T} -  \mathrm{Cov}(X_t,Z_t)\right)v\right| > 2\pi\left( \mathcal{M}(f_X) + \mathcal{M}(f_Z) + \mathcal{M}(f_{X,Z}) \right)\eta   \right] \leq 6\exp\left(-cT\min\{\eta,\eta^2\}\right).
	\end{equation*}
\end{lemma}

\begin{proof}
	Let $\xi_t = \langle u,X_t\rangle $, $\eta_t = \langle v, Z_t \rangle$. Let $f_X(\theta),f_Z(\theta)$ denote the spectral density of $\{X_t\}$ and $\{Z_t\}$, respectively. Then, the spectral density of $\{\xi_t\}$ and $\{\eta_t\}$, respectively, is $f_\xi(\theta) = u'f_X(\theta)u$, $f_\eta(\theta) = v'f_Z(\theta)v$. Also, we note that $\mathcal{M}(f_\xi)\leq\mathcal{M} (f_X)$, $\mathcal{M}(f_\eta) \leq \mathcal{M}(f_Z)$. Then,
	\begin{equation}\label{eqn:HW0}
	\begin{split}
	\frac{2}{T}\left[ \sum_{t=0}^T \xi_t\eta_t - \mathrm{Cov}(\xi_t,\eta_t)  \right] & = \left[\frac{1}{T}\sum_{t=0}^T (\xi_t+\eta_t)^2 - \mathrm{Var}(\xi_t + \eta_t)  \right]  \\
	& \quad - \left[\frac{1}{T}\sum_{t=0}^T (\xi_t)^2 - \mathrm{Var}(\xi_t)  \right] - \left[\frac{1}{T}\sum_{t=0}^T (\eta_t)^2 - \mathrm{Var}(  \eta_t)  \right].
	\end{split}
	\end{equation}
	By Proposition 2.7 in \cite{basu2015estimation}, 
	\begin{equation*}
	\mathbb{P}\left(  \left|\frac{1}{T}\sum_{t=0}^T (\xi_t)^2 - \mathrm{Var}(\xi_t)  \right| > 2\pi \mathcal{M}(f_X)\eta   \right) \geq  2\exp\left[ -cn\min(\eta,\eta^2)\right],
	\end{equation*}
	and
	\begin{equation*}
	\mathbb{P}\left(  \left|\frac{1}{T}\sum_{t=0}^T (\eta_t)^2 - \mathrm{Var}(\eta_t)  \right| > 2\pi \mathcal{M}(f_Z)\eta   \right) \geq 2\exp\left[ -cn\min(\eta,\eta^2)\right].
	\end{equation*}
	What remains to be  considered is the first term in~\eqref{eqn:HW0}, whose spectral density is given by
	\begin{equation*}
	f_{\xi+\eta}(\theta) = u'f_X(\theta) u + v'f_Z(\theta)z + u'f_{X,Z}(\theta) v + v'f^*_{X,Z}(\theta) u,
	\end{equation*}
	and its upper extreme satisfies
	\begin{equation*}
	\mathcal{M}(f_{\xi+\eta}) \leq \mathcal{M}(f_X) + \mathcal{M}(f_Z) + 2\mathcal{M}(f_{X,Z}).
	\end{equation*}
	Hence, we get:
	\small
	\begin{equation*}
	\mathbb{P} \left( \left| \frac{1}{T}\sum_{t=0}^T (\xi_t+\eta_t)^2 - \mathrm{Var}(\xi_t + \eta_t) \right| > 2\pi[\mathcal{M}(f_X) + \mathcal{M}(f_Z) + 2\mathcal{M}(f_{X,Z})]\eta \right) \geq 2\exp[-cn\min(\eta,\eta^2)]. 
	\end{equation*}
	\normalsize
	Combining all three terms yields the desired result. 
\end{proof}

\begin{lemma}\label{lemma:decomposition} Let $N$ and $M$ be matrices of the same dimension. Then, there exists a decomposition $M=M_1+M_2$, such that 
	\begin{enumerate}[(a)]
		\item $\text{rank}(M_1)\leq \text2{rank}(N)$; 
		\item $\llangle M_1, M_2 \rrangle =0$. 
	\end{enumerate}
\end{lemma}

\begin{proof}
	Let the SVD of $N$ be given by $N = U\Sigma V'$, where both $U$ and $V$ are orthogonal matrices and assume $\text{rank}(N) = r$. Define $\widetilde{M}$ and do partition it as follows:
	\begin{equation*}
	\widetilde{M} = U' M V = \begin{bmatrix}
	\widetilde{M}_{11} & \widetilde{M}_{12} \\ \widetilde{M}_{21} & \widetilde{M}_{22}
	\end{bmatrix}.
	\end{equation*}
	Next, let 
	\begin{equation*}
	M_1 = U \begin{bmatrix}
	\widetilde{M}_{11} & \widetilde{M}_{12} \\ \widetilde{M}_{21} & O
	\end{bmatrix}V' \qquad \text{and} \qquad M_2 = U\begin{bmatrix}
	O & O \\ O & \widetilde{M}_{22}
	\end{bmatrix}V',  \qquad \widetilde{M}_{11} \in \mathbb{R}^{r\times r}.
	\end{equation*}
	Then,  $M_1 + M_2 = M$ and 
	\begin{equation*}
	\text{rank}(M_1) \leq \text{rank}\left(U \begin{bmatrix}
	\widetilde{M}_{11} & \widetilde{M}_{12} \\ O & O
	\end{bmatrix}V' \right) + \text{rank}\left(U \begin{bmatrix}
	\widetilde{M}_{11} & O \\ \widetilde{M}_{21} & O
	\end{bmatrix}V' \right) \leq 2r.
	\end{equation*}
	Moreover,
	\begin{equation*}
	\llangle M_1, M_2 \rrangle = \text{tr} \left[ M_1M_2'\right] = 0.
	\end{equation*}
\end{proof}

\begin{lemma}\label{lemma:inequality1}
	Define the error matrix by $\Delta^B = \widehat{B} - B^\star$ and $\Delta^C = \widehat{C} - C^\star$, and let the weighted regularizer $\mathcal{Q}$ be defined as
	\begin{equation*}
	\mathcal{Q}(B,C) = \vertiii{B}_* + \frac{\lambda_C}{\lambda_B}\vertii{C}_1.
	\end{equation*}
	With the subspaces defined in (\ref{defn:basisspace}) and (\ref{defn:supportspace}), the following inequality holds:
	\begin{equation*}
	\mathcal{Q}(B^\star,C^\star) - \mathcal{Q}(\widehat{B},\widehat{C}) \leq \mathcal{Q}(\Delta^B_{\mathcal{S}_{B^\star}},\Delta^C_{\mathcal{J}_{C^\star}}) - \mathcal{Q}(\Delta^B_{\mathcal{S}^\perp_{B^\star}},\Delta^C_{\mathcal{J}^c_{C^\star}}). 
	\end{equation*}
\end{lemma}
\begin{proof}
	First, from definitions  (\ref{defn:basisspace}) and (\ref{defn:supportspace}), we know that $B^\star_{\mathcal{S}^\perp}=0$ and $C^\star_{\mathcal{J}^c_{C^\star}}=0$.  Using the definition of $\mathcal{Q}$, we obtain 
	\begin{equation*}
	\mathcal{Q}(B^\star,C^\star) = \vertiii{B^\star_\mathcal{S} + B^\star_{\mathcal{S}^\perp}}_* + \frac{\lambda_C}{\lambda_B} \vertii{C^\star_{\mathcal{J}_C^\star} + C^\star_{\mathcal{J}^c_{C^\star}} }_1 = \vertiii{B^\star_{\mathcal{S}}}_*  + \frac{\lambda_C}{\lambda_B} \vertii{ C^\star_{\mathcal{J}_C^\star}}_1,
	\end{equation*}
	and
	\small
	\begin{equation*}
	\begin{split}
	\mathcal{Q}(\widehat{B},\widehat{C}) & =  \mathcal{Q}(B^\star + \Delta^B, C^\star+\Delta^C) \\
	& =  \vertiii{B^\star_\mathcal{S} + \Delta^B_{\mathcal{S}^\perp_{B^\star}}+ \Delta^B_{\mathcal{S}_{B^\star}} +B^\star_{\mathcal{S}^\perp}  }_* + \frac{\lambda_C}{\lambda_B} \vertii{C^\star_{\mathcal{J}_C^\star} + \Delta^C_{\mathcal{J}_{C^\star}} + C^\star_{\mathcal{J}^c_{C^\star}} + \Delta^C_{\mathcal{J}^c_{C^\star}} }_1 \\
	& \geq \vertiii{B^\star_\mathcal{S} + \Delta^B_{\mathcal{S}^\perp_{B^\star}}}_* - \vertiii{\Delta^B_{\mathcal{S}_{B^\star}}}_* + \frac{\lambda_C}{\lambda_B} \left( \vertii{C^\star_{\mathcal{J}_C^\star} + \Delta^C_{\mathcal{J}_{C^\star}}}_1 + \vertii{\Delta^C_{\mathcal{J}^c_{C^\star}} }_1\right)  \\
	& \geq \vertiii{B^\star_\mathcal{S}}_* + \vertiii{\Delta^B_{\mathcal{S}^\perp_{B^\star}}}_* - \vertiii{\Delta^B_{\mathcal{S}_{B^\star}}}_* + \frac{\lambda_C}{\lambda_B} \left( \vertii{C^\star_{\mathcal{J}_C^\star}}_1 + \vertii{\Delta^C_{\mathcal{J}_{C^\star}}}_1 - \vertii{\Delta^C_{\mathcal{J}^c_{C^\star}} }_1\right).
	\end{split}
	\end{equation*}
	\normalsize
	The decomposition of the first term comes from the construction of $\Delta^B_{\mathcal{S}^\perp_{B^\star}}$. It then follows that
	\small
	\begin{equation*}
	\begin{split}
	\mathcal{Q}(B^\star,C^\star) - \mathcal{Q}(\widehat{B},\widehat{C}) & \leq \frac{\lambda_C}{\lambda_B} \vertii{ C^\star_{\mathcal{J}_C^\star}}_1 +  \vertiii{\Delta^B_{\mathcal{S}_{B^\star}}}_* - \vertiii{\Delta^B_{\mathcal{S}^\perp_{B^\star}}}_* + \frac{\lambda_C}{\lambda_B} \left(  \vertii{\Delta^C_{\mathcal{J}_{C^\star}}}_1 - \vertii{\Delta^C_{\mathcal{J}^c_{C^\star}} }_1 - \vertii{C^\star_{\mathcal{J}_C^\star}}_1 \right) \\
	& = \vertiii{\Delta^B_{\mathcal{S}_{B^\star}}}_* + \frac{\lambda_C}{\lambda_B} \vertii{\Delta^C_{\mathcal{J}_{C^\star}} }_1 - \left(\vertiii{\Delta^B_{\mathcal{S}^\perp_{B^\star}}}_* + \frac{\lambda_C}{\lambda_B} \vertii{\Delta^C_{\mathcal{J}^c_{C^\star}} }_1 \right) \\
	& = \mathcal{Q}(\Delta^B_{\mathcal{S}_{B^\star}},\Delta^C_{\mathcal{J}_{C^\star}}) - \mathcal{Q}(\Delta^B_{\mathcal{S}^\perp_{B^\star}},\Delta^C_{\mathcal{J}^c_{C^\star}}). 
	\end{split}
	\end{equation*}
\end{proof}
\normalsize
\begin{lemma}\label{lemma:RSC-result}
	Under the conditions of Theorem~\ref{thm:consistencyBC}, the following bound holds:
	\begin{equation*}
	\frac{1}{T}\vertiii{\mathfrak{W}_0(\Delta^B_{\text{aug}} + \Delta^C_{\text{aug}})}_F^2  \geq \frac{\alpha_{\text{RSC}}}{2} (\vertiii{\Delta^B}_F^2 + \vertiii{\Delta^C}_F^2) - \frac{\lambda_B}{2}\mathcal{Q}(\Delta^B,\Delta^C). 
	\end{equation*}
\end{lemma}
\begin{proof}
	This lemma directly follows from Lemma 2 in \citet{agarwal2012noisy}, by setting $\Theta^\star = B^\star$, $\Gamma^\star = C^\star$, with the regularizer $\mathcal{R}(\cdot)$ being the element-wise $\ell_1$ norm. Note that $\sigma_j(B^\star)=0$ for $j=r+1,\cdots,\min\{p_1,p_2\}$ since $\text{rank}(B)=r$. For our problem, it suffices to set $\mathbb{M}^\perp$ as $\mathcal{J}^c_{C^\star}$, and therefore $\|C^\star_{\mathcal{J}^c_{C^\star}}\|_1=0$.
\end{proof}

\begin{lemma}\label{lemma:boundoperator}
	Consider the two centered Gaussian processes $\{X_t\}\in\mathbb{R}^{p_1}$ and $\{Z_t\}\in\mathbb{R}^{p_2}$, and denote their cross covariance matrix by $\Gamma_{X,Z}(h) = (X_t, Z_{t+h})=\mathbb{E}(X_tZ'_{t+h})$. Let $\mathcal{X}$ and $\mathcal{Z}$ denote the data matrix. There exist positive constants $c_i>0$ such that whenever $T\succsim c_3(p_1+p_2)$, with probability at least
	\begin{equation*}
	1-c_1\exp[-c_2(p_1 + p_2)],
	\end{equation*}
	the following bound holds:
	\begin{equation*}
	\frac{1}{T}\vertiii{\mathcal{X}'\mathcal{Z}}_{\text{op}} \leq \mathbb{Q}_{X,Z}\sqrt{\frac{p_1+p_2}{T}} + 4\vertiii{\Gamma_{X,Z}(0)}_{\text{op}},
	\end{equation*}
	where 
	\begin{equation*}
	\mathbb{Q}_{X,Z} = c_0\left[\mathcal{M}(f_X) + \mathcal{M}(f_Z) + \mathcal{M}(f_{X,Z}) \right].
	\end{equation*}
\end{lemma}
\begin{proof}
	The main structure of this proof follows from that of Lemma~3 in \citet{negahban2011estimation}, and here we focus on how to handle the temporal dependency present in our problem. Let $S^p=\{u\in\mathbb{R}^p|\|u\|=1\}$ denote the $p$-dimensional unit sphere. The operator norm has the following variational representation form:
	\begin{equation*}
	\frac{1}{T}\vertiii{\mathcal{X}'\mathcal{Z}}_{\text{op}} = \frac{1}{n} \sup\limits_{u\in S^{p_1}}\sup\limits_{v\in S^{p_2}} u'\mathcal{X}'\mathcal{Z}v.
	\end{equation*} 
	For positive scalars $s_1$ and $s_2$, define
	\begin{equation*}
	\Psi(s_1,s_2) = \sup\limits_{u\in s_1S^{p_1}}\sup\limits_{v\in s_2S^{p_2}} \langle \mathcal{X} u,  \mathcal{Z}v\rangle,
	\end{equation*}
	and the goal is to establish an upper bound for $\Psi(1,1)/T$. Let $\mathcal{A}=\{u^1,\cdots,u^A\}$ and $\mathcal{B}=\{v^1,\cdots,v^B\}$ denote the $1/4$ coverings of $S^{p_1}$ and $S^{p_2}$, respectively. \citet{negahban2011estimation} showed that
	\begin{equation*}
	\Psi(1,1)\leq 4\max\limits_{u^a\in\mathcal{A},v^b\in\mathcal{B}}  \langle \mathcal{X} u^a,  \mathcal{Z}v^b\rangle, 
	\end{equation*} 
	and by \citet{anderson1998multivariate} and \citet{anderson2011statistical}, there exists a $1/4$ covering of $S^{p_1}$ and $S^{p_2}$ with at most $A\leq 8^{p_1}$ and $B\leq 8^{p_2}$ elements, respectively. Consequently, 
	\begin{equation*}
	\mathbb{P}\left[ \left|\frac{1}{T}\Psi(1,1)\right|\geq 4\delta  \right] \leq 8^{p_1+p_2} \max\limits_{u^a,v^b} \mathbb{P}\left[\frac{|(u^a)'\mathcal{X}\mathcal{Z}(v^b)|}{T} \geq \delta\right].
	\end{equation*}
	What remains to be bounded is 
	\begin{equation*}
	\frac{1}{T}u'\mathcal{X}'\mathcal{Z}v, \qquad \text{for an arbitrary fixed pair of }(u,v)\in S^{p_1}\times S^{p_2}.
	\end{equation*}
	By Lemma~\ref{lemma:DeviationBound0}, we have
	\small
	\begin{equation*}
	\mathbb{P}\left[ \left|u'\left(\frac{\mathcal{X}'Z}{T}\right)v\right| > 2\pi\left( \mathcal{M}(f_X) + \mathcal{M}(f_Z) + \mathcal{M}(f_{X,Z}) \right)\eta + \vertiii{\Gamma_{X,Z}(0)}_{\text{op}}  \right] \leq 6\exp\left(-cT\min\{\eta,\eta^2\}\right).
	\end{equation*}
	\normalsize
	Therefore, we have 
	\small
	\begin{equation*}
	\begin{split}
	\mathbb{P}\left[ \left|\frac{1}{T}\Psi(1,1)\right|\geq 8\pi \left( \mathcal{M}(f_X) + \mathcal{M}(f_Z) + \mathcal{M}(f_{X,Z}) \right)\eta + 4\vertiii{\Gamma_{X,Z}(0)}_{\text{op}}    \right] \leq 6\exp\left[(p_1+p_2)\log 8 - cT\min\{\eta,\eta^2\}\right].
	\end{split}
	\end{equation*}
	\normalsize
	With the specified choice of sample size $T$, the probability vanishes by choosing $\eta = c_0\sqrt{\frac{p_1+p_2}{T}}$, for $c_0$ large enough, and we yield the conclusion in Lemma~\ref{lemma:boundoperator}.
\end{proof}

\begin{lemma}\label{lemma:eigenbound}
	Consider some generic matrix $A\in\mathbb{R}^{m\times n}$ and let $\gamma=\{\gamma_1,\dots,\gamma_p\}~(p<n)$ denote the set of column indices of interest. Then, the following inequalities hold
	\begin{equation*}
	\Lambda_{\min}(A'A)\leq \Lambda_{\min}(A'_\gamma A_\gamma) \leq \Lambda_{\max}(A'_\gamma A_\gamma)\leq \Lambda_{\max}(A'A).
	\end{equation*}
\end{lemma}

\begin{proof}
	Let 
	\begin{equation*}
	\mathcal{V} := \{v=(v_1,\cdots,v_n)\in\mathbb{R}^n| v'v = 1\}.
	\end{equation*}
	and
	\begin{equation*}
	\mathcal{V}_\gamma := \{ v=(v_1,\cdots,v_n)\in\mathbb{R}^n| v'v=1 \text{ and }v_j=0~\forall j\notin \gamma \}.
	\end{equation*}
	It is obvious that $\mathcal{V}_\gamma\subseteq \mathcal{V}$. By the definition of eigenvalues through their {\em Rayleigh quotient} characterization, 
	\begin{equation*}
	\Lambda_{\min}(A'_\gamma A_\gamma) = \min\limits_{u'u=1,u\in\mathbb{R}^p} u' (A'_\gamma A_\gamma) u = \min\limits_{v'v=1,v\in \mathcal{V}_\gamma} v' (A'A)v \geq \min\limits_{v'v=1,v\in \mathcal{V}} v' (A'A)v = \Lambda_{\min}(A'A).
	\end{equation*}
	Similarly, 
	\begin{equation*}
	\Lambda_{\max}(A'_\gamma A_\gamma) = \max\limits_{u'u=1,u\in\mathbb{R}^p} u' (A'_\gamma A_\gamma) u = \max\limits_{v'v=1,v\in \mathcal{V}_\gamma} v' (A'A)v \leq \max \limits_{v'v=1,v\in \mathcal{V}} v' (A'A)v = \Lambda_{\max}(A'A).
	\end{equation*}
\end{proof}

\begin{lemma}\label{lemma:boundsubprocess}
	Let $\{X_t\}$ and $\{\varepsilon_t\}$ be two generic processes, where $\varepsilon_t=(U_t',V_t')'$. Suppose the spectral density of the joint process $(X_t',\varepsilon_t')$ exists. Then, the following inequalities hold 
	\begin{equation*}
	\mathfrak{m}(f_{X,V}) \geq \mathfrak{m}(f_{X,\varepsilon}), \qquad \mathcal{M}(f_{X,V}) \leq \mathcal{M}(f_{X,\varepsilon}).
	\end{equation*}
\end{lemma}
\begin{proof}
	By definition, the spectral density $f_{X,\varepsilon}(\theta)$ can be written as
	\begin{equation*}
	\begin{split}
	f_{X,\varepsilon}(\theta) & = \left(\frac{1}{2\pi}\right) \sum_{\ell=-\infty}^{\infty} \Gamma_{X,\varepsilon}(\ell) e^{-i\ell \theta}, \qquad \theta \in[-\pi,\pi] \\
	& = \left(\frac{1}{2\pi}\right) \sum_{\ell=-\infty}^{\infty} \left( \mathbb{E}X_t U_{t+\ell},~~ \mathbb{E}X_t V_{t+\ell} \right)e^{-i\ell \theta} \\
	& = (f_{X,U}(\theta),~~ f_{X,V}(\theta)).
	\end{split}
	\end{equation*}
	It follows that 
	\begin{equation*}
	\mathcal{M}(f_{X,\varepsilon}) = \mathop{\text{ess sup}}\limits_{\theta\in [-\pi,\pi]} \sqrt{\Lambda_{\max} (H(\theta)) },
	\end{equation*}
	where 
	\begin{equation*}
	H(\theta) = \begin{bmatrix}
	f^*_{X,U}(\theta) \\ f^*_{X,V}(\theta)
	\end{bmatrix} \begin{bmatrix}
	f_{X,U}(\theta) & f_{X,V}(\theta)
	\end{bmatrix}= \begin{bmatrix}
	f^*_{X,U}(\theta) f_{X,U}(\theta) & f^*_{X,U}(\theta) f_{X,V}(\theta) \\ f^*_{X,V}(\theta) f_{X,U}(\theta) & f^*_{X,V}(\theta) f_{X,V}(\theta)
	\end{bmatrix}.
	\end{equation*}
	Note that 
	\begin{equation*}
	\mathcal{M}(f_{X,V}) = \mathop{\text{ess sup}}\limits_{\theta\in [-\pi,\pi]} \sqrt{\Lambda_{\max} (f^*_{X,V}(\theta) f_{X,V}(\theta))}.
	\end{equation*}
	By Lemma~\ref{lemma:eigenbound}, $\forall \theta$, $\Lambda_{\min}(f^*_{X,V}(\theta) f_{X,V}(\theta))\geq \Lambda_{\min}(H(\theta))$
	and $\Lambda_{\max} (f^*_{X,V}(\theta) f_{X,V}(\theta))\leq \Lambda_{\max}(H(\theta))$, hence
	\begin{equation*}
	\mathfrak{m}(f_{X,V}) \geq \mathfrak{m}(f_{X,\varepsilon}), \qquad \mathcal{M}(f_{X,V}) \leq \mathcal{M}(f_{X,\varepsilon}).
	\end{equation*}
\end{proof}

\section{Testing group Granger-causality under a sparse alternative.}\label{appendix:sparse-testing}

In this section, we develop a testing procedure to test the null hypothesis against its sparse alternatives, that is, 
\begin{equation*}
H_0: B=0 \qquad \text{vs} \qquad H_A:B\text{ is nonzero and sparse}.
\end{equation*}
Throughout, we impose assumptions on the sparsity level of $B$ (to be specified later), and use the {\em higher criticism} framework \citep[c.f.][]{tukey1989higher,donoho2004higher,arias2011global} as the building block of the testing procedure.  

Once again, we start with testing sparse alternatives in a simpler model setting
\begin{equation*}
Y_t = \Pi X_t + \epsilon_t,
\end{equation*}
where $Y_t\in\mathbb{R}^{p_2}$, $X_t\in\mathbb{R}^{p_1}$, and $\epsilon_t\in\mathbb{R}^{p_2}$ with each component being independent 
and identically distributed (i.i.d) and also independent of $X_t$. We would like to test the null hypothesis $H_0:\Pi=0$. Written in a compact form, the model is given by
\begin{equation}\label{eqn:YXE}
\mathcal{Y} = \mathcal{X}\Pi' + \mathcal{E},
\end{equation}
where $\mathcal{Y}\in\mathbb{R}^{T\times p_2}$ , $\mathcal{X}\in\mathbb{R}^{T\times p_1}$, and $\mathcal{E}$ are both contemporaneously and temporally independent. The latter shares similarities to the setting in \citet{arias2011global}, with the main difference being that here we have a multi-response $\mathcal{Y}$. By rewriting~\eqref{eqn:YXE} using Kronecker products, we have
\begin{equation*}
\vect{\mathcal{Y}} = (I_{p_2}\otimes \mathcal{X})\vect{\Pi'} + \vect{\mathcal{E}} \qquad \text{i.e.,} \qquad \mathbf{Y} = \mathbf{X}\mathbf{\Pi} + \mathbf{E},
\end{equation*}
where $\mathbf{Y}=\vect{\mathcal{Y}}\in\mathbb{R}^{Tp_2},\mathbf{X}= I_{p_2}\otimes \mathcal{X}\in\mathbb{R}^{Tp_2\times p_1p_2}$ and $\mathbf{\Pi}=\vect{\Pi'}\in\mathbb{R}^{p_1p_2}$. Each coordinate in $\mathbf{E}$ is iid. In this form, using the higher criticism \citep{donoho2004higher,ingster2010detection,arias2011global} with proper scaling, the test statistic is given~by:
\begin{equation}\label{eqn:HCraw}
\HC^*(\mathbf{X},\mathbf{Y}) = \sup\limits_{t>0} H(t,\mathbf{X},\mathbf{Y}) :=  \sqrt{\frac{p_1p_2}{2\bar{\Phi}(t)(1-2\bar{\Phi}(t))}}\Big[\tfrac{1}{p_1p_2}\sum_{k=1}^{p_1p_2} \mathbf{1}\left(\tfrac{\sqrt{T}\cdot\mathbf{X}'_k\mathbf{Y}}{\|\mathbf{X}_k\|_2\|\mathbf{Y}\|_2} > t \right)-2\bar{\Phi}(t)\Big],
\end{equation}
where $\mathbf{X}_k$ is the $k^{\text{th}}$ column of $\mathbf{X}$ and $\bar{\Phi}(t)=1-\Phi(t)$ with $\Phi(t)$ being the cumulative distribution
function of a standard Normal random variable. Intuitively, $$\big(\tfrac{1}{p_1p_2}\big)\sum_{k=1}^{p_1p_2} \mathbf{1}\big\{\sqrt{T}\mathbf{X}'_k\mathbf{Y}/(\|\mathbf{X}_k\|_2\|\mathbf{Y}\|_2) > t \big\}$$ is the fraction of significance beyond a given level $t$, after scaling for the vector length and the noise level. To conduct a level $\alpha$ test, $H_0$ is rejected when $\HC^*(\mathbf{X},\mathbf{Y}) > h(p_1p_2,\alpha_{p_1p_2})$ where $h(p_1p_2,\alpha_{p_1p_2})\approx \sqrt{2\log\log(p_1p_2)}$, provided that $\alpha_{p_1p_2}\rightarrow 0$ slowly enough in the sense that $h(p_1p_2,\alpha_{p_1p_2}) = 2\sqrt{\log\log (p_1p_2)} (1+ o(1))$ \citep[see][]{donoho2004higher}. The effectiveness of the test relies on a number of assumptions on the design matrix and the sparse vector to be tested. Next, we introduce the three most relevant definitions for subsequent developments, originally mentioned in \citet{arias2011global}. 
\begin{definition}[Bayes risk]
	Following \citet{arias2011global}, the Bayes risk of a test $\mathcal{T}$ for testing $\mathbf{\Pi}=0$ vs. $\mathbf{\Pi}\sim \pi$, when $H_0$ and $H_1$ occur with the same probability, is defined as the sum of type I error and its average probability of type II error; i.e., 
	\begin{equation*}
	\text{Risk}_\pi(\mathcal{T}) = \mathbb{P}_0(\mathcal{T}=1)  + \pi[\mathbb{P}_\mathbf{\Pi}(\mathcal{T}=0)],
	\end{equation*}
	where $\pi$ is a prior on the set of alternatives $\Omega$. When no prior is specified, the risk is defined as the worst-case risk:
	\begin{equation*}
	\text{Risk}(\mathcal{T}) = \mathbb{P}_0(\mathcal{T}=1)  + \max\limits_{\mathbf{\Pi}\in\Omega}[\mathbb{P}_\mathbf{\Pi}(\mathcal{T}=0)].
	\end{equation*}
\end{definition}

\begin{definition}[Asymptotically powerful] We use $\mathcal{T}_{n,p}$ to denote the dependency of the test on the sample size $n$ and the parameter dimension $p$. With $p\rightarrow \infty$ and $n=n(p)\rightarrow\infty$, a sequence of tests $\{\mathcal{T}_{n,p}\}$ is said to be {\em asymptotically powerful} if
	\begin{equation*}
	\lim_{p\rightarrow\infty}\text{Risk}(\mathcal{T}_{n,p})=0.
	\end{equation*}	
\end{definition}

\begin{definition}[Weakly correlated] Let $\mathcal{S}_p(\gamma,\Delta)$ denote the set of $p\times p$ correlation matrices $C=[c_{jk}]$ satisfying the weakly correlated assumption: for all $j=1,\dots,p$,
	\begin{equation*}
	|c_{jk}|<1-(\log p)^{-1} \qquad \text{and} \qquad \{k: |c_{jk}|>\gamma\}\leq \Delta, \qquad \text{for some }\gamma\leq 1, \Delta \geq 1. 
	\end{equation*} 
\end{definition}
With the above definitions, \citet{arias2011global} establishes that using the test based on higher criticism is asymptotically powerful, provided that (1) $\mathbf{\Pi}$ satisfies the {\em strong sparsity assumption}, that is, the total number of nonzeros $s_{\mathbf{\Pi}}^\star=(p_1p_2)^{\theta}$ with $\theta\in(1/2,1)$; (2) the correlation matrix of $\mathbf{X}$ belongs to $\mathcal{S}(\gamma,\Delta)$ with $\gamma$ and $\Delta$ satisfying certain assumptions in terms of their relative order with respect to parameter dimension and sample size; and (3) the minimum magnitude of the nonzero elements of $\mathbf{\Pi}$ exceeds a certain lower detection threshold. 
\medbreak

Switching to our model setting in which
\begin{equation*}
Z_t = BX_{t-1} + CZ_{t-1} + V_t, \qquad B\in\mathbb{R}^{p_2\times p_1},
\end{equation*}
where $B$ encodes the dependency between $Z_t$ and $X_{t-1}$, conditional on $Z_{t-1}$, the above discussion suggests that we can use higher criticism on the residuals $\mathcal{R}_1$ and $\mathcal{R}_0$, where $\mathcal{R}_1$ and $\mathcal{R}_0$ are identically defined to those in the low-rank testing; that is, $\mathcal{R}_1$ is the residual after regressing $\mathcal{X}$ on $\mathcal{Z}$, and $\mathcal{R}_0$ is the residual after regressing $\mathcal{Z}^T$ on $\mathcal{Z}$:
\begin{equation*}
\mathcal{R}_1 = (I-P_z)\mathcal{X} \qquad \text{and} \qquad \mathcal{R}_0 = (I-P_z)\mathcal{Z}^T,
\end{equation*}
where $P_z = \mathcal{Z}(\mathcal{Z}'\mathcal{Z})^{-1}\mathcal{Z}'$. Writing the model in terms of $\mathcal{R}_1$ and $\mathcal{R}_0$, we have
\begin{equation*}
\mathcal{R}_0 = \mathcal{R}_1B' + \mathcal{V}, \qquad \text{i.e.,} \qquad \mathbf{R}_0 = \mathbf{R}_1\beta_B + \mathbf{V},
\end{equation*}
where $\mathbf{R}_0=\vect{\mathcal{R}_0}, \mathbf{R}_1=I\otimes \mathcal{R}_1$, $\mathbf{V}=\vect{\mathcal{V}}$, and $\beta_B=\vect{B'}\in\mathbb{R}^{p_1p_2}$. To test $H_0:\beta_B=0$, the higher criticism is given by
\small
\begin{equation}\label{eqn:HC}
\begin{split}
\HC^*(\mathbf{R}_1,\mathbf{R}_0) & = \sup\limits_{t>0}H(t,\mathbf{R}_1,\mathbf{R}_0) : = 
\sqrt{\frac{p_1p_2}{2\bar{\Phi}(t)(1-2\bar{\Phi}(t))}} \Big[\tfrac{1}{p_1p_2}\sum_{j=1}^{p_2}\sum_{i=1}^{p_1}\mathbf{1}\Big(\tfrac{\sqrt{T} |\mathbf{R}_{1k}'\mathbf{R}_0|}{\|\mathbf{R}_{1k}\|_2\|\mathbf{R}_0\|_2} >t \Big) - 2\bar{\Phi}(t) \Big] \\
& = \sup\limits_{t>0}\sqrt{\frac{p_1p_2}{2\bar{\Phi}(t)(1-2\bar{\Phi}(t))}} \Big[\tfrac{1}{p_1p_2}\sum_{j=1}^{p_2}\sum_{i=1}^{p_1}\mathbf{1}\Big(\tfrac{\sqrt{T} |S_{10,ij}|}{\sqrt{S_{11,ii}S_{00,jj}}} >t \Big) - 2\bar{\Phi}(t) \Big]
\end{split}
\end{equation}
\normalsize
where $S_{10}=\mathcal{R}_1'\mathcal{R}_0/T$, $S_{11}=\mathcal{R}_1'\mathcal{R}_1/T$ and $S_{00}=\mathcal{R}_0'\mathcal{R}_0/T$. The second equality is due to the block-diagonal structure of $\mathbf{R}_1$. We reject the null hypothesis if 
\begin{equation*}
\HC^*(\mathbf{R}_1,\mathbf{R}_0) >  2\sqrt{\log\log(p_1p_2)}.
\end{equation*}
Empirically, $t$ can be chosen from $\{[1,\sqrt{5\log(p_1p_2)}]\cap \mathbb{N}\}$ \citep{arias2011global}. 
\medbreak

Next, we analyze the theoretical properties of the above testing procedure. If the parameter dimension is fixed, then classical consistency results in terms of convergence (in probability or almost surely) hold when letting $T\rightarrow\infty$, and everything follows trivially, as long as the corresponding population quantities satisfy the posited assumptions. 

In the remainder, we allow the parameter dimension $p_1p_2$ to slowly vary with the sample size $T$. 
Let $S_{\mathbf{R}_1} = \mathbf{R}_1'\mathbf{R}_1/n$ be the sample covariance matrix based on the residuals $\mathbf{R}_1$, and let $C_{\mathbf{R}_1}$ be the corresponding correlation matrix. The following proposition directly follows from Theorem~4 in \citet{arias2011global}. 
\begin{proposition}[An asymptotically powerful test]\label{prop:asymptoticallypowerful}
	Under the following conditions, the testing procedure associated with the Higher Criticism statistics defined in~\eqref{eqn:HC} is asymptotically powerful, provided that the smallest magnitude of nonzero entries of $B^\star$ exceeds the lower detection boundary.\footnote{For a thorough discussion on the lower detection boundary, we refer the reader to \citet{ingster2010detection,arias2011global} and references therein.}
	\begin{enumerate}[(a)]
		\item Strong sparsity: let $p_B=p_1p_2$ be the dimension of $\beta_B^\star$, then the total number of nonzeros satisfies $s^\star_{B} = p_B^{\theta}$, where $\theta\in(1/2,1)$. 
		\item Weakly correlated design: $C_{\mathbf{R}_1}\in\mathcal{S}(\gamma,\Delta)$ with the parameters satisfying $\Delta = O(p_1^{\epsilon})$, $\gamma=O(p_1^{-1/2+\epsilon})$, $\forall~\epsilon>0$.  
	\end{enumerate}	
\end{proposition}

Note that $S_{\mathbf{R}_1} = I\otimes S_{11}$, where $S_{11}=\mathcal{R}_1'\mathcal{R}_1/T$; hence, $C_{\mathbf{R}_1}= I\otimes C_{11}$, with $C_{11}$ being the sample correlation matrix based on $\mathcal{R}_1$. The weakly correlated design assumption is thus effectively imposed on $C_{11}$, with the parameters $\gamma$ and $\Delta$ satisfying the same condition. The weakly correlated design assumption on $C_{11}$ in Proposition~\ref{prop:asymptoticallypowerful} is for a deterministic realization of $\mathcal{R}_1$. The following corollary states that for a random realization of $\mathcal{R}_1$, obtained by regressing a random $\mathcal{X}$ on $\mathcal{Z}$, to satisfy the weakly correlated design assumption with high probability, it is sufficient that the population counterparts of the associated quantities satisfy the required assumptions.   

\begin{corollary}\label{cor:Tlogp1p2}
	Consider residual $\mathcal{R}_1$ obtained by regressing a random realization of $\mathcal{X}$ on that of $\mathcal{Z}$. Let $\Sigma_{11}:=\Gamma_X - \Gamma_{X,Z}\Gamma_Z^{-1}\Gamma_{X,Z}'$ be the covariance of $X_t$ conditional on $Z_t$, and $\rho_{11}$ be the corresponding correlation matrix. Suppose $\rho_{11}\in S(\gamma,\Delta)$ with $\gamma$ and $\Delta$ satisfying the same condition as in Proposition~\ref{prop:asymptoticallypowerful}. Then with high probability, the sample correlation matrix based on $\mathcal{R}_1$ belongs to $\mathcal{S}(\gamma',\Delta')$, where $\gamma'$ and $\Delta'$ respectively satisfy the same condition as $\gamma$ and $\Delta$, provided that the same condition imposed on $\gamma$ holds for $\sqrt{T^{-1}\log(p_1p_2)}$. Moreover, the conclusion in Proposition~\ref{prop:asymptoticallypowerful} holds. 
\end{corollary}

\begin{remark}
	In the work of \citet{anderson2002canonical} and \citet{arias2011global}, the authors focus their analysis primarily on the multiple regression setting, where the regression coefficient matrix directly encodes the relationship between the response variable and the covariates, in an iid data setting. We consider a more complicated model setting in which the regression coefficient matrix of interest encodes the partial auto-correlations between a multivariate response and a set of exogenous variables, while the data exhibit temporal dependence. It is worth pointing out that with the presence of temporal dependence, the rate with respect to the model dimension $p$ and sample size $T$ stays the same, as in the case where the data are iid \citep[e.g.,][]{rudelson2013hanson,basu2015estimation}; specifically, it is $\sqrt{\log p/T}$ in terms of the element-wise infinity norm, whereas the associated constant is a function of the lower and upper extremes of the spectral density, which intricately controls the exact coverage and power of the testing procedures. Therefore, as long as the rate constraint on $p$ and $T$ is satisfied (as in Corollary~\ref{cor:Tlogp1p2}), the main conclusion is compatible with previous work, and asymptotically, we either obtain the distribution of the test statistic (low rank testing), or have a powerful test (sparse testing). 
\end{remark}

\begin{remark}
	To solve the global testing problem for the sparse setting, a possible alternative is to construct a test statistic based on estimates of the regression coefficients, then perform a global or max test on the estimated coefficients. A key issue for such a test is that the estimated entries of $B$ are biased due to the use of Lasso; therefore, a debiasing procedure \citep[e.g.][]{javanmard2014confidence} would be required to obtain valid marginal distributions for the entries of the $B$ matrix. In contrast, the higher criticism test statistic is based on the correlation between the response and the covariates (see Equation~\eqref{eqn:HCraw}), and here we employ the idea on the residuals so that the effect of $Z_t$ block is removed. We do not directly deal with the estimates of the $B$ matrix and thus avoid the complications induced by the potentially biased estimates of $B$.  
\end{remark}

\bigbreak
\section{Estimation and Consistency for an Alternative Model Specification.}\label{appendix:sparseB}

In this section, we consider the finite-sample error bound for the case where both $B$ and $C$ are sparse. We assume the presence of a sparse contemporaneous conditional dependence, hence the alternate between the estimation of transition matrices and that of the covariance matrix is required. In what follows, we briefly outline the estimation procedure and the error bounds of the estimates. All notations follow from those in Section~\ref{sec:theory}.

The joint optimization problem is given by
\begin{align}\label{eqn:optB-sparse-C-sparse}
(\widehat{B},\widehat{C},\widehat{\Omega}_v) = \argmin\limits_{B,C,\Omega_v} &\Big\{ \text{tr} \big[ \Omega_v (\mathcal{Z}^T - \mathcal{X}B' - \mathcal{Z}C')'(\mathcal{Z}^T - \mathcal{X}B' - \mathcal{Z}C')/T \big] - \log\det \Omega_v  \nonumber \\
&+ \lambda_B \vertii{B}_1 + \lambda_C \vertii{C}_1 + \rho_v\vertii{\Omega_v}_{1,\text{off}} \Big\}.
\end{align}
For every iteration, with a fixed $\widehat{\Omega}^{(k)}_v$, $\widehat{B}^{(k+1)}$ and $\widehat{C}^{(k+1)}$ are both updated via Lasso; for fixed $(\widehat{B}^{(k)},\widehat{C}^{(k)})$, $\widehat{\Omega}^{(k)}_v$ is updated by the graphical Lasso. 
\begin{corollary}\label{cor:B-sparse-C-sparse}
	Consider the stable Gaussian VAR system defined in~\eqref{eqn:model0} in which $B^\star$ is assumed  to be low rank with rank $r_B^\star$ and $C^\star$ is assumed to be $s_C^\star$-sparse. Further, assume the following 
	\begin{itemize}
		\item[C.1] The incoherence condition holds for $\Omega_v^\star$.
		\item[C.2] $\Omega_v^\star$ is diagonally dominant. 
		\item[C.3] The maximum node degree of $\Omega_v^\star$ satisfies $d_{\Omega_v^\star}^{\max} = o(p_2)$.   
	\end{itemize}
	Then, for random realizations of $\{X_t\}$, $\{Z_t\}$ and $\{V_t\}$, and the sequence $\{(\widehat{B}^{(k)},\widehat{C}^{(k)}),\widehat{\Omega}_v^{(k)}\}_k$ returned by Algorithm~\ref{algo:2} outlined in Section~\ref{sec:estimation}, with high probability, 
	the following bounds hold for all iterations $k$ for sufficiently large sample size $T$:
	\begin{equation*}
	\smalliii{\widehat{B}^{(k)}-B^\star}_F+\smalliii{\widehat{C}^{(k)}-C^\star}_F = O\Big(\sqrt{\tfrac{\max\{s_B^\star,s_C^\star\}\big(\log(p_1+p_2)+\log p_2\big)}{T}}\Big),
	\end{equation*}
	and
	\begin{equation*}
	\smallii{\widehat{\Omega}_v^{(k)} - \Omega_v^\star}_F = O\Big(\sqrt{\tfrac{(s^\star_{\Omega_v}+p_2)(\log(p_1+p_2) + \log p_2)}{T}}\Big).
	\end{equation*}
\end{corollary}

Note that when no contemporaneous dependence is present, $(\widehat{B},\widehat{C})$ solves
\begin{equation}\label{optBC-sparse-sparse}
(\widehat{B},\widehat{C}) = \argmin\limits_{(B,C)}\Big\{\tfrac{1}{T}\vertiii{\mathcal{Z}^T-\mathfrak{W}_0(B_{\text{aug}} + C_{\text{aug}})}_F^2 + \lambda_B\|B\|_1 + \lambda_C\|C\|_1 \Big\},
\end{equation}
whose error bound is given by 
\begin{equation}\label{eqn:bound-sparse-sparse-I}
\smalliii{\widehat{B}-B^\star}_F^2 + \smalliii{\widehat{C}-C^\star}_F^2 \leq 4(\lambda_B^2 + \lambda_C^2)/\alpha^2_{\text{RSC}}, 
\end{equation}
provided that the RSC condition holds and the regularization parameters are chosen properly. By setting the weighted regularizer as $\mathcal{Q}(B,C) = \|B\|_1+\tfrac{\lambda_C}{\lambda_B}\|C\|_1$ and $\Delta^B:=\widehat{B}-B^\star$ can be decomposed as (see equation~\eqref{defn:supportspace})
\begin{equation*}
\|\Delta^B_{\mathcal{J}_B} + \Delta^B_{\mathcal{J}_B^c}\|_1 = \|\Delta^B_{\mathcal{J}_B}\|_1 + \|\Delta^B_{\mathcal{J}_B^c}\|_1.
\end{equation*}
The rest of the proof is similar to that of Theorem~\ref{thm:consistencyBC} hence is omitted here.

\section{Proof of Propositions and Corollaries.}

%
%
\begin{proof}[\textbf{Proof of Proposition~\ref{prop:testing-rank}}]
	The joint process $W_t=\{(X_t',Z_t')'\}$ is a stationary $\VAR(1)$ process, and it follows that 
	\begin{equation*}
	S_w(h): = \begin{bmatrix}
	S_x(h) & S_{x,z}(h) \\ S_{z,x}(h) & S_z(h)
	\end{bmatrix}= \frac{1}{T}\sum_{t=1}^T w_tw_{t+h}' \stackrel{p}{\to} \Gamma_W(h):= \E W_t W_{t+h}' , \quad \text{as}~~T\rightarrow\infty,
	\end{equation*}
	which implies 
	\begin{equation*}
	S_x \stackrel{p}{\to} \Gamma_X, \quad S_z\stackrel{p}{\to} \Gamma_Z, \quad S_{x,z} \stackrel{p}{\to}\Gamma_{X,Z}, \quad S_{x,z}(1)\stackrel{p}{\to} \Gamma_{X,Z}(1).
	\end{equation*}
	Note that sample partial regression residual covariances can be obtained by
	\begin{equation*}
	S_{00} = S_z - S_z(1)S_z^{-1}S_z'(1), \qquad 
	S_{11} = S_x - S_{x,z}S_z^{-1}S_{x,z}', \qquad 
	S_{10} = S_{x,z}(1) - S_z(1)S_z^{-1}S_{x,z}'.
	\end{equation*}
	An application of the Continuous Mapping Theorem yields
	\begin{equation*}
	S_{00}\stackrel{p}{\to} \Sigma_{00}, \qquad S_{10}\stackrel{p}{\to}\Sigma_{10}, \qquad S_{11}\stackrel{p}{\to} \Sigma_{11}.
	\end{equation*}
	By \citet{hsu1941limiting,hsu1941problem}, the limiting behavior of $T\Psi_r$ is given by
	\begin{equation*}
	T\Psi_r \sim \chi^2_{(p_1-r)(p_2-r)}, \qquad \text{as}\quad T\rightarrow\infty.
	\end{equation*}
	Note that since $\mu$ is of multiplicity one and the ordered eigenvalues are continuous functions of the matrices, the following holds:
	\begin{equation*} 
	\phi_k \stackrel{p}{\to} \mu_k, \qquad \forall\, k=1,\dots,\min (p_1,p_2).
	\end{equation*}
\end{proof}

%
%
\begin{proof}[\textbf{Proof of Corollary~\ref{cor:Tlogp1p2}}]
	First, we note that $\mathcal{R}_1$ effectively comes from the following stochastic regression:
	\begin{equation}\label{eqn:XZR}
	X_t = QZ_t + R_t, \qquad \text{for some regression matrix $Q$},
	\end{equation}
	with $\mathcal{R}_1=\mathcal{X}-\mathcal{Z}\widehat{Q}$ being the sample residual. The population covariance of $R_t$ is the {\em conditional covariance} of $X_t$ on $Z_t$, given by
	\begin{equation*}
	\Sigma_{11} = \Sigma^\star_{R} := \Gamma_X - \Gamma_{X,Z}\Gamma_Z^{-1}\Gamma'_{X,Z}.
	\end{equation*}
	$\Sigma_{11}$ is identical to that defined in equation~\eqref{eqn:Sigma01}. Writing the model in terms of data, we~have
	\begin{equation*}
	\mathcal{X} = \mathcal{Z}Q + \mathcal{R}^\star,
	\end{equation*}
	where we use $\mathcal{R}^\star$ to denote the true error term, for the purpose of distinguishing it from the residuals by regressing $\mathcal{X}$ on $\mathcal{Z}$. Note that $\mathcal{R}^\star$ is also sub-Gaussian. First, we would like to obtain a bound for $\|S_{11} - \Sigma_{11}\|_\infty$. Let $S_{R}$ be the sample covariance matrix based on the actual errors, i.e., $S_{R^\star} = (\mathcal{R}^\star)'(\mathcal{R}^\star)/T$, then 
	\begin{equation*}
	\|S_{11} - \Sigma_{11}\|_\infty \leq \| S_{R^\star} - \Sigma_{11}\|_\infty + \|S_{11} - S_{R^\star}\|_\infty
	\end{equation*}
	The first term can be directed bounded by Lemma 1 in \citet{ravikumar2011high}, that is, there exists some constant $\tau>2$, such that for large enough sample size $T$, with probability at least $1-1/p_2^{\tau-2}$,
	\begin{equation*}
	\| S_{R^\star} - \Sigma_{11}\|_\infty \leq C_0\sqrt{\log p_2/T}, \qquad \text{for some constant }C_0>0.
	\end{equation*} 
	Consider the second term. Rewrite it as
	\begin{equation*}
	S_{11} - S_{R^\star} = \tfrac{2}{T}(\mathcal{R}^\star)'\mathcal{Z}(Q^\star - \widehat{Q}) + (Q^\star - \widehat{Q})'\Big(\tfrac{\mathcal{Z}'\mathcal{Z}}{T}\Big)(Q^\star - \widehat{Q}):= I_1 + I_2,
	\end{equation*}
	then for $I_1$,
	\begin{equation*}
	I_1 \leq 2 \vertiii{Q^\star - \widehat{Q}}_1 \vertii{\frac{1}{T} (\mathcal{R}^\star)'\mathcal{Z}}_\infty \leq 2\vertii{\text{vec}(Q^\star) - \text{vec}(\widehat{Q})}_1 \vertii{\frac{1}{T} (\mathcal{R}^\star)'\mathcal{Z}}_\infty.
	\end{equation*}
	By Lemma~\ref{lemma:operator-infinity}, there exist constants $c_i>0$ such that with probability at least $1-c_1\exp(-c_2\log (p_1p_2))$, for sufficiently large sample size $T$, we get
	\begin{equation}\label{eqn:RZ}
	\vertii{\frac{1}{T} (\mathcal{R}^\star)'\mathcal{Z}}_\infty \leq C_1\sqrt{\log (p_1p_2)/T},\qquad \text{for some constant }C_1>0.
	\end{equation}
	For $I_2$, we have that
	\begin{equation*}
	I_2 \leq \vertiii{Q^\star - \widehat{Q}}^2_1\vertii{\frac{\mathcal{Z}'\mathcal{Z}}{T}}_\infty \leq \vertii{\text{vec}(Q^\star) - \text{vec}(\widehat{Q})}_1^2\vertii{\frac{\mathcal{Z}'\mathcal{Z}}{T}}_\infty.
	\end{equation*}
	By Proposition 2.4 in \citet{basu2015estimation} and taking the union bound, there exist some constants $c_1'$ and $c_2'$ such that with probability at least $1-c_1'\exp(-c_2'\log p_2)$, for sufficiently large sample size $T$, we obtain
	\begin{equation*}
	\vertii{\frac{\mathcal{Z}'\mathcal{Z}}{T} - \Gamma_Z}_\infty \leq C_2\sqrt{\log p_2/T}, \qquad \text{for some constant }C_2>0,
	\end{equation*}
	which implies 
	\begin{equation}\label{eqn:ZZ}
	\vertii{\frac{\mathcal{Z}'\mathcal{Z}}{T}}_\infty \leq C_2\sqrt{\log p_2/n} + \max_i(\Gamma_{Z,ii}).
	\end{equation}
	By assuming that $\vertii{\text{vec}(Q^\star) - \text{vec}(\widehat{Q})}_1\leq \varepsilon_Q$, it follows that 
	\begin{equation*}
	I_1 + I_2 \leq C_1'\varepsilon_Q \sqrt{\frac{\log(p_1p_2)}{T}} + C'_2\varepsilon_Q^2 \sqrt{\frac{\log p_2}{T}},
	\end{equation*}
	hence 
	\begin{equation}\label{eqn:boundresidual}
	\|S_{11} - \Sigma_{11}\|_\infty \leq C_0\sqrt{\frac{\log p_2}{T}} + C_1'\varepsilon_Q \sqrt{\frac{\log(p_1p_2)}{T}} + C'_2\varepsilon_Q^2 \sqrt{\frac{\log p_2}{T}}.
	\end{equation}
	Regardless of the relative order of $p_1$ and $p_2$, one can easily verify that 
	\begin{equation}
	\|S_{11} - \Sigma_{11}\|_\infty  = O\Big(\sqrt{\tfrac{\log (p_1p_2)}{T}} \Big).
	\end{equation}
	by assuming $\log(p_1p_2)/T$ being a small quantity. Since
	\begin{equation*}
	C_{11} = \left(\text{diag}(S_{11})\right)^{-1/2} S_{11} \left(\text{diag}(S_{11})\right)^{-1/2}
	\end{equation*}
	and letting $\widetilde{R}_t = \text{diag}(\Sigma_{11})^{-1/2}R_t$, we then have that $C_{11}$ is simply the sample covariance matrix based on residual surrogates of $\widetilde{R}_t$, whose error rate stays unchanged by scaling, i.e, $\vertii{C_{11}-\rho_{11}}_\infty=O(\sqrt{T^{-1}\log(p_1p_2)})$. The latter fact further implies that if $\rho_{11}\in\mathcal{S}(\gamma,\Delta)$, then $C_{11}\in\mathcal{S}(\gamma',\Delta')$ with $\Delta'\geq \Delta - {\text{(const)}}\sqrt{\log(p_1p_2)/T}$ and $\gamma'\geq \gamma + {\text{(const)}}\sqrt{\log(p_1p_2)/T}$.
	
	It then follows that as long as $\sqrt{T^{-1}\log(p_1p_2)}$ satisfies the same condition imposed on $\gamma = O(p_1^{-1/2+\epsilon})$, that is, 
	\begin{equation*}
	p_1^{1-2\epsilon}\log(p_1p_2) = O(T), \qquad \text{for all }\epsilon>0, 
	\end{equation*}
	with high probability, the sample covariance matrix based on the residuals $\mathcal{R}_1$ satisfies the weakly correlated design assumption, for a random realization $\mathcal{X}$ and $\mathcal{Z}$. 
\end{proof}

Finally, we briefly outline the main steps of how to obtain the error bound in Corollary~\ref{cor:B-sparse-C-sparse}. 
\begin{proof}[\textbf{Proof of Corollary~\ref{cor:B-sparse-C-sparse}}]
	Let $W_t=(X_t',Z_t')'$.	At iteration 0, $(\widehat{B}^{(0)},\widehat{C}^{(0)})$ solves~\eqref{optBC-sparse-sparse}, and the following bound holds:
	\begin{equation*}
	\smalliii{\widehat{B}^{(0)}-B^\star}_F^2 + \smalliii{\widehat{C}^{(0)}-C^\star}_F^2 \leq 4(\lambda_B^2 + \lambda_C^2)/\alpha^2_{\text{RSC}}, 
	\end{equation*}
	provided that $\mathfrak{W}$ satisfies the RSC condition, and $\lambda_B,\lambda_C$ both satisfy
	\begin{equation*}
	\lambda_B \geq 4\|\mathcal{W}'\mathcal{V}/T\|_\infty, \qquad \lambda_C \geq 4\|\mathcal{W}'\mathcal{V}/T\|_\infty.
	\end{equation*}
	In particular, by Lemma~\ref{lemma:RSCW} and Lemma~\ref{lemma:operator-infinity}, for random realizations of $\{X_t\}$, $\{Z_t\}$ and $\{V_t\}$, for sufficiently large sample size, with high probability
	\begin{equation*}
	\mathfrak{W} \text{ satisfies the RSC condition},
	\end{equation*}
	and
	\begin{equation*}
	\vertii{\mathcal{W}'\mathcal{V}/T}_\infty \leq C_1\sqrt{\frac{\log(p_1+p_2)+\log p_2}T},
	\end{equation*}
	for some constant $C_1$. Hence, with high probability, 
	\begin{equation}
	\smalliii{\widehat{B}^{(0)} - B^\star}_F^2 + \smalliii{\widehat{C}^{(0)} - C^\star}_F^2 =  O\Big(\tfrac{\log(p_1+p_2)+\log p_2}T\Big).
	\end{equation}
	For $\widehat{\Omega}_v^{(0)}$, it solves a graphical Lasso problem:
	\begin{equation*}
	\widehat{\Omega}_v^{(0)} = \argmin\limits_{\Omega_v\in\mathbb{S}_{++}^{p_2\times p_2}} \Big\{\log\det\Omega_v - \tr\big(\widehat{S}_u^{(0)}\Omega_v \big) + \rho_v \|\Omega_v\|_{1,\text{off}}\Big\},
	\end{equation*}
	where $\widehat{S}_v^{(0)} = \tfrac{1}T(\mathcal{Z}^T - \mathcal{X}\widehat{B}^{(0)'} - \mathcal{Z}\widehat{C}^{(0)'})'(\mathcal{Z}^T - \mathcal{X}\widehat{B}^{(0)'} - \mathcal{Z}\widehat{C}^{(0)'})$.
	Similar to the proof of Theorem~\ref{thm:algo2}, its error bound depends on $\|\widehat{S}_v^{(0)}-\Sigma_v^\star\|_\infty$. With the same decomposition and consider only the leading term, 
	\begin{equation*}
	\|\widehat{S}_v^{(0)}-\Sigma_v^\star\|_\infty = O \Big(\sqrt{\tfrac{\log (p_1+p_2) + \log p_2}{T}}\Big), \qquad \Rightarrow \quad \|\widehat{\Omega}_v^{(0)}-\Omega_v^\star\|_\infty =  O \Big(\sqrt{\tfrac{\log (p_1+p_2) + \log p_2}{T}}\Big).
	\end{equation*}
	At iteration 1, the bound of $\smalliii{\widehat{B}^{(1)}-B^\star}_F^2 + \smalliii{\widehat{C}^{(1)}-C^\star}^2_F$ relies on 
	\begin{equation}
	\begin{split}
	\smallii{\frac{1}{T}\mathcal{W}'\mathcal{V}\widehat{\Omega}_v^{(0)}}_\infty&\leq \smallii{\tfrac{1}T\mathcal{W}'\mathcal{V}(\widehat{\Omega}_v^{(0)}-\Omega_v^\star)}_\infty + \smallii{\tfrac{1}T\mathcal{W}'\mathcal{V}\Omega_v^\star}_\infty, \\
	&\leq C_2\sqrt{\tfrac{\log(p_1+p_2)+\log p_2}T} + d^{\Omega_v^\star}_{\max}\smallii{\tfrac{1}T\mathcal{W}'\mathcal{V}}_\infty \smallii{\widehat{\Omega}_v^{(0)}-\Omega_v^\star}_\infty\\
	&= O\Big(\sqrt{\tfrac{\log(p_1+p_2)+\log p_2}T}\Big), 
	\end{split}
	\end{equation}
	hence $\smalliii{\widehat{B}^{(1)}-B^\star}_F^2 + \smalliii{\widehat{C}^{(1)}-C^\star}_F^2 = O(\tfrac{\log(p_1+p_2)+\log p_2}T)$, which coincides with the bound of the estimator of iteration 0, implying the error rate remains unchanged henceforth. Up to this step, all sources of randomness have been captured. Consequently, the following bounds hold with high probability for all iterations $k$:
	\begin{equation*}
	\|\mathcal{W}'\mathcal{V}\widehat{\Omega}_v^{(k)}/T\|_\infty = O\Big(\sqrt{\tfrac{\log(p_1+p_2)+\log p_2}{T}}\Big),
	\end{equation*}
	and
	\begin{equation*}
	\| \widehat{S}_v^{(k)} - \Sigma_v^\star\|_\infty= O\Big(\sqrt{\tfrac{\log (p_1+p_2)+\log p_2}{T}}\Big),
	\end{equation*}
	which imply the bounds in Corollary~\ref{cor:B-sparse-C-sparse}.
\end{proof}

\section{List of Stock and Macroeconomic Variables}\label{appendix:Dictionary}
\begin{table}[H]
	\centering
	\caption{List of Stocks used in the Analysis}
		\begin{tabular}{l|l|l|l}
			\specialrule{.1em}{.05em}{.05em} 
			Stock Symbol	& 	Name	& 	Stock Symbol	& 	Company Name	\\ 
			\specialrule{.1em}{.05em}{.05em}
			AAPL	& 	Apple Inc.	& 	JNJ	& 	Johnson \& Johnson Inc	\\
			AIG	& 	American International Group Inc.	& 	JPM	& 	JP Morgan Chase \& Co	\\
			ALL	& 	Allstate Corp.	& 	KO	& 	The Coca-Cola Company	\\
			AMGN	& 	Amgen Inc.	& 	LMT	& 	Lockheed-Martin	\\
			AXP	& 	American Express Inc.	& 	LOW	& 	Lowe's	\\
			BA	& 	Boeing Co.	& 	MCD	& 	McDonald's Corp	\\
			BAC	& 	Bank of America Corp	& 	MDLZ	& 	Mondel International	\\
			BK	& 	Bank of New York Mellon Corp	& 	MDT	& 	Medtronic Inc.	\\
			BMY	& 	Bristol-Myers Squibb	& 	MMM	& 	3M Company	\\
			C	& 	Citigroup Inc	& 	MO	& 	Altria Group	\\
			CAT	& 	Caterpillar Inc	& 	MRK	& 	Merck \& Co.	\\
			CL	& 	Colgate-Palmolive Co.	& 	MS	& 	Morgan Stanley	\\
			CMCSA	& 	Comcast Corporation	& 	MSFT	& 	Microsoft	\\
			COF	& 	Capital One Financial Corp.	& 	NSC	& 	Norfolk Southern Corp	\\
			COP	& 	ConocoPhillips	& 	ORCL	& 	Oracle Corporation	\\
			CSCO	& 	Cisco Systems	& 	OXY	& 	Occidental Petroleum Corp.	\\
			CVS	& 	CVS Caremark	& 	PEP	& 	Pepsico Inc.	\\
			CVX	& 	Chevron	& 	PFE	& 	Pfizer Inc	\\
			DD	& 	DuPont	& 	PG	& 	Procter \& Gamble Co	\\
			DIS	& 	The Walt Disney Company	& 	RTN	& 	Raytheon Company	\\
			DOW	& 	Dow Chemical	& 	SLB	& 	Schlumberger	\\
			DVN	& 	Devon Energy Corp	& 	SO	& 	Southern Company	\\
			EMC	& 	EMC Corporation	& 	T	& 	AT\&T Inc	\\
			EXC	& 	Exelon	& 	TGT	& 	Target Corp.	\\
			F	& 	Ford Motor	& 	TWX	& 	Time Warner Inc.	\\
			FCX	& 	Freeport-McMoran	& 	TXN	& 	Texas Instruments	\\
			FDX	& 	FedEx	& 	UNH	& 	UnitedHealth Group Inc.	\\
			GD	& 	General Dynamics	& 	UPS	& 	United Parcel Service Inc	\\
			GE	& 	General Electric Co.	& 	USB	& 	US Bancorp	\\
			GILD	& 	Gilead Sciences	& 	UTX	& 	United Technologies Corp	\\
			GS	& 	Goldman Sachs	& 	VZ	& 	Verizon Communications Inc	\\
			HAL	& 	Halliburton	& 	WBA	& 	Walgreens Boots Alliance	\\
			HD	& 	Home Depot	& 	WFC	& 	Wells Fargo	\\
			HON	& 	Honeywell	& 	WMT	& 	Wal-Mart	\\
			IBM	& 	International Business Machines	& 	XOM	& 	Exxon Mobil Corp	\\
			INTC	& 	Intel Corporation	& 		& 		\\			
			\specialrule{.1em}{.05em}{.05em} 
		\end{tabular}
\end{table}		

\begin{table}[H]
\centering
\caption{List of Macroeconomic Variables and the transformation used in the Analysis}
		\begin{tabular}{l|l|l}
			\specialrule{.1em}{.05em}{.05em} 
			Symbol & Description &  Transformation\\ \hline
			FFR & Federal Funds Rate & abs diff \\
			T10yr & 10-Year Treasury Yield with Constant Maturity & abs diff\\
			UNEMPL & Unemployment Rate for 16 and above & abs diff \\
			IPI & Industrial Production Index & relative diff \\
			ETTL & Employment Total & relative diff \\
			M1 & M1 Money Stock & relative diff\\
			AHES & Average Hourly Earnings of Production and Nonsupervisory Employees &  relative diff \\
			CU & Capital Utilization & relative diff \\
			M2 & M2 Money Stock & relative diff \\
			HS & Housing starts & relative diff \\
			EX & US Exchange Rate & abs diff \\
			PCEQI & Personal Consumption Expenditures Quantity Index & relative diff \\
			GDP & real Gross Domestic Product & relative diff \\
			PCEPI & Personal Consumption Expenditures Price Index & relative diff \\
			PPI & Producer Price Index & relative diff \\
			CPI & Consumer Price Index & relative diff \\
			SP.IND & S\&P Industrial Sector index & relative diff  \\
			\specialrule{.1em}{.05em}{.05em} 
		\end{tabular}\vspace*{3mm}\\
		*abs diff: $x_t-x_{t-1}$, relative diff: $\tfrac{x_t-x_{t-1}}{x_{t-1}}$
\end{table}

\clearpage
\vskip 0.2in
\bibliographystyle{plainnat}
\nocite{*}
\bibliography{refbib}

\end{document}